\documentclass[a4paper,10pt]{llncs}

\usepackage{epsfig}
\usepackage{times}
\usepackage{graphicx}
\usepackage{amsmath}
\usepackage{amssymb}
\usepackage{gensymb}
\usepackage[english]{babel}
\usepackage{multirow}

\usepackage[T1]{fontenc}

\newcounter{prop}
\renewenvironment{proof}
{{\bf Proof:}}{\hspace*{\fill}$\Box$\par\vspace{2mm}}

\newcommand{\rephrase}[3]{\noindent\textbf{#1 #2}.~\emph{#3}}

\newcommand{\T}{\mbox{$\mathcal T$ }}
\renewcommand{\P}{\mbox{$\mathcal P$ }}
\newcommand{\remove}[1]{}

\newcommand{\eat}[1] {{}}

\begin{document}

%%% Title of contribution
\title{On a Tree and a Path with no Geometric Simultaneous Embedding}

\author{P. Angelini $^\dag$, M. Geyer $^\ddag$,  M. Kaufmann $^\ddag$ and D. Neuwirth $^\ddag$}

%%%% List of authors for the TOC
\tocauthor{Patrizio Angelini, Markus Geyer, Michael Kaufmann, and Daniel Neuwirth}

\institute{%
		$^\dag$Dipartimento di Informatica e Automazione -- Universit\`a  Roma Tre, Italy\\
        \email{angelini@dia.uniroma3.it}\\
    $^\ddag$Wilhelm-Schickard-Institut f\"{u}r Informatik -- Universit\"{a}t T\"{u}bingen, Germany\\
    \email{geyer/mk/neuwirth@informatik.uni-tuebingen.de}
}

\maketitle              % typeset the title of the contribution

%%%%%%%%%%%%%%%%%%%%%%%%%%%%%%%%%%%%%%%%%%%%%%%%%%%%%%%%%%%%%%%%%%%

%%%%%%%%%%%%%%%%%%%%%%%%%%%%%%%%%%%%%%%%%%%%%%%%%%%%%%%%%%%%%%%%%%%

\begin{abstract}
Two graphs $G_1=(V,E_1)$ and $G_2=(V,E_2)$ admit a \emph{geometric simultaneous embedding} if there exists a set of points $P$ and a bijection $M: P \rightarrow V$ that induce planar straight-line embeddings both for $G_1$ and for $G_2$. While it is known that two caterpillars always admit a geometric simultaneous embedding and that two trees not always admit one, the question about a tree and a path is still open and is often regarded as the most prominent open problem in this area. We answer this question in the negative by providing a counterexample. Additionally, since the counterexample uses disjoint edge sets for the two graphs, we also negatively answer another open question, that is, whether it is possible to simultaneously embed two edge-disjoint trees. As a final result, we study the same problem when some constraints on the tree are imposed. Namely, we show that a tree of depth $2$ and a path always admit a geometric simultaneous embedding. In fact, such a strong constraint is not so far from closing the gap with the instances not admitting any solution, as the tree used in our counterexample has depth $4$.
\end{abstract}

%%%%%%%%%%%%%%%%%%%%%%%%%%%%%%%%%%%%%%%%%%%%%%%%%%%%%%%%%%%%%%%%%%%
%%%%%%%%%%%%%%%%%%%%%%%%%%%%%%%%%%%%%%%%%%%%%%%%%%%%%%%%%%%%%%%%%%%

\section{Introduction}

Embedding planar graphs is a well-established field in graph theory and algorithms with a great variety of applications. Keystones in this field are the works of Thomassen~\cite{t-eg-94}, of Tutte~\cite{t-hdg-63}, and of Pach and Wenger~\cite{pw-epgfvl-01}, dealing with planar and convex representations of graphs in the plane.

Since recently, motivated by the need of contemporarily represent several different relationships among the same set of elements, a major focus in the research lies on \emph{simultaneous graph embedding}. In this setting, given a set of graphs with the same vertex-set, the goal is to find a set of points in the plane and a mapping between these points and the vertices of the graphs such that placing each vertex on the point it is mapped to yields a planar embedding for each of the graphs, if they are displayed separately. Problems of this kind frequently arise when dealing with the visualization of evolving networks and with the visualization of huge and complex relationships, as in the case of the graph of the Web.

Among the many variants of this problem, the most important and natural one is the \emph{geometric simultaneous embedding}. Given two graphs $G_1=(V,E')$ and $G_2=(V,E'')$, the task is to find a set of points $P$ and a bijection $M: P \rightarrow V$ that induce planar straight-line embeddings for both $G_1$ and $G_2$.

In the seminal paper on this topic~\cite{J-bcdeeiklm-spge-07}, Brass \emph{et al.} proved that geometric simultaneous embeddings of pairs of paths, pairs of cycles, and pairs of caterpillars always exist. A \emph{caterpillar} is a tree such that deleting all its leaves yields a path. On the other hand, many negative results have been shown. Brass \emph{et al.}~\cite{J-bcdeeiklm-spge-07} presented a pair of outerplanar graphs not admitting any simultaneous embedding and provided negative results for three paths, as well. Erten and Kobourov~\cite{ek-sepgfb-04} found a planar graph and a path not allowing any simultaneous embedding. Geyer \emph{et al.}~\cite{gkv-ttsids-09} proved that there exist two trees that do not admit any geometric simultaneous embedding. However, the two trees used in the counterexample have common edges, and so the problem is still open for edge-disjoint trees.

The most important open problem in this area is the question whether a tree and a path always admit a geometric simultaneous embedding or not. In this paper we answer this question in the negative.

Many variants of the problem, where some constraints are relaxed, have been studied in the literature. If the edges do not need to be straight-line segments, a famous result of Pach and Wenger~\cite{pw-epgfvl-01} shows that any number of planar graphs admit a simultaneous embedding, since it states that any planar graph can be planarly embedded on any given set of points in the plane. However, the same result does not hold if the edges that are shared by two graphs have to be represented by the same Jordan curve. In this setting the problem is called {\it simultaneous embedding with fixed edges}~\cite{f-egsfe-06,gjpss-sgefe-06,fjks-crppgasefe-08}.

The research on this problem opened a new exciting field of problems and techniques, like ULP trees and graphs \cite{efk-culpt-06,fk-culpg-07,fk-mlnpt-07}, colored simultaneous embedding~\cite{g-csge-07}, near-simultaneous embedding \cite{fkk-csnse-07}, and matched drawings \cite{gdkls-mdpg-07}, deeply
related to the general fundamental question of point-set embeddability.

In this paper we study the geometric simultaneous embedding problem of a tree and a path. We answer the question in the negative by providing a counterexample, that is, a tree and a path not admitting any geometric simultaneous embedding. Moreover, since the tree and the path used in our counterexample do not share any edge, we also negatively answer the question on two edge-disjoint trees.

The main idea behind our counterexample is to use the path to enforce a part of the tree to be in a certain configuration which cannot be drawn planar. Namely, we make use of level nonplanar trees~\cite{efk-culpt-06,fk-mlnpt-07}, that is, trees not admitting any planar embedding if their vertices have to be placed inside certain regions according to a particular leveling. The tree of the counterexample contains many copies of such trees, while the path is used to create the regions. To prove that at least one copy has to be in the particular leveling that determines a crossing, we need a quite huge number of vertices. However, such a huge number is often needed just to ensure the existence of particular structures playing a role in our proof. A much smaller counterexample could likely be constructed with the same techniques, but we decided to prefer the simplicity of the argumentations rather than the search for the minimum size.

The paper is organized as follows. In Sect.~\ref{se:preliminaries} we give preliminary definitions and we introduce the concept of level nonplanar trees. In Sect.~\ref{se:tree-path} we describe the tree \T and the path \P used in the counterexample. In Sect.~\ref{se:overview} we give an overview of the proof that \T and \P do not admit any geometric simultaneous embedding, while in Sect.~\ref{se:proofs} we give the details of such a proof. In Sect.~\ref{se:depth2} we present an algorithm for the simultaneous embedding of a tree of depth $2$ and a path, and in Sect.~\ref{se:conclusions} we make some final remarks.

%%%%%%%%%%%%%%%%%%%%%%%%%%%%%%%%%%%%%%%%%%%%%%%%%%%%%%%%%%%%%%%%%%
%%%%%%%%%%%%%%%%%%%%%%%%%%%%%%%%%%%%%%%%%%%%%%%%%%%%%%%%%%%%%%%%%%

\section{Preliminaries}\label{se:preliminaries}

A (undirected) \emph{$k$-level tree} $T=(V,E,\phi)$ on $n$ vertices is
a tree $T'=(V,E)$, called the \emph{underlying tree} of $T$, together with a
leveling of its vertices given by a function $\phi: V \mapsto \{1, \ldots, k\}$, such that for every edge $(u,v) \in E$, it holds $\phi(u)\not= \phi(v)$ (See~\cite{efk-culpt-06,fk-mlnpt-07}). A drawing of $T=(V,E,\phi)$ is a
\emph{level drawing} if each vertex $v \in V$ such that $\phi(v)=i$ is
placed on a horizontal line $l_i =\{(x,i) \mid x \in {\mathbb R}\}$. A
level drawing of $T$ is \emph{planar} if no two edges intersect
except, possibly, at common end-points. A tree $T=(V,E,\phi)$ is \emph{level nonplanar} if it does not admit any planar level drawing.

We extend this concept to the one of \emph{region-level drawing} by enforcing the vertices of each level to lie inside a certain region rather than on a horizontal line. Let $l_1,\ldots,l_k$ be $k$ pairwise non-crossing straight lines and let $r_1,\ldots,r_{k+1}$ be the regions of the plane such that any straight-line segment connecting a point in $r_i$ and a point in $r_h$, with $1 \leq i < h \leq k+1$, cuts all and only the lines $l_i,l_{i+1},\ldots, l_{h-1}$, in this order. A drawing of a $k$-level tree $T=(V,E,\phi)$ is called \emph{region-level drawing} if each vertex $v \in V$ such that $\phi(v)=i$ is placed inside region $r_i$. A region-level drawing of $T$ is \emph{planar} if no two edges intersect except, possibly, at common end-points. A tree $T=(V,E,\phi)$ is \emph{region-level nonplanar} if it does not admit any planar region-level drawing.

\begin{figure}[tb]
\begin{center}
\begin{tabular}{c c c}
\mbox{\includegraphics[height=2.7cm]{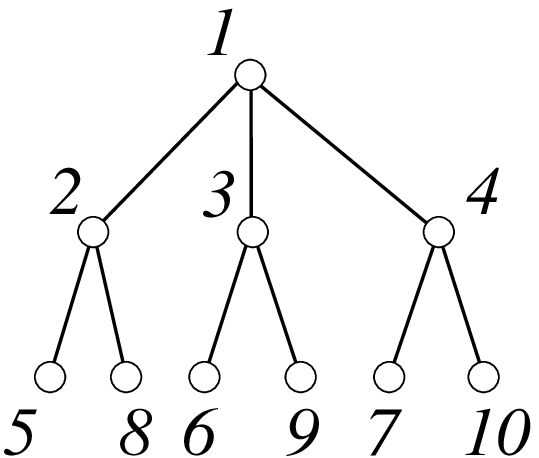}} \hspace{0.5cm} &
\mbox{\includegraphics[height=3.3cm]{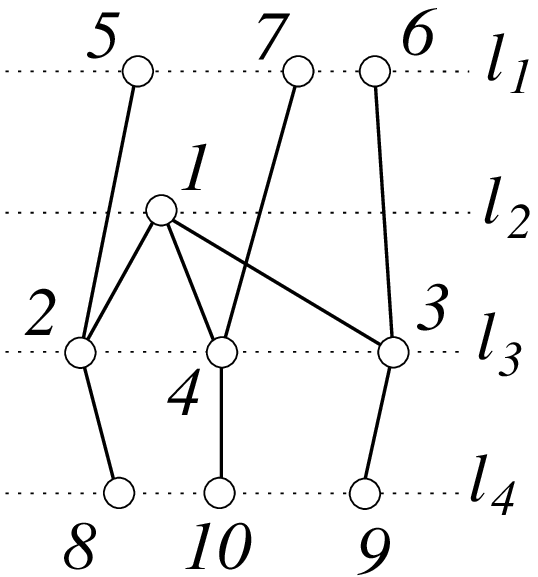}} \hspace{0.5cm} &
\mbox{\includegraphics[height=3.3cm]{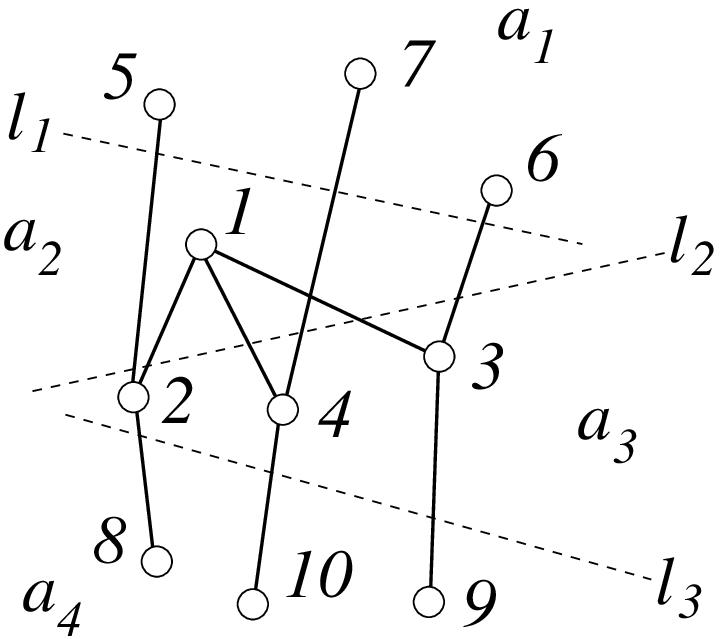}} \hspace{0.5cm} \\
(a) \hspace{0.5cm} & (b)  \hspace{0.5cm} & (c)
\end{tabular}
\caption{(a) A tree $T_u$. (b) A level nonplanar tree $T$ whose
underlying tree is $T_u$. (c) A region-level nonplanar tree $T$ whose
underlying tree is $T_u$.}\label{fig:T_level_nonplanar}
\end{center}
\end{figure}

The $4$-level tree $T$ whose underlying tree is shown in
Fig.~\ref{fig:T_level_nonplanar}(a) has been shown to be level
nonplanar~\cite{fk-mlnpt-07} (see
Fig.~\ref{fig:T_level_nonplanar}(b)). In the next lemma we show that $T$ is also region-level nonplanar (see Fig.~\ref{fig:T_level_nonplanar}(c)).

\begin{lemma}\label{lem:uap-tree}
The $4$-level tree $T$ whose underlying tree is shown in Fig.~\ref{fig:T_level_nonplanar}(a) is region-level nonplanar.
\end{lemma}

\begin{proof}
Refer to Fig.~\ref{fig:T_level_nonplanar}(c). First observe that, in any possible region-level planar drawing of $T$, the paths $p_1=v_5,v_2,v_8$ and $p_2=v_6,v_3,v_9$ define a polygon $Q_2$ (a polygon $Q_3$) inside region $r_2$ (region $r_3$). We have that $v_1$ is inside $Q_2$, as otherwise one of edges $(v_1,v_2)$ or $(v_1,v_3)$ would cross one of $p_1$ or $p_2$. Hence, vertex $v_4$ has to be inside $Q_3$, as otherwise edge $(v_1,v_4)$ would cross one of $p_1$ or $p_2$. However, in this case, there is no placement for vertices $v_7$ and $v_{10}$ that avoids a crossing between one of edges $(v_4,v_7)$ or $(v_4,v_{10})$ and one of the already drawn edges.
\end{proof}

Lemma~\ref{lem:uap-tree} will be vital for proving that there exist a tree \T and a path \P not admitting any geometric simultaneous embedding. In fact, \T contains many copies of the underlying tree of $T$, while \P connects vertices of \T in such a way to create the regions satisfying the above conditions and to enforce at least one of such copies to lie inside these regions according to the leveling making it nonplanar.

\section{The Counterexample}\label{se:tree-path}

In this section we describe a tree $\T$ and a path $\P$ not admitting any geometric simultaneous embedding.

\subsection{Tree \T}
The tree \T contains a root $r$ and $q$ \remove{:= {535\choose 7}} vertices $j_1, \ldots, j_q$ at distance $1$ from $r$, called {\it joints}. Each joint $j_h$, with $h=1,\dots,q$, is connected to $x$ copies $B_1,\ldots,B_x$ of a subtree, called {\it branch}, and to $l:=(s-1)^4 \cdot 3^2 \cdot x$ vertices of degree 1, called {\it stabilizers}. See Fig.~\ref{fig:complete_tree}(a).
Each branch $B_i$ consists of a root $r_i$, $(s-1)\cdot 3$ vertices of degree $(s-1)$ adjacent to $r_i$, and $(s-2) \cdot (s-1)\cdot 3$ leaves at distance $2$ from $r_i$. Vertices belonging to a branch $B_i$ are called $B$-\emph{vertices} and denoted by $1-$, $2-$, or $3-$vertices, according to their distance from their joint. Fig.~\ref{fig:complete_tree}(b) displays $1-$, $2-$, and $3-$vertices of a branch $B_i$.

\begin{figure}[tb]
\begin{center}
\begin{tabular}{c c}
\mbox{\includegraphics[height=3.6cm]{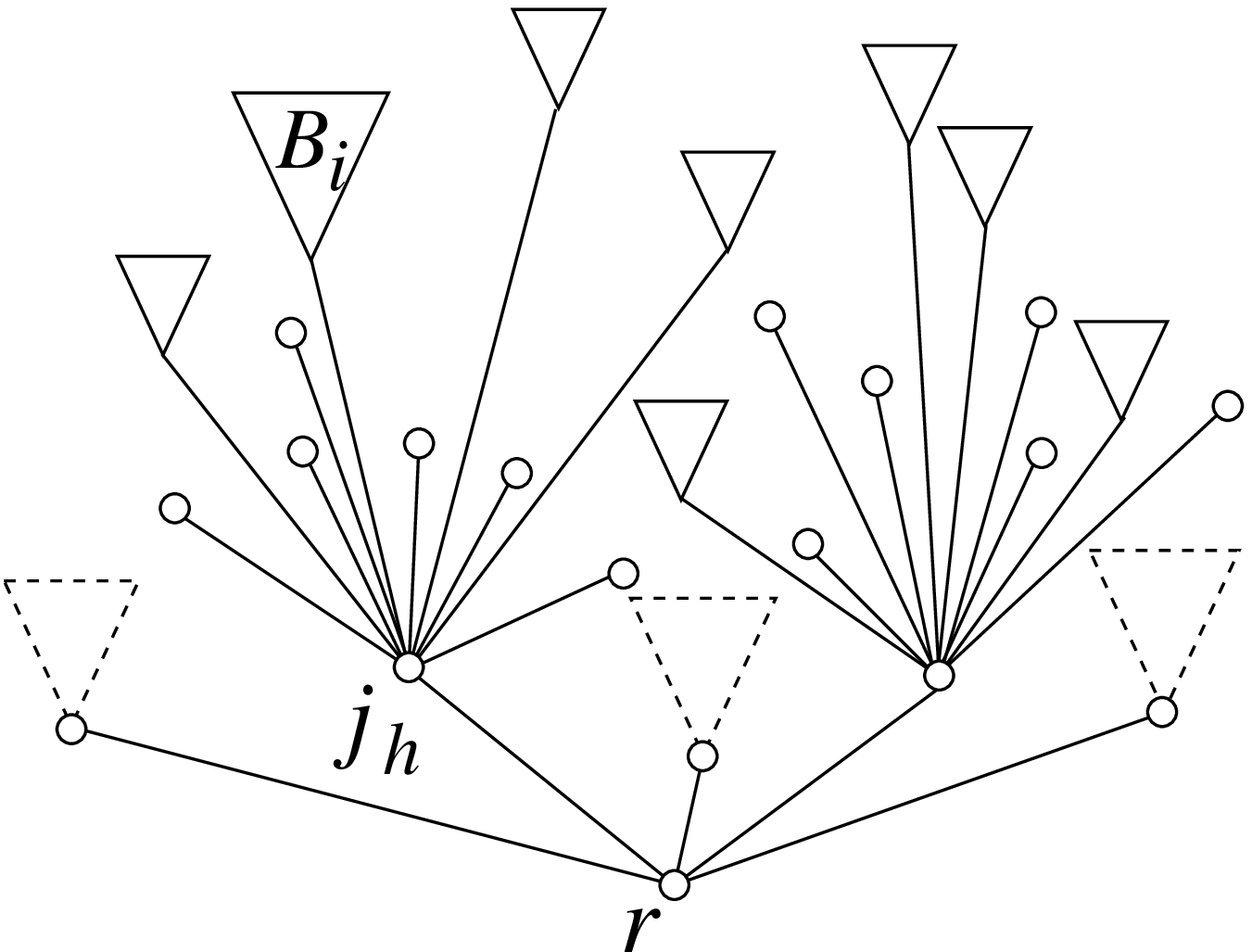}} \hspace{1.5cm} &
\mbox{\includegraphics[height=3.3cm]{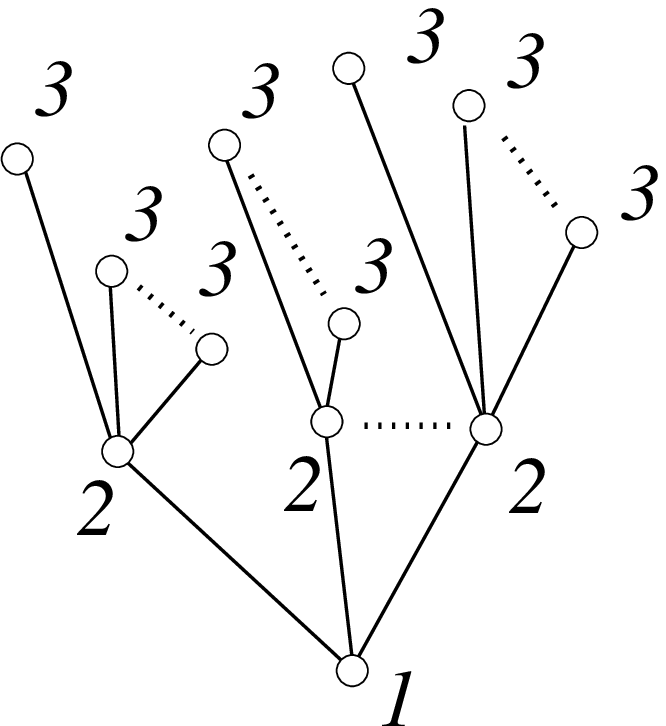}} \\
(a) \hspace{1.5cm} & (b)
\end{tabular}
\caption{(a) A schematization of the complete tree \T. Joints and stabilizers are small circles, branches are solid triangles, while complete subtrees connected to a joint are dashed triangles. (b) A schematization of a branch $B_i$.}\label{fig:complete_tree}
\end{center}
\end{figure}

Because of the huge number of vertices, in the rest of the paper, for the sake of readability, we use variables $n$, $s$, and $x$ as parameters describing the size of certain configurations. Such parameters will be given a value when the technical details of the argumentations are described. At this stage we just claim that a total number $n \leq {2^7\cdot 3\cdot x + 2 \choose 3} $ of vertices (see Lemmata~\ref{lemma:2_channels} and~\ref{lemma:k_cluster_passage}) suffices for the counterexample.

As a first observation we note that, despite the oversized number of vertices, tree \T has limited \emph{depth}, that is, every vertex is at distance from the root at most $4$. This leads to the following property.

\begin{property}\label{prop:three_bends}
Any path of tree edges starting at the root has at most $3$ bends.
\end{property}

\subsection{Path \P}
Path \P is given by describing some basic and recurring subpaths on the vertices of \T and how such subpaths are connected to each other. The idea is to partition the set of branches $B_i$ adjacent to each joint $j_h$ into subsets of $s$ branches each and to connect their vertices with path edges, according to some features of the tree structure, so defining the first building block, called {\it cell}.
Then, cells belonging to different branches are connected to each other, hence creating structures, called \emph{formations}, for which we can ensure certain properties regarding the intersection between tree and path edges. Further, different formations are connected to each other by path edges in such a way to create bigger structures, called \emph{extended formations}, which are, in their turn, connected to create a \emph{sequence of extended formations}.

All of these structures are constructed in such a way that there exists a set of cells such that any four of its cells, connected to the same joint and being part of the same formation or extended formation, contain a region-level nonplanar tree for any possible leveling, where the levels correspond to cells. Hence, proving that four of such cells lie in different regions satisfying the properties of separation described above is equivalent to proving the existence of a crossing in the tree. This allows us to consider only the bigger structures instead of dealing with single copies of the region-level nonplanar tree.

In the following we define such structures more formally and state their properties.

{\bf Cell:}
The most basic structure defined by \P is defined by looking at how it connects vertices of some branches $B_i$ connected to the same joint $j_h$ of \T. Consider a set of $s$ branches $B_i$, $i=1,\ldots,s$, connected to $j_h$. Assume the vertices of a level inside each tree to be arbitrarily ordered. For each $r=1,\ldots, s$, define a \emph{cell} $c_r(h)$ to be composed of its \emph{head}, its \emph{tail}, and a number $t$ of
\remove{ $64827 \cdot (441 + 64386) = 4202539929$} stabilizers of $j_h$.

\begin{figure}[tb]
  \centering{
\begin{tabular}{c}
\includegraphics[height=3cm]{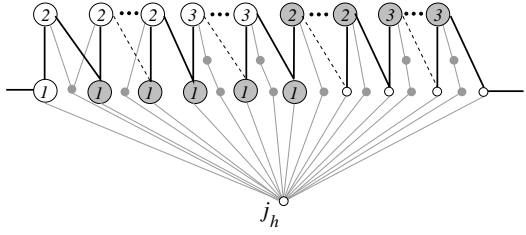}
\end{tabular}}
\caption{A cell. $B$-vertices of the head are depicted by large white circles,  $B$-vertices of the tail are large grey circles, $B$-vertices not part of the cell (showing the tree structure) are small grey circles and stabilizers are small white cirlces. Tree edges are grey and path edges are black.}\label{fig:cell}
\end{figure}

The {\it head} of $c_r(h)$ consists of the unique 1-vertex of $B_r$, the first three 2-vertices of each branch $B_k$, with $1\leq k  \leq s$ and $k \neq r$, that are not already used in a cell $c_a(h)$, with $1\leq a < r$, and, for each 2-vertex not in $c_r(h)$ and not in $B_r$, the first 3-vertices not already used in a cell $c_a(h)$, with $1\leq a<r$.

The \emph{tail} of $c_r(h)$ consists of a set of $3 \cdot s \cdot (s-1)^2$ branches $B_k$ adjacent to $j_h$. This set is partitioned into $3 \cdot (s-1)^2$ subsets of $s$ subtrees each. The vertices of each of the subsets are distributed between the cells in the same way as for the vertices of the head.

This implies that each cell contains one 1-vertex, $3 \cdot (s-1)$ 2-vertices, and $3 \cdot(s-2)\cdot (s-1)$ 3-vertices of the head, an additional $3 \cdot (s-1)^2$ 1-vertices, $3^2 \cdot (s-1)^3$ 2-vertices, and $3^2 \cdot (s-2) \cdot (s-1)^3$ 3-vertices of the tail, plus $3^2 \cdot (s-1)^4$ stabilizers.

Path \P inside cell $c_r(h)$ visits the vertices in the following order: It starts at the unique 1-vertex of the head, then it reaches all the 2-vertices of the head, then all the 3-vertices of the head, then all the 2-vertices of the tail, and finally all the 3-vertices of the tail, visiting each set in arbitrary order.
After each occurrence of a 2- or 3-vertex of the head, \P visits a 1-vertex of the tail, and after each occurrence of a 2- or a 3-vertex of the tail, it visits a stabilizer of joint $j_h$ (see Fig.~\ref{fig:cell}).

Note that, by this construction, for each joint there exists a set of cells such that each subset of size four contains region-level nonplanar trees with all possible levelings, where the levels correspond to the membership of the vertices to a cell. We now define two bigger structures describing how cells of this set are connected to cells of sets connected to other joints.

{\bf Formation:}
In the definition of a cell we described how the path traverses through one set of branches connected to the same joint. Now we describe how cells from four different sets are connected.

A {\it formation} $F(H), H=(h_1,h_2,h_3,h_4)$ consists of 592 cells, namely of 148 cells $c_r(h_i)$ from the set of cells constructed above for each $1 \leq i \leq 4$. Path \P connects these cells in the order $((h_1h_2h_3)^{37}h_4^{37})^{4}$, that is, \P repeats four times the following sequence: It connects $c_1(h_1)$ to $c_1(h_2)$, then to $c_1(h_3)$, then to $c_2(h_1)$, and so on till $c_{37}(h_3)$, from which it then connects to $c_1(h_4)$, to $c_2(h_4)$, and so on till $c_{37}(h_4)$ (see Fig.~\ref{fig:formation}(a)). A connection between two consecutive cells $c_r(a)$ and $c_r(b)$ is done with an edge connecting the end vertices of the parts \P$(c_r(a))$ and \P$(c_r(b))$ of \P restricted to the vertices of $c_r(a)$ and $c_r(b)$, respectively. Namely, the unique vertex in $c_r(a)$ having degree 1 both in \P$(c_r(a))$ and in \T is connected to the unique vertex in $c_r(b)$ having degree 1 in \P$(c_r(b))$ but not in \T. The following property holds:

\noindent
\begin{property}\label{prop:four_sep_areas_PS}
For any formation $F(H)$ and any joint $j_h$, with $h\in H$, if four cells $c_r(h) \in F(H)$ are pairwise separated by straight lines, then there exists a crossing in $\T$.
\end{property}

\begin{figure}[tb]
  \centering{
\begin{tabular}{c c}
\includegraphics[height=3.5cm]{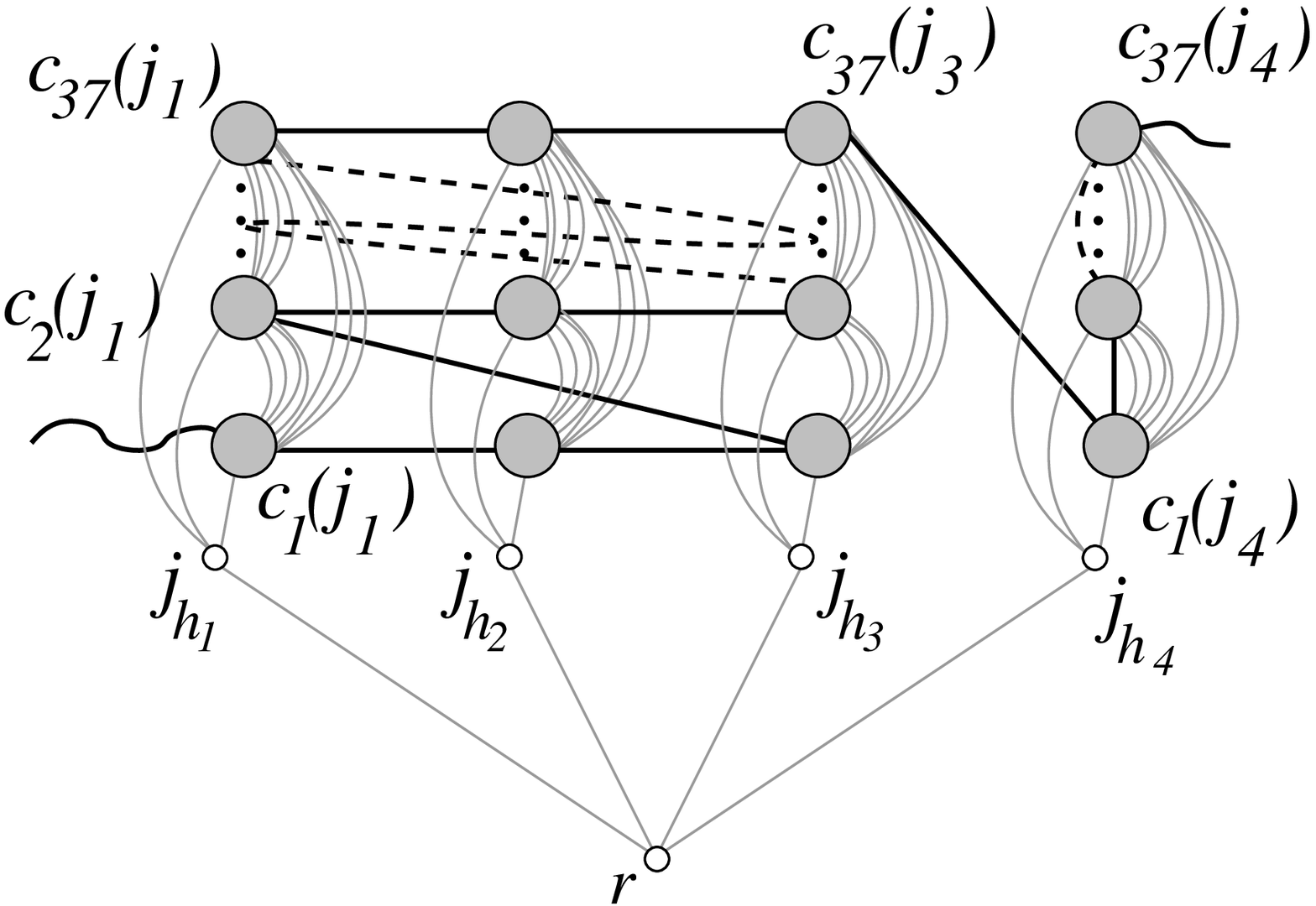} \hspace{1cm} &
\includegraphics[height=3.7cm]{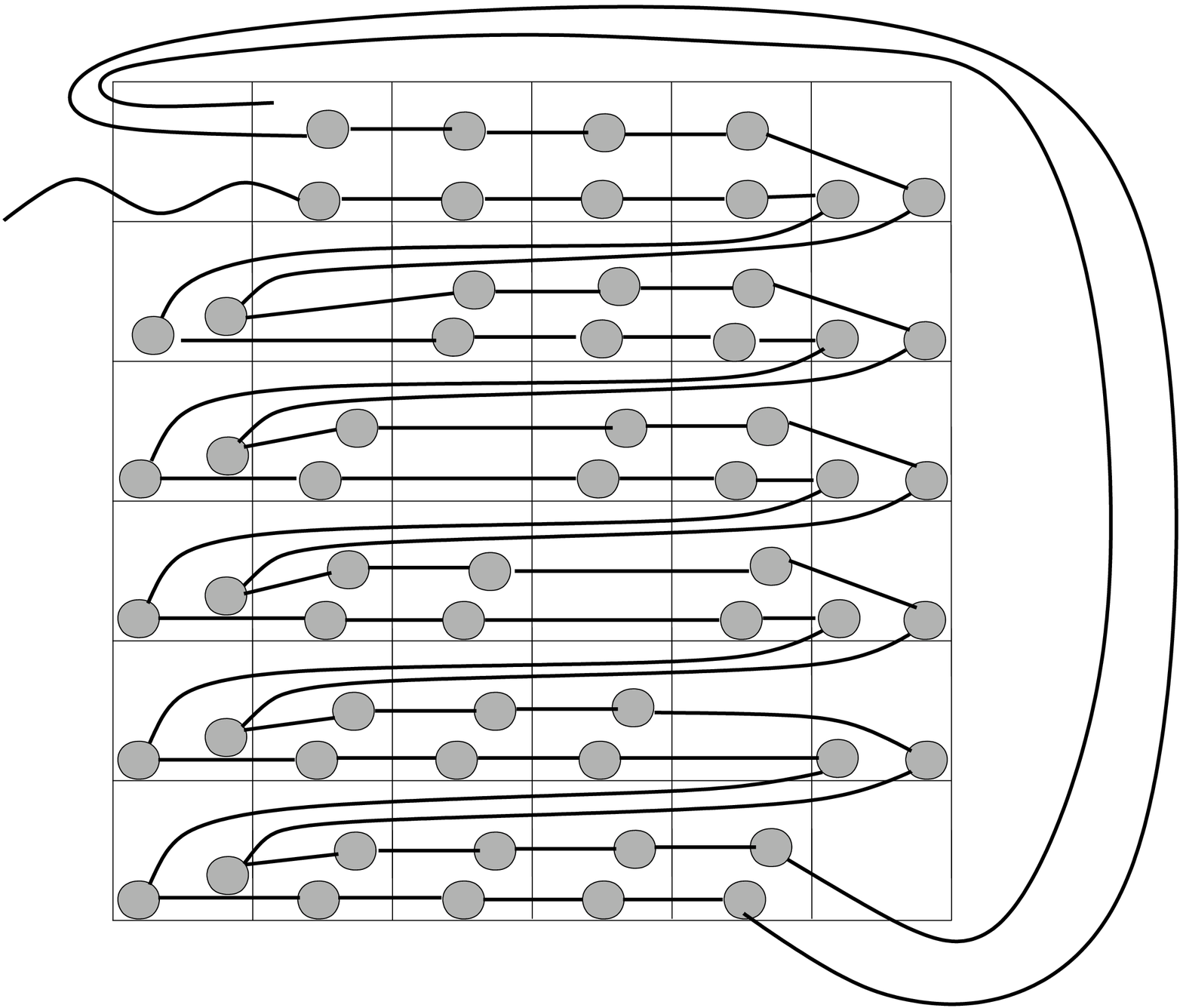} \\
(a) \hspace{1cm} & (b)
\end{tabular}}
\caption{(a) A formation. Tree edges are depicted by grey and path edges by black lines. Please note in this figure also the bundle of tree edges connecting the different cells belonging to the same branch. (b) A subsequence $(H_1,\dots,H_x)^2$ of an extended formation. Formations are inside a table to represent the $4$-tuple they belong to and to emphasize that in each repetition (a row of the table) a formation at a certain $4$-tuple is missing.}\label{fig:formation}
\end{figure}

{\bf Extended Formation:}
Formations are connected by the path in a special sequence, defined as {\it extended formation} and denoted by $EF(H)$, where $H=(H_1=(h_1,\ldots, h_4), \\H_2=(h_5,\ldots , h_8), \ldots, $ $H_x= (h_{4x-3}, \ldots h_{4x}))$ is a tuple of $4-$tuples of disjoint indexes of joints (see Fig.~\ref{fig:formation}(b)). Let $F_1(H_i),\ldots ,F_{y-\frac{y}{x}}(H_i)$ be $y-\frac{y}{x}$ formations  not belonging to any other extended formation and composed of cells of the same set $S$. These formations are connected in the order $(H_1,H_2,\ldots ,H_x)^y$, but in each of these $y$ repetitions one $H_i$ is missing. Namely, in the $k$-th repetition the path does not reach any formation at $H_m$, with $m = k \,\mod\, x$. We say that the $k$-th repetition has a {\it defect} at $m$. We call a subsequence $(H_1,H_2,\ldots,H_x)^x$ a \emph{full repetition} inside $EF(J)$. A full repetition has exactly one defect at each tuple.

Note that the size of $s$ can now be fixed as the number of formations creating repetitions inside one extended formation times the number of cells inside each of these formations, that is $s:=(y-\frac{y}{x})\cdot 37 \cdot 4$. We claim that $x \leq 7\cdot 3^2\cdot 2^{23}$ and $y \leq 7^2\cdot 3^3 \cdot 2^{26}$ is sufficient throughout the proofs. However, for readability reasons, we will keep on using variables $x$ and $y$ in the remainder of the paper.

{\bf Sequence of Extended Formation:}
Extended formations are connected by the path in a special sequence, called {\it sequence of extended formations} and denoted by $SEF(H)$, where $H=(H_1^*, \dots, H_{12}^*)$ is a $12-$tuple of tuples of $4-$tuples. For each tuple $H_i^*$, where $i=1,\ldots, 12$, consider $110$ extended formations $(EF_i(H_1^*), \dots, EF_i(H_{12}^*))$, with $i=1,\dots,110$, not already belonging to any other sequence of extended formations.
These extended formations are connected inside $SEF(H)$ in the order $(H_1^*,\dots,\\H_{12}^*)^{(120)}$. There exist two types of sequences of extended formations. Namely, in the first type there is one extended formation missing in each subsequence $(H_1^*,\dots,H_{12}^*)$, that we call \emph{defect}, as for the extended formations. In the second type, two consecutive extended formations are missing. Namely, in the $k$-th repetition the path skips the extended formations connecting at $H_m^*$ and at $H_{m+1}^*$, with $m = k \,\mathrm{mod}\, 12$. In this case, we say that the repetition has a \emph{double defect}.

Since, for each set of $48x$ joints, $(48x)!$ different disjoint sequences of extended formations exist, we just consider the sequences where the order defined by the tuple is the order of the joints around the root.

\section{Overview}\label{se:overview}

In this section we present the main argumentations leading to the final conclusion that the tree \T and the path \P described in Sect.~\ref{se:tree-path} do not admit any geometric simultaneous embedding. The main idea in this proof scheme is to use the structures given by the path to fix a part of the tree in a specific shape creating specific restrictions for the placement of the further substructures of \T and of \P attached to it.

We first give some further definitions and basic topological properties on the interaction among cells that are enforced by the preliminary arguments about region-level planar drawings and by the order in which the subtrees are connected inside one formation.

{\bf Passage:} Consider two cells $c_1(h),c_2(h)$ that can not be separated by a straight line and a cell $c'(h')$, with $h' \not= h$. We say that there exists a {\it passage} $P$ between $c_1$, $c_2$, and $c'$ if the polyline given by the path of $c'$ separates vertices of $c_1$ from vertices of $c_2$ (see Fig.~\ref{fig:passage}(a)). Since the polyline can not be straight, there is a vertex of $c'$ lying inside the convex hull of the vertices of $c_1\cup c_2$, which implies the following.

\begin{figure}[tb]
  \centering{
\begin{tabular}{c c}
\includegraphics[height=3.4cm]{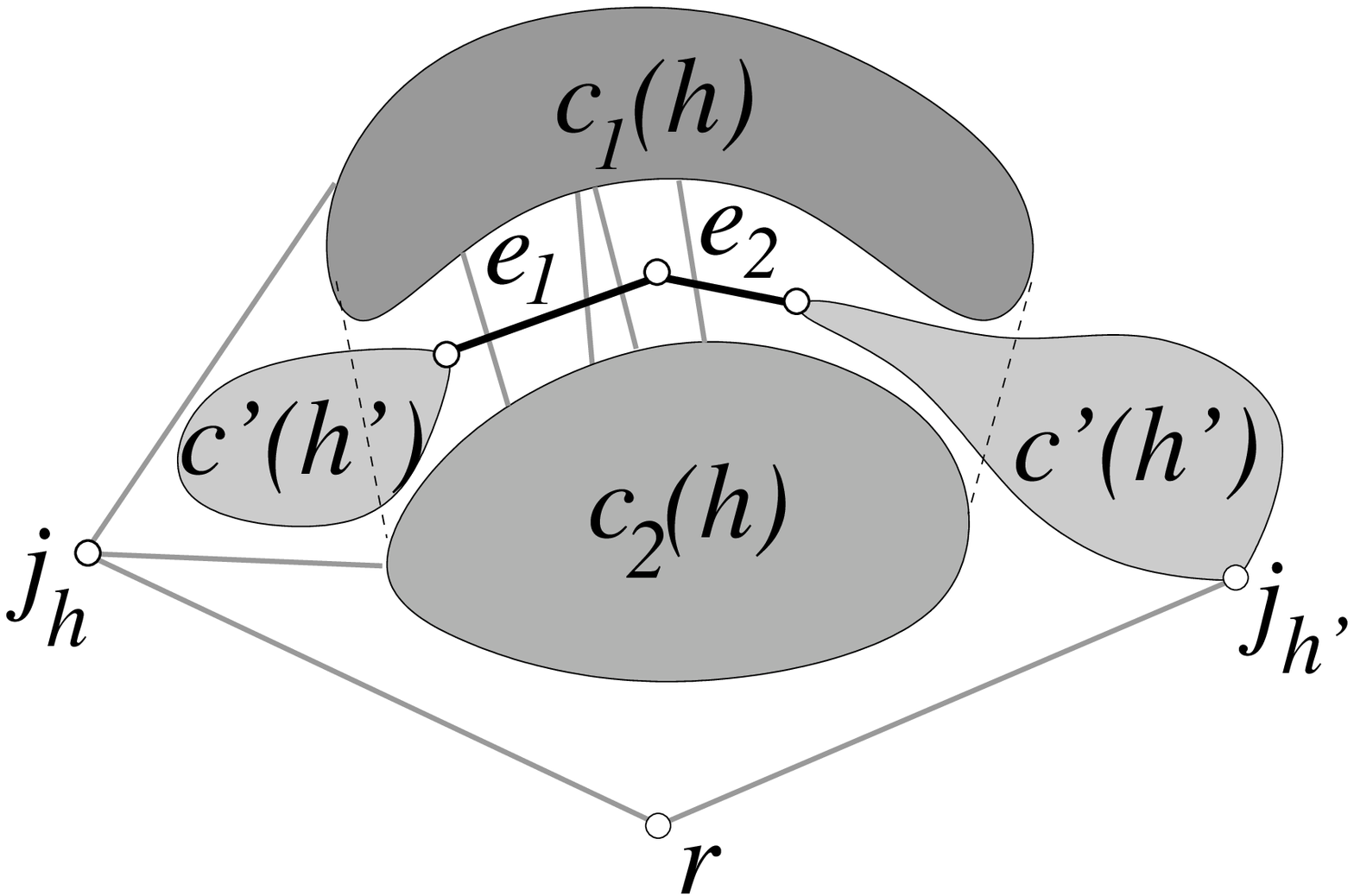} \hspace{1cm} &
\includegraphics[height=3.1cm]{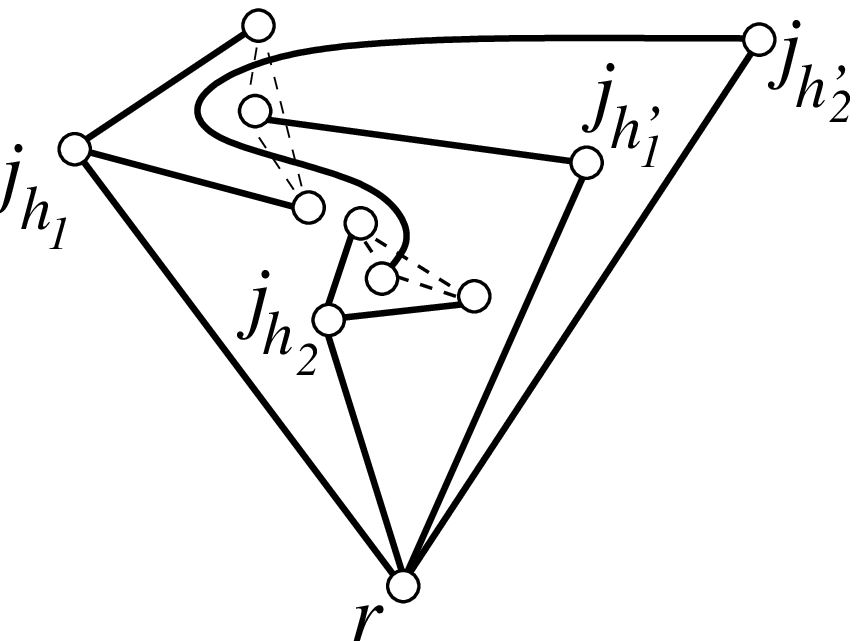} \\
(a) \hspace{1cm} & (b)
\end{tabular}}
\caption{(a) A passage between cells $c_1(h)$, $c_2(h)$, and $c'(h')$. (b) Two interconnected passages.}\label{fig:passage}
\end{figure}

\begin{property}\label{prop:passage_2_edges}
In a passage between cells $c_1$, $c_2$, and $c'$ there exist at least two path-edges $e_1,e_2$ of $c'$ such that both $e_1$ and $e_2$ are intersected by tree-edges connecting vertices of $c_1$ to vertices of $c_2$.
\end{property}

For two passages $P_1$ between $c_1(h_1)$, $c_2(h_1)$, and $c'(h_1')$, and $P_2$ between $c_{3}(h_2)$ , $c_{4}(h_2)$, and $c'(h_2')$ (w.l.o.g., we assume $h_1< h_1'$, $h_2<h_2'$, and $h_1<h_2$), we distinguish three different configurations: (i) If $h_1' < h_2$, $P_1$ and $P_2$ are {\it independent}; (ii) if $h_2' < h_1'$, $P_2$ is {\it nested} into $P_1$; and (iii) if $h_2< h_1'< h_2'$, $P_1$ and $P_2$ are {\it interconnected} (see Fig.~\ref{fig:passage}(b)).

{\bf Doors:}
Let $c_1(h),c_2(h)$, and $c'(h')$ be three cells creating a passage. Consider any triangle given by a vertex $v'$ of $c'$ inside the convex hull of $c_1\cup c_2$ and by any two vertices of $c_1\cup c_2$.
This triangle is a \emph{door} if it encloses neither any other vertex of $c_1,c_2$ nor any vertex of $c'$ that is closer than $v'$ to $j_{h'}$ in \T. A door is \emph{open} if no tree edge incident to $v'$ crosses the opposite side of the triangle, that is, the side between the vertices of $c_1$ and $c_2$ (see Fig.~\ref{fig:door}(a)), otherwise it is \emph{closed} (see Fig.~\ref{fig:door}(b)).

\begin{figure}[hb]
  \centering{
\begin{tabular}{c c}
\includegraphics[height=3.5cm]{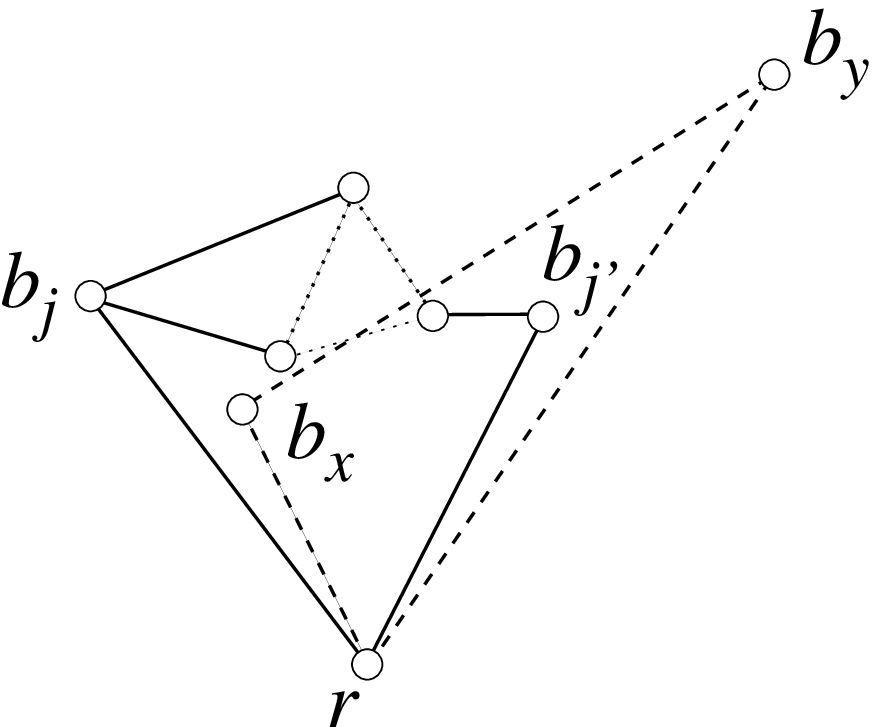} \hspace{1cm} &
\includegraphics[height=3.5cm]{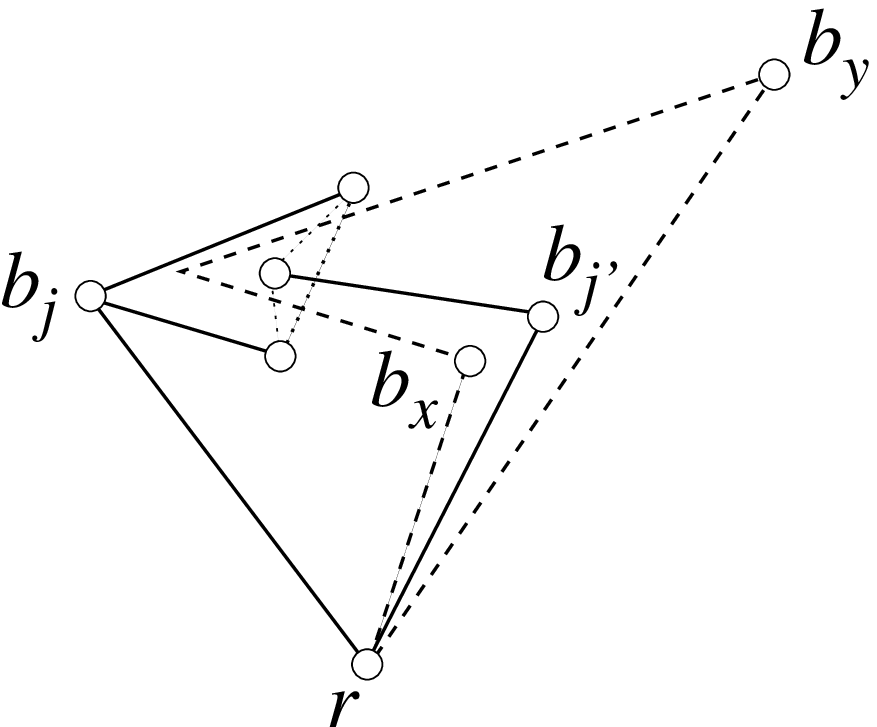} \\
(a) \hspace{1cm} & (b)
\end{tabular}}
\caption{(a) An open door. (B) A closed door.}\label{fig:door}
\end{figure}

Consider two joints $j_a$ and $j_b$, with $h,a,h',b$ appearing in this circular order around the root. Any polyline connecting the root to $j_a$, then to $j_b$, and again to the root, without crossing tree edges, must traverse each door by crossing both the sides adjacent to $v'$. If a door is closed, such a polyline has to bend after crossing one side adjacent to $v'$ and before crossing the other one.
Also, if two passages $P_1$ and $P_2$ are interconnected, either all the closed doors of $P_1$ are traversed by a path of tree-edges belonging to $P_2$ or all the closed doors of $P_2$ are traversed by a path of tree-edges belonging to $P_1$ (see Fig.~\ref{fig:passage}(b)).

In the rest of the argumentation we will exploit the fact that the closed door of a passage requests a bend in the tree to obtain the claimed property that a large part of \T has to follow the same shape. In view of this, we state the following lemmata relating the concepts of doors, passages, and formations.

\begin{lemma}\label{lemma:PS_passage}
For each formation $F(H)$, with $H=(h_1, \ldots ,h_4)$, there exists a passage between some cells $c_1(h_a),c_2(h_a),c'(h_b) \in F(H)$, with $1 \leq a,b \leq 4$.
\end{lemma}

\begin{lemma}\label{lemma:closed_door_in_each_passage}
Each passage contains at least one closed door.
\end{lemma}

From the previous lemmata we conclude that each formation contains at least one closed door. To prove that the effects of closed doors belonging to different formations can be combined to obtain more restrictions on the way in which the tree has to bend, we exploit a combinatorial argument based on the Ramsey Theorem~\cite{grs-rt-90} and state that there exists a set of joints pairwise creating passages.

\begin{lemma}\label{lemma:k_cluster_passage}
Given a set of joints $J=\{j_1,\ldots ,j_y\}$, with $|J|=y:= {2^7\cdot 3\cdot x + 2 \choose 3}$, there exists a subset $J'=\{j'_1,\ldots ,j'_r\}$, with $|J'|=r \geq 2^7\cdot 3\cdot x$, such that for each pair of joints $j'_i,j'_h \in J'$ there exist two cells $c_1(i),c_2(i)$ creating a passage with a cell $c'(h)$.
\end{lemma}

Now we formally define the claimed property that part of the tree has to follow a fixed shape by considering how the drawing of the subtrees attached to two different joints force the drawing of the subtrees attached to the joints between them in the order around the root.

{\bf Enclosing bendpoints:}
Consider two paths $p_1 = \{u_1,v_1,w_1\}$ and $ p_2= \{ u_2,v_2,w_2\}$. The bendpoint $v_1$ of $p_1$ \emph{encloses} the bendpoint $v_2$ of $p_2$ if $v_2$ is internal to triangle $\triangle (u_1,v_1,w_1)$. See Fig.~\ref{fig:bendpoint}(a).

{\bf Channels:}
Consider a set of joints $J =\{j_1, \ldots , j_k \}$ in clockwise order around the root. The \emph{channel} $c_i$ of a joint $j_i$, with $i=2,\dots,k-1$, is the region given by the pair of paths, one path of $j_{i-1}$ and one path of $j_{i+1}$, with the maximum number of enclosing bendpoints with each other. We say that $c_i$ is an \emph{x-channel} if the number of enclosing bendpoints is $x$. Observe that, by Property~\ref{prop:three_bends}, $x\leq 3$. A $3$-channel is depicted in Fig.~\ref{fig:bendpoint}(b). Note that, given an $x$-channel $c_i$ of $j_i$, all the vertices of the subtree rooted at $j_i$ that are at distance at most $x$ from the root lie inside $c_i$.

\begin{figure}[tb]
  \centering{
\begin{tabular}{c c}
\includegraphics[height=2.8cm]{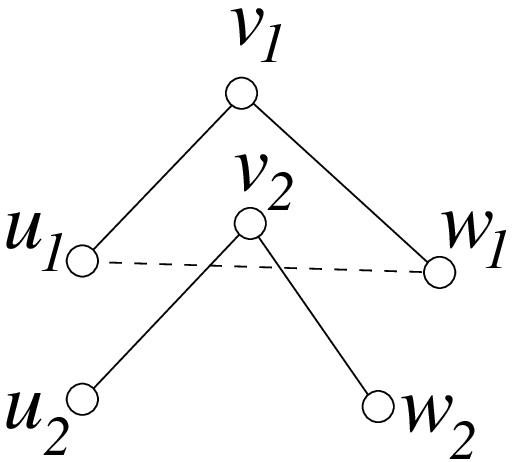} \hspace{1cm} &
\includegraphics[height=3cm]{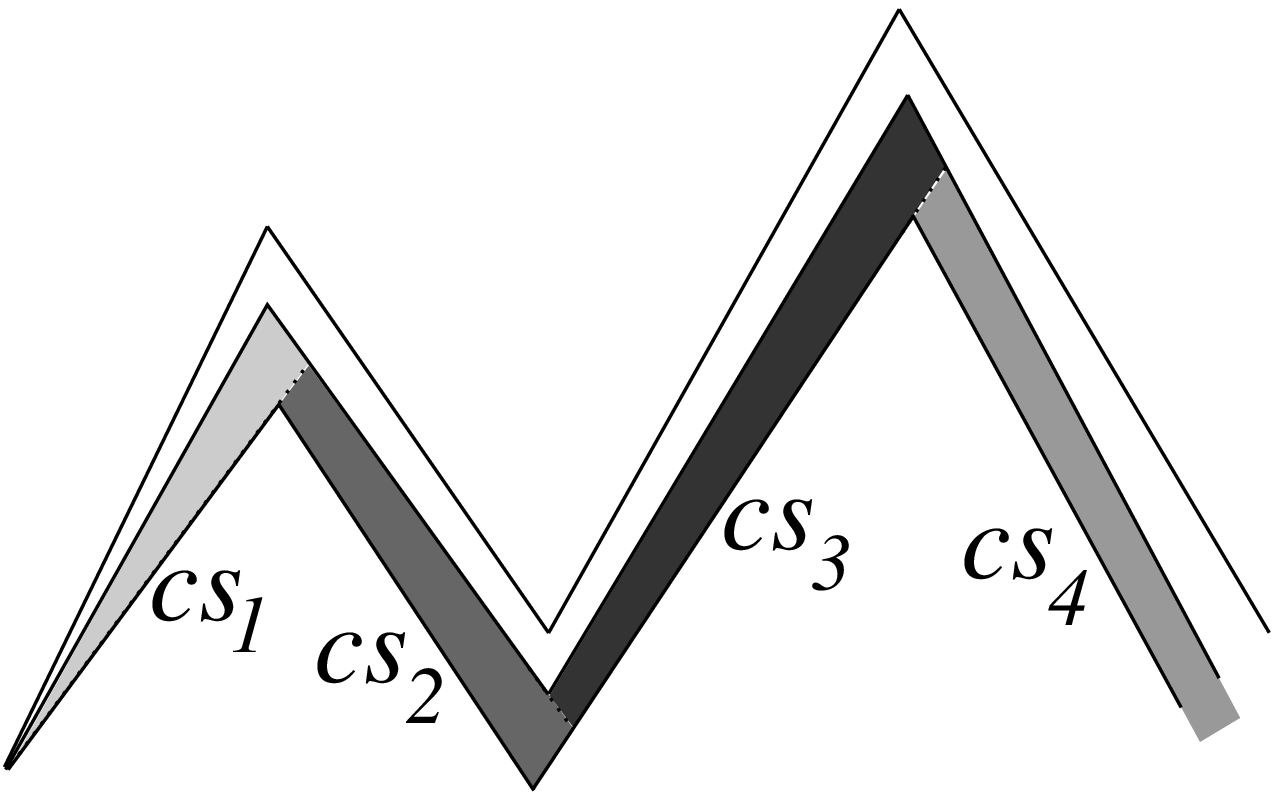} \\
(a) \hspace{1cm} & (b)
\end{tabular}}
\caption{(a) An enclosing bendpoint. (b) A $3$-channel and its channel segments.}\label{fig:bendpoint}
\end{figure}

{\bf Channel segments:}
An $x$-channel $c_i$ is composed of $x+1$ parts called \emph{channel segments} (see Fig.~\ref{fig:bendpoint}(b)). The first channel segment $cs_1$ is the part of $c_i$ that is visible from the root. The $h$-th channel segment $cs_h$ is the region of $c_i$ disjoint from $cs_{h-1}$ that is bounded by the elongations of the paths of $j_{i-1}$ and $j_{i+1}$ after the $h$-th bend.

Observe that, since the channels are created by tree-edges, any tree-edge connecting vertices in the channel has to be drawn inside the channel, while path-edges can cross other channels. In the following we study the relationships between path-edges and channels. The following property descends from the fact that every second vertex reached by \P in a cell is either a 1-vertex or a stabilizer.

\begin{property}\label{prop:CS_1_2}
For any path edge $e=(a,b)$, at least one of $a$ and $b$ lie inside either $cs_1$ or $cs_2$.
\end{property}

{\bf Blocking cuts:}
A \emph{blocking cut} is a path edge connecting two consecutive channel segments by cutting some of the other channels twice. See Fig.~\ref{fig:blocking}.

\begin{property}\label{prop:blocking-cut}
Let $c$ be a channel that is cut twice by a blocking cut. If $c$ has vertices in both the channel segments cut by the path edge, then it has some vertices in a different channel segment.
\end{property}
\begin{proof}
Consider the vertices lying in the two channel segments of $c$. In order to connect them in \T, a vertex $v$ is needed in the bendpoint area of $c$. However, in order to have path connectivity between $v$ and the vertices in the two channel segments, some vertices in a different channel segment are needed.
\end{proof}

\begin{figure}[hb]
  \centering{
\begin{tabular}{c}
\includegraphics[height=2.9cm]{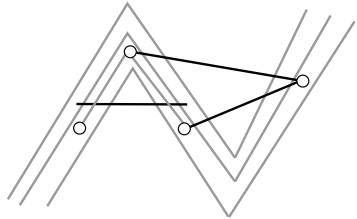}
\end{tabular}}
\caption{A blocking cut.}
\label{fig:blocking}
\end{figure}

In the following lemma we show that in a set of joints as in Lemma~\ref{lemma:k_cluster_passage} it is possible to find a suitable subset such that each pair of paths of tree-edges starting from the root and containing such joints has at least two common enclosing bendpoints, which implies that most of them create $2$-channels.

\begin{lemma}\label{lemma:2_channels}
Consider a set of joints $J= \{j_1,\ldots ,j_k\}$ such that there exists a passage between each pair $(j_i,j_h)$, with $1 \leq i,h \leq k$. Let ${\mathcal P}_1 =\{P \mid P\textrm{ connects } c_i \textrm{ and } c_{\frac{3k}{4}+1-i}$, for $i=1,\dots, \frac{k}{4} \}$ and ${\mathcal P}_2 =\{P \mid P\textrm{ connects } c_{\frac{k}{4}+i} $ and $c_{k+1-i}  \textrm{, for } i=1,\dots, \frac{k}{4} \}$ be two sets of passages between pairs of joints in $J$ (see Fig.~\ref{fig:lemma5}). Then, for at least $\frac{k}{4}$ of the joints of one set of passages, say ${\mathcal P}_1$, there exist paths in \T, starting at the root and containing these joints, which traverse all the doors of ${\mathcal P}_2$ with at least 2 and at most 3 bends. Also, at least half of these joints create an $x$-channel, with $2\leq x \leq 3$.
\end{lemma}

\begin{figure}[tb]
  \centering{
\begin{tabular}{c}
\includegraphics[height=3.3cm]{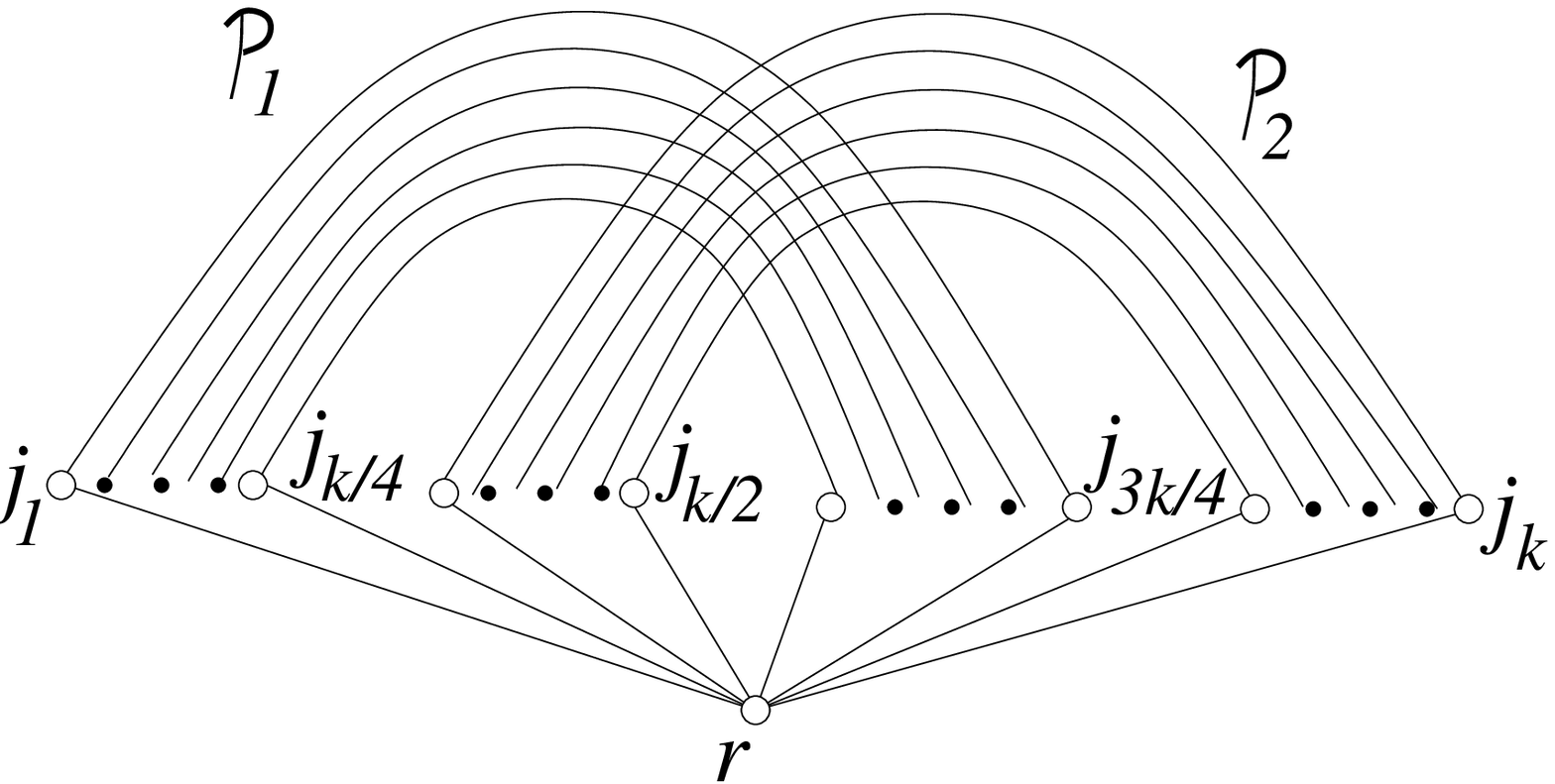}
\end{tabular}}
\caption{Two sets of passages ${\mathcal P}_1$ and ${\mathcal P}_2$ as described in Lemma~\ref{lemma:2_channels}.}
  \label{fig:lemma5}
\end{figure}

By Lemma~\ref{lemma:2_channels}, any formation attached to a certain subset of joints must use at least three different channel segments. In the remainder of the argumentation we focus on this subset of joints and give some properties holding for it, in terms of interaction between different formations with respect to channels. Since we need a full sequence of extended formations attached to these joints, $k$ has to be at least eight times the number of channels inside a sequence of extended formations, that is, $k \geq 8 \cdot 48x=2^7\cdot 3 x$.

First, we give some further definitions.

{\bf Nested formations}
A formation $F$ is {\it nested} in a formation $F'$ if there exist two edges $e_1,e_2 \in F$ and two edges $e'_1,e'_2 \in F'$ cutting a border $cb$ of a channel $c$ such that all the vertices of the path in $F$ between $e_1$ and $e_2$ lie inside the region delimited by $cb$ and by the path in $F'$ between $e'_1$ and $e'_2$ (see Fig.~\ref{fig:nested_formations}(a)).

A series of pairwise nested formations $F_1, \ldots, F_k$ is {\it $r$-nested} if there exist $r$ formations $F_{q_1}, \ldots, F_{q_r}$, with $1\leq q_1,\ldots,q_r \leq k$, belonging to the same channel and such that, for each pair $F_{q_p},F_{q_{p+1}}$, there exists at least one formation $F_z$, $1 \leq z \leq k$, belonging to another channel and such that $F_{q_p}$ is nested in $F_z$ and $F_z$ is nested in $F_{q_{p+1}}$ (see Fig.~\ref{fig:nested_formations}(b)).

\begin{figure}[htb]
  \centering{
\begin{tabular}{c c}
\includegraphics[height=3.3cm]{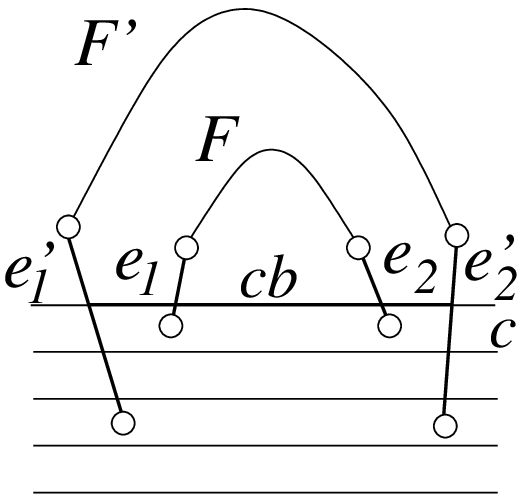} \hspace{1cm} &
\includegraphics[height=3.3cm]{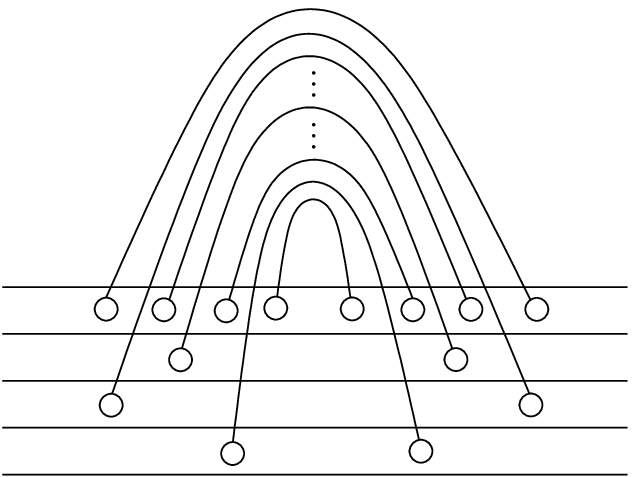} \\
(a) \hspace{1cm} & (b)
\end{tabular}}
\caption{(a) A formation $F$ nested in a formation $F'$. (b) A series of $r$-nested formations.}
  \label{fig:nested_formations}
\end{figure}

{\bf Independent sets of formations}
Let $S_1, \ldots , S_k$ be sets of formations of one extended formation such that each set $S_i$ contains formations $F_i(H_1), \ldots, F_i(H_r)$ on the set of $4$-tuples $H=\{H_1, \ldots, H_r\}$, where the joints of $H_i$ are between the joints of $H_{i-1}$ and of $H_{i+1}$ in the order around the root. Further, let $F_a(H_c)$ be not nested in $F_b(H_d)$, for each $1\leq a,b \leq k, \; a \neq b,$ and $1\leq c,d \leq r$. If for each pair of sets $S_a,S_b$ there exist two lines $l_1,l_2$ separating the vertices of $S_a$ and $S_b$ inside channel segment $cs_1$ and $cs_2$, respectively, the sets are \emph{independent} (see Fig.~\ref{fig:independent}).

\begin{figure}[htb]
  \centering{
\begin{tabular}{c}
\includegraphics[height=3.8cm]{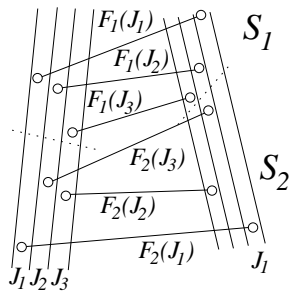}
\end{tabular}}
\caption{Two independent sets $S_1$ and $S_2$.}
  \label{fig:independent}
\end{figure}

In the following lemmata we prove that in any extended formation there exists a nesting of a certain depth (Lemma~\ref{lemma:nest_independent}). This important property will be the starting point for the final argumentation and will be deeply exploited in the rest of the paper. We get to this conclusion by first proving that in an extended formation the number of independent sets of formations is limited (Lemma~\ref{n-independent nestings}) and then by showing that, although there exist formations that are neither nested nor independent, in any extended formation there exists a certain number of pairs of formation that have to be either independent or nested (Lemma~\ref{lem:nest_sequence}).

\begin{lemma}\label{n-independent nestings}
There exist no $n\geq 2^{22}\cdot 14$ independent sets of formations $S_1,\ldots ,S_n$ inside any extended formation,
where each $S_i$ contains formations of a fixed set of channels of size $r\geq 22$.
\end{lemma}

\begin{lemma}\label{lem:nest_sequence}
Consider four subsequences $Q_1,\ldots ,Q_4$, where $Q_i = (H_1,H_2,\ldots,H_x)$, of an extended formation $EF$, each consisting of a whole repetition of $EF$. Then, there exists either a pair of nested subsequences or a pair of independent subsequences.
\end{lemma}

\begin{lemma}\label{lemma:nest_independent}
Consider an extended formation $EF(H_1,H_2,\ldots,H_x)$.
Then, there exists a $k$-nesting, where $k \geq 6$, among the formations of $EF$.
\end{lemma}

Once the existence of $2$-channels and of a nesting of a certain depth in each extended formation has been shown, we turn our attention to study how such a deep nesting can be performed inside the channels.

Let $cs_a$ and $cs_b$, with $1\leq a,b \leq 4$, be two channel segments. If the elongation of $cs_a$ intersects $cs_b$, then it is possible to connect from $cs_b$ to $cs_a$ by cutting both the sides of $cs_a$. In this case, $cs_a$ and $cs_b$ have a \emph{$2-$side connection} (see Fig.~\ref{fig:12-side-connection}(b)). On the contrary, if the elongation of $cs_a$ does not intersect $cs_b$, only one side of $cs_a$ can be used. In this case, $cs_a$ and $cs_b$ have a \emph{$1-$side connection} (see Fig.~\ref{fig:12-side-connection}(a)).

\begin{figure}[htb]
\begin{center}
\begin{tabular}{c c}
\mbox{\includegraphics[height=2.7cm]{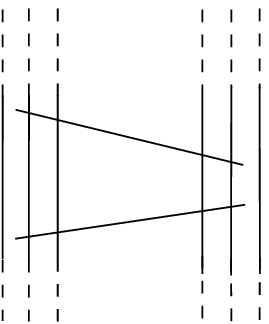}} \hspace{1cm} &
\mbox{\includegraphics[height=2.7cm]{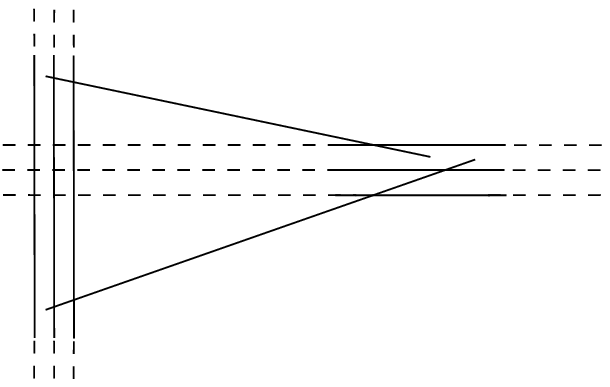}} \\
(a) & (b)\\
\end{tabular}
\caption{(a) A $1-$side connection. (b) A $2-$side connection.}
\label{fig:12-side-connection}
\end{center}
\end{figure}

Based on these different ways of connecting distinct channel segments, we split our proof into three parts, the first one dealing with the setting in which only $1$-side connections are allowed, the second one allowing one single $2$-side connection, and the last one tackling the general case.

\begin{proposition}\label{prop:only-1-side}
If there exist only $1-$side connections, then \T and \P do not admit any geometric simultaneous embedding.
\end{proposition}

We prove this proposition by showing that, in this configuration, the existence of a deep nesting in a single extended formation, proved in Lemma~\ref{lemma:nest_independent}, results in a crossing in either \T or \P.

\begin{lemma}\label{lemma:k-nesting}
If an extended formation lies in a part of the channel that contains only $1-$side connections, then \T and \P do not admit any geometric simultaneous embedding.
\end{lemma}

Next, we study the case in which there exist $2$-side connections. We distinguish two types of $2$-side connections, based on the fact that the elongation of channel segment $cs_a$ intersecting channel segment $cs_b$ starts at the bendpoint that is closer to the root, or not. In the first case we have a \emph{low Intersection} (see Fig.~\ref{fig:intersection}(a)), denoted by $I^l_{(a,b)}$, and in the second case we have a \emph{high Intersection} (see Fig.~\ref{fig:intersection}(b)), denoted by $I^h_{(a,b)}$, where $a,b \in \{1,\dots,4\}$. We use the notation $I_{(a,b)}$ to describe both $I^h_{(a,b)}$ and $I^l_{(a,b)}$. We say that two intersections $I_{(a,b)}$ and $I_{(c,d)}$ are \emph{disjoint} if $a,d \in \{1,2\}$ and $b,c \in \{3,4\}$. For example, $I_{(1,3)}$ and $I_{(4,2)}$ are disjoint, while $I_{(1,3)}$ and $I_{(2,4)}$ are not.

\begin{figure}[htb]
\begin{center}
\begin{tabular}{c c}
\mbox{\includegraphics[height=2.5cm]{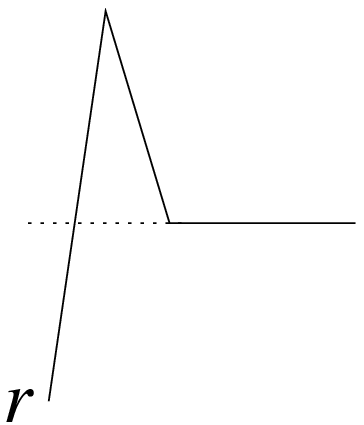}} \hspace{2cm} &
\mbox{\includegraphics[height=2.5cm]{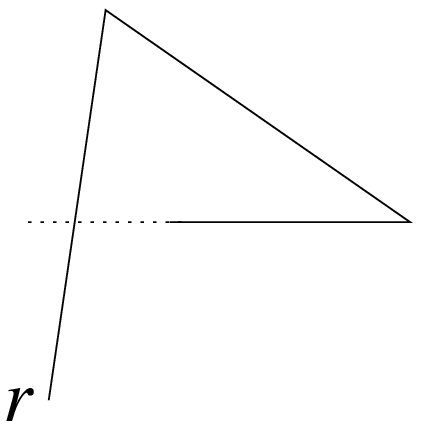}} \\
(a) \hspace{2cm} & (b)\\
\end{tabular}
\caption{(a) A low Intersection. (b) A high Intersection.}
\label{fig:intersection}
\end{center}
\end{figure}

Since consecutive channel segments can not create any $2$-side connection, in order to explore all the possible shapes we consider all the combinations of low and high intersections created by channel segments $cs_1$ and $cs_2$ with channel segments $cs_3$ and $cs_4$.
With the intent of proving that intersections of different channels have to maintain certain consistencies, we state the following lemma.

\begin{lemma}\label{lem:2-different-shapes}
Consider two channels $ch_p,ch_q$ with the same intersections. Then, none of channels $ch_i$, where $p<i<q$, have an intersection that is disjoint with the intersections of $ch_p$ and of $ch_q$.
\end{lemma}

As for Proposition~\ref{prop:only-1-side}, in order to prove that $2$-side connections are not sufficient to obtain a simultaneous embedding of \T and \P, we exploit the existence of the deep nesting shown in Lemma~\ref{lemma:nest_independent}. First, we analyze some properties relating such nesting to channel segments and bending areas. A \emph{bending area} $b(a,a+1)$ is the region between $cs_a$ and $cs_{a+1}$ where bendpoints can be placed. We first observe that all the extended formations have to place vertices inside the bending area of the channel segment where the nesting takes place, and then prove that not many of the formations involved in the nesting can use the part of the path that creates the nesting to place vertices in such a bending area, which implies that the extended formations have to reach the bending area in a different way.

\begin{lemma}\label{lem:nesting-bending-area}
Consider an $x$-nesting of a sequence of extended formations on an intersection $I_{(a,b)}$, with $a\leq 2$.
Then, there exists a triangle $t$ in the nesting that separates some of the triangles nesting with $t$ from the bending area $b(a,a+1)$ (or $b(a-1,a)$).
\end{lemma}

Then, we study some of the cases involving $2$-side connections and we show that the connections between the bending area and the "endpoints" of the nesting create a further nesting of depth greater than $6$. Hence, if no further $2-$side connection is available, this second nesting is not drawable.

\begin{proposition}\label{prop:triangle}
Let $t$ be a triangle open on a side splitting a channel segment $cs$ into two parts such that every extended formation $EF$ has vertices in both parts. If the only possibility to connect vertices in different parts of $cs$ is with a $1$-side connection and if any such connection creates a triangle open on a side that is nested with $t$, then \T and \P do not admit any geometric simultaneous embedding.
\end{proposition}

\begin{figure}[htb]
\begin{center}
\begin{tabular}{c}
\mbox{\includegraphics[height=4cm]{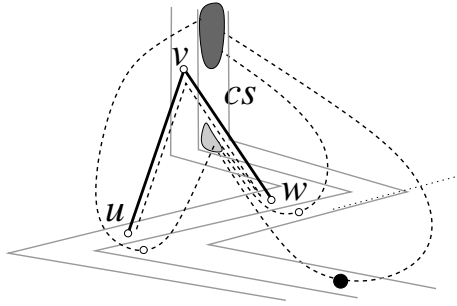}} \hspace{0.2cm}
\end{tabular}
\caption{A situation as in Proposition~\ref{prop:triangle}. The chosen turning vertex is represented by a big black circle and is in configuration $\beta$. The inner and the outer areas are represented by a light grey and a dark grey region, respectively.}
\label{fig:turning-vertex}
\end{center}
\end{figure}

Refer to Fig.~\ref{fig:turning-vertex}. Consider the two path-edges $e_1=(u,v),e_2=(v,w)$ creating $t$ such that the common point $v$ is in the channel segment $cs$ that is split into two parts, that we call {\it inner area} and {\it outer area}, respectively. We assume that $e_1,e_2$ do not cut any channel segment $cs'$ completely, since such a cut would create more restrictions than placing $u$ or $w$ inside $cs'$. Consider the path in an extended formation $EF$ connecting the inner and the outer area through a $1$-side connection at $cs'$. As a generalization, consider for such a path of $EF$ only a vertex, called \emph{turning vertex}, which is placed in $cs'$ and for which no other path in $EF$ exists that connects the inner and the outer area by using a channel segment $cs''$ such that the subpath to $cs''$ intersects either $cs''$ or its elongation. If there exist more than one of such vertices, then arbitrarily choose one of them. Observe that the path connecting from the inner area to the outer area through the turning vertex encloses exactly one of $u$ and $w$. If it encloses $u$, it is in configuration $\alpha$, otherwise it is in configuration $\beta$. If there exist both paths in $\alpha$ and paths in $\beta$ configuration, then we arbitrarily consider one of them. Finally, consider the connections between different extended formations inside a sequence of extended formations. Consider a turning vertex $v$ in a channel segment $cs$ of a channel $ch$ such that the edges incident to $v$ cut a channel $ch'$. Then, any connection of an extended formation of $ch'$ from the inner to the outer area in the same configuration as $ch$ and with its turning vertex $v'$ in $cs$ is such that $v'$ lies inside the convex hull of the two edges incident to $v$.

In the following two lemmata we show that in the setting described in Proposition~\ref{prop:triangle} there exists a crossing either in \T or in \P.

\begin{lemma}
\label{lem:one_channel_segment}
In a situation as described in Proposition~\ref{prop:triangle}, not all the extended formations in a sequence of extended formations can place turning vertices in the same channel segment.
\end{lemma}

\begin{lemma}\label{lem:prop3-nodrawing}
In a situation as described in Proposition~\ref{prop:triangle}, \T and \P do not admit any geometric simultaneous embedding.
\end{lemma}

Based on the property given by Proposition~\ref{prop:triangle}, we present the second part of the proof, in which we show that having two intersections $I_{(a,b)}$ and $I_{(c,d)}$ does not help if $I_{(a,b)}$ and $I_{(c,d)}$ are not disjoint.

\begin{proposition}\label{prop:non-disjoint-intersections}
If there exists no pair of disjoint $2$-side connections, then \T and \P do not admit any geometric simultaneous embedding.
\end{proposition}

Observe that, in this setting, it is sufficient to restrict the analysis to cases $I_{(1,3)}$ and $I_{(3,1)}$, since the cases involving $2$ and $4$ can be reduced to them.

\begin{lemma}\label{lem:intersects_one_three}
If a shape contains an intersection $I_{(1,3)}$ and does not contain any other intersection that is disjoint with $I_{(1,3)}$, then \T and \P do not admit any geometric simultaneous embedding.
\end{lemma}

\begin{lemma}\label{lem:intersects_three_one}
If there exists a sequence of extended formation in any shape containing an intersection $I_{(3,1)}$, then \T and \P do not admit any geometric simultaneous embedding.
\end{lemma}

Observe that, in the latter lemma, we proved a property that is stronger than the one stated in Proposition~\ref{prop:non-disjoint-intersections}. In fact, we proved that a simultaneous embedding cannot be obtained in any shape containing an intersection $I_{(3,1)}$, even if a second intersection that is disjoint with $I_{(3,1)}$ is present.

Finally, in the third part of the proof, we tackle the general case where two disjoint intersections exist.

\begin{proposition}\label{prop:disjoint}
If there exists two disjoint $2$-side connections, then \T and \P do not admit any geometric simultaneous embedding.
\end{proposition}

Since the cases involving intersection $I_{3,1}$ were already considered in Lemma~\ref{lem:intersects_three_one}, we only have to consider the eight different configurations where one intersection is $I_{(1,3)}$ and the other is one of $I_{(4,\{1,2\})}$. In the next three lemmata we cover the cases involving $I_{(1,3)}^h$ and in Lemma~\ref{lem:cs_two_convex_hull} the ones involving $I_{(1,3)}^l$.

Consider two consecutive channel segments $cs_i$ and $cs_{i+1}$ of a channel $c$ and let $e$ be a path-edge crossing the border of one of $cs_i$ and $cs_{i+1}$, say $cs_i$. We say that $e$ creates a \emph{double cut} at $c$ if the elongation of $e$ cuts $c$ in $cs_{i+1}$. A double cut is \emph{simple} if $e$ does not cross $cs_{i+1}$ (see Fig.~\ref{fig:double-cut}(a)) and \emph{non-simple} otherwise (see Fig.~\ref{fig:double-cut}(b)). Also, a double cut of an extended formation $EF$ is \emph{extremal} with respect to a bending area $b(x,x+1)$ if there exists no double cut of $EF$ that is closer than it to $b(x,x+1)$.

\begin{figure}[htb]
\begin{center}
\begin{tabular}{c c}
\mbox{\includegraphics[height=3.3cm]{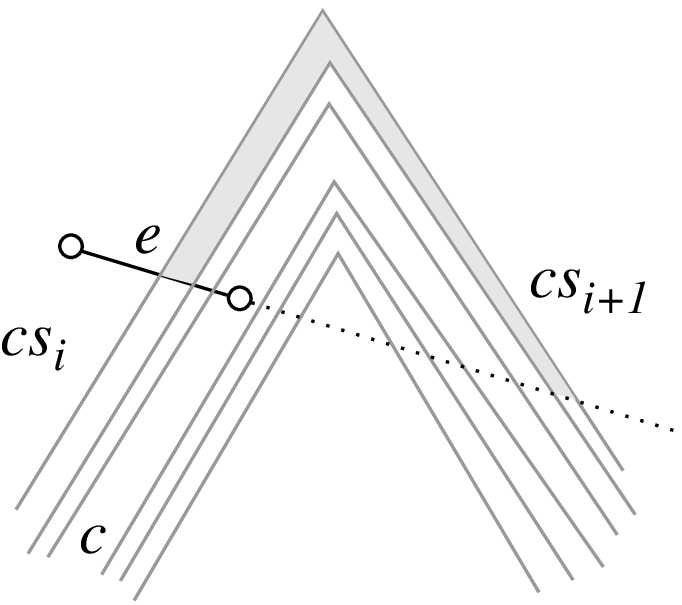}} \hspace{1.5cm} &
\mbox{\includegraphics[height=3.3cm]{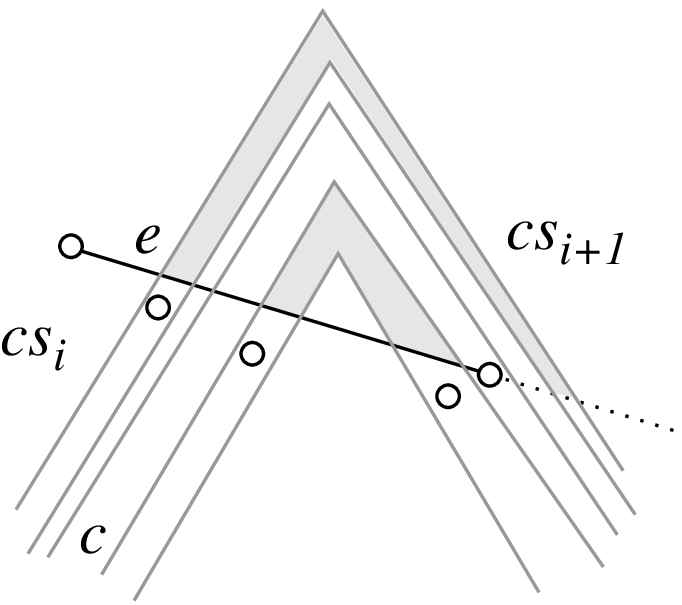}} \\
(a) \hspace{1.5cm} & (b)\\
\end{tabular}
\caption{(a) A simple double cut. (b) A non-simple double cut.}
\label{fig:double-cut}
\end{center}
\end{figure}

\begin{property}\label{prop:double-cut}
Any edge $e_k$ creating a double cut at a channel $k$ in channel segment $cs_i$ blocks visibility to the bending area $b(i,i+1)$ for a part of $cs_i$ in each channel $ch_h$ with $h>k$ (with $h<k$).
\end{property}

In the following lemma we show that a particular ordering of extremal double cuts in two consecutive channel segments leads to a non-planarity in \T or \P. Note that, any order of extremal double cuts corresponds to an order of the connections of a subset of extended formations to the bending area.

\begin{lemma}\label{lemma:no-ordered-double-cuts}
Let $cs_i$ and $cs_{i+1}$ be two consecutive channel segments. If there exists an ordered set $S:=(1,2,\ldots ,5)^3$ of extremal double cuts cutting $cs_i$ and $cs_{i+1}$ such that the order of the intersections of the double cuts with $cs_i$ (with $cs_{i+1}$) is coherent with the order of $S$, then \T and \P do not admit any geometric simultaneous embedding.
\end{lemma}

Then, we show that shape $I_{(1,3)}^h$ $I_{(4,2)}$ induces this order. To prove this, we first state the existence of double cuts in shape $I_{(1,3)}^h$ $I_{(4,2)}^h$. The existence of double cuts in shape $I_{(1,3)}^h$ $I_{(4,2)}^l$ can be easily seen.

\begin{lemma}\label{lem:double_cuts_13}
Each extended formation in shape $I_{(1,3)}^h$ $I_{(4,2)}^h$ creates double cuts in at least one bending area.
\end{lemma}

\begin{lemma}\label{lem:ordered-set-of-double-cuts-exists}
Every sequence of extending formations in shape $I_{(1,3)}^h$ $I_{(4,2)}^{h,l}$ contains an ordered set $(1,2,\ldots ,5)^3$ of extremal double cuts with respect to bending area either $b(2,3)$ or $b(3,4)$.
\end{lemma}

Finally, we consider the configurations where one intersection is $I_{(1,3)}^l$ and the other one is one of $I_{(4,2)}^{h,l}$. Observe that, in both cases, channel segment $cs_2$ is on the convex hull.

\begin{lemma}\label{lem:cs_two_convex_hull}
If channel segment $cs_2$ is part of the convex hull, then \T and \P do not admit any geometric simultaneous embedding.
\end{lemma}

Based on the above discussion, we state the following theorem.

\begin{theorem}
There exist a tree and a path that do not admit any geometric simultaneous embedding.
\end{theorem}
\begin{proof}
Let $\T$ and $\P$ be the tree and the path described in Sect.~\ref{se:tree-path}. Then, by Lemma~\ref{lemma:2_channels}, Lemma~\ref{lem:2-different-shapes}, and Property~\ref{prop:three_bends}, a part of \T has to be drawn inside channels having at most four channel segments. Also, by Lemma~\ref{lemma:nest_independent}, there exists a nesting of depth at least $6$ inside each extended formation.

By Proposition~\ref{prop:only-1-side}, if there exist only $1$-side connections, then \T and \P do not admit any simultaneous embedding. By Proposition~\ref{prop:non-disjoint-intersections}, if there exists either one $2$-side connections or a pair of non-disjoint intersections, then \T and \P do not admit any simultaneous embedding. By Proposition~\ref{prop:disjoint}, even if there exist two disjoint $2$-side intersections, then \T and \P do not admit any simultaneous embedding. Since it is not possible to have more than two disjoint $2$-side intersections, the statement follows.
\end{proof}

\section{Detailed Proofs}\label{se:proofs}

\rephrase{Lemma}{\ref{lemma:PS_passage}}{
For each formation $F(H)$, with $H=(h_1, \ldots ,h_4)$, there exists a passage between some cells $c_1(h_a),c_2(h_a),c'(h_b) \in F(H)$, with $1 \leq a,b \leq 4$.
}

\begin{proof}
Suppose, for a contradiction, that there exists no passage inside $F(H)$. First observe that, if two cells $c_1(h_a),c_2(h_a) \in F(H)$ are separated by a polyline given by the path passing through $F(H)$, then either they are separable by a straight line or such a polyline is composed of edges belonging to a cell $c_3(h_a)$ of the same joint $j_{h_a}$. Since, by Property~\ref{prop:four_sep_areas_PS}, there exists no set of four cells of a given joint inside $F(H)$ that are separable by a straight line, it follows that all the cells of $F(H)$ of a given joint can be grouped into at most $3$ different sets $S^1$, $S^2$, and $S^3$ such that cells from different sets can be separated by straight lines, but cells from the same set can not. Therefore, the cells inside one of these sets can only be separated by other cells of the same set.

\begin{figure}[ht]
\begin{center}
\includegraphics[width=8cm]{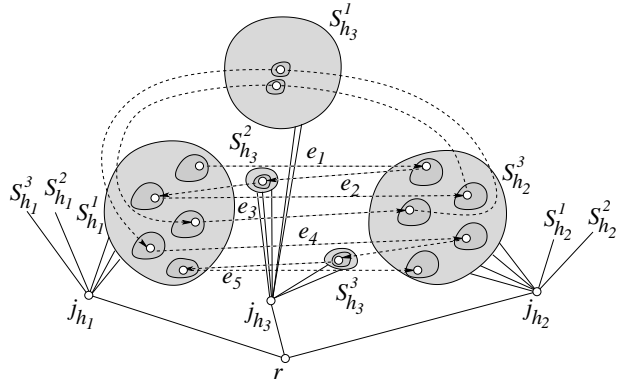}
\caption{The five path edges $e_1,\ldots ,e_5$ connecting five cells of set $S_{h_1}^a$ with five cells of set $S_{h_2}^b$.}
\label{fig:regions}
\end{center}
\end{figure}

Consider the connections of the path through $F(H)$ with regard to this notion of sets of cells.
Let $S_{h_x}^y$, with $x=1,\dots,4$ and $y=1,\dots,3$, be the set of cells belonging to set $S^y$ and attached to joint $j_{h_x}$. Hence, for any two cells $c_1(h_x),c_2(h_{x+1})$ there are nine possible ways to connect between some $S_{h_x}^y$ and $S_{h_{x+1}}^{y'}$.
Since the part of \P through $F(H)$ visits 37 times cells from $j_{h_1}, j_{h_2}, j_{h_3}$, in this order, there exist five path edges $e_1,\ldots ,e_5$ connecting five cells of set $S_{h_1}^a$ with five cells of set $S_{h_2}^b$, where $1\leq a,b \leq 3$ (see Fig.\ref{fig:regions}). Without loss of generality, we assume that edges $e_1,\ldots ,e_5$ appear in this order in the part of \P through $F(H)$. Observe that $e_1,\ldots ,e_5$, together with the five cells of $S_{h_1}^a$ and the five cells of $S_{h_2}^b$ they connect, subdivide the plane into five regions. Since the path is continuous in $F(H)$, it connects from the end of $e_1$ (a cell of joint $j_{h_2}$) to the beginning of $e_2$ (a cell of joint $j_{h_1}$), from the end of $e_2$ to the beginning of $e_3$, and so on. If in the region between edges $e_s$ and $e_{s+1}$, with $1\leq s \leq 4$, there exists no cell of joint $j_{h_3}$, then the path through $F(H)$ will not traverse the region between these edges in the opposite direction, since the path contains no edges going from a cell of $j_{h_2}$ to a cell of $j_{h_1}$ and since the start- (and end-) cells of these edges cannot be separated by straight lines.
Furthermore, note that, in this case, the path-connection from $e_s$ to $e_{s+1}$ does not traverse the region between the edges, therefore forming a spiral shape, in the sense that the part of the path following $e_{s+1}$ is separated from the part of the path prior to $e_s$.
Since we have five edges between $S_{h_1}^a$ and $S_{h_2}^b$ but only 3 possible sets of cells on joint $j_{h_3}$, at least one pair of edges exists creating an empty region and therefore a spiral separating the path.

By this argument, it follows that cells attached to joint $j_{h_4}$ in different repetitions of the subsequence $((h_1h_2h_3)^{37}h_4^{37})$ in $F(H)$ are separated by path edges of the spirals formed by the repeated subsequence of visited cells of the joints $j_{h_1}, j_{h_2} , j_{h_3}$. Since four repetitions create four of such separated cells on $j_{h_4}$, by Property~\ref{prop:four_sep_areas_PS} there exists a pair of cells that are not separable by a straight line but are separated by the path. Since the path of the spiral separating them consists only of cells on different joints, any possible separating polyline leads to a contradiction to the non-existence of a passage inside $F(H)$.
\end{proof}

\rephrase{Lemma}{\ref{lemma:closed_door_in_each_passage}}{
Each passage contains at least one closed door.
}

\begin{figure}[ht]
\begin{center}
\includegraphics[height=3.5cm]{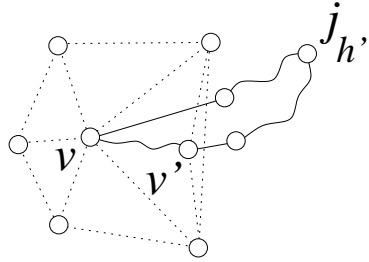}
\caption{There exists a closed door in each passage.}
\label{fig:closed_door_in_passage}
\end{center}
\end{figure}
\begin{proof}
Refer to Fig.~\ref{fig:closed_door_in_passage}. Let $P_1$ be a passage between $c_1(h)$, $c_2(h)$, and $c'(h')$. Consider any vertex $v$ of $c'$ inside the convex hull of $C:= c_1\cup c_2$. Further, consider all the triangles $\triangle(v,v_1,v_2)$ created by $v$ with any two vertices $v_1,v_2 \in C$ such that $\triangle(v,v_1,v_2)$ does not enclose any other vertex of $C$. The path of tree edges connecting $v$ to $j_{h'}$ enters one of the triangles. Then, either it leaves the triangle on the opposite side, thereby creating a closed door, or it encounters a vertex $v'$ of $c'$. Since at least one vertex of $c'$ lies outside the convex hull of $C$, otherwise they would not be separated by $c'$, it is possible to repeat the argument on triangle $\triangle(v',v_1,v_2)$ until a closed door is found.
\end{proof}

\rephrase{Lemma}{\ref{lemma:k_cluster_passage}}{
Given a set of joints $J=\{j_1,\ldots ,j_y\}$, with $|J|=y:= {2^7\cdot 3\cdot x + 2 \choose 3}$, there exists a subset $J'=\{j'_1,\ldots ,j'_r\}$, with $|J'|=r \geq 2^7\cdot 3\cdot x$, such that for each pair of joints $j'_i,j'_h \in J'$ there exist two cells $c_1(i),c_2(i)$ creating a passage with a cell $c'(h)$.
}

\begin{proof}
By construction of the tree, for each set of four joints, there are formations that visit only cells of these joints. By Lemma~\ref{lemma:PS_passage}, there exists a passage inside each of these formations, which implies that for each set of four joints there exists a subset of two joints creating a passage.
The actual number of joints needed to ensure the existence of a subset of joints of size $r$ such that passages exist between each pair of joints is given by the Ramsey Number $R(r,4)$. This number is defined as the minimal number of vertices of a graph $G$ such that $G$ either has a complete subgraph of size $r$ or an independent set of size $4$. Since in our case we can never have an independent set of size $4$, we conclude that a subset of size $r$ exists with the claimed property. The Ramsey number $R(r,4)$ is not exactly known, but we can use the upper bound directly extracted from the proof of the Ramsey theorem to arrive at the bound stated above. \cite{grs-rt-90}
\end{proof}

\rephrase{Lemma}{\ref{lemma:2_channels}}{
Consider a set of joints $J= \{j_1,\ldots ,j_k\}$ such that there exists a passage between each pair $(j_i,j_h)$, with $1 \leq i,h \leq k$. Let ${\mathcal P}_1 =\{P \mid P\textrm{ connects } c_i \textrm{ and } c_{\frac{3k}{4}+1-i}$, for $i=1,\dots, \frac{k}{4} \}$ and ${\mathcal P}_2 =\{P \mid P\textrm{ connects } c_{\frac{k}{4}+i} $ and $c_{k+1-i}  \textrm{, for } i=1,\dots, \frac{k}{4} \}$ be two sets of passages between pairs of joints in $J$ (see Fig.~\ref{fig:lemma5}). Then, for at least $\frac{k}{4}$ of the joints of one set of passages, say ${\mathcal P}_1$, there exist paths in \T, starting at the root and containing these joints, which traverse all the doors of ${\mathcal P}_2$ with at least 2 and at most 3 bends. Also, at least half of these joints create an $x$-channel, with $2\leq x \leq 3$.
}

\begin{figure}[ht]
\begin{center}
\includegraphics[height=4cm]{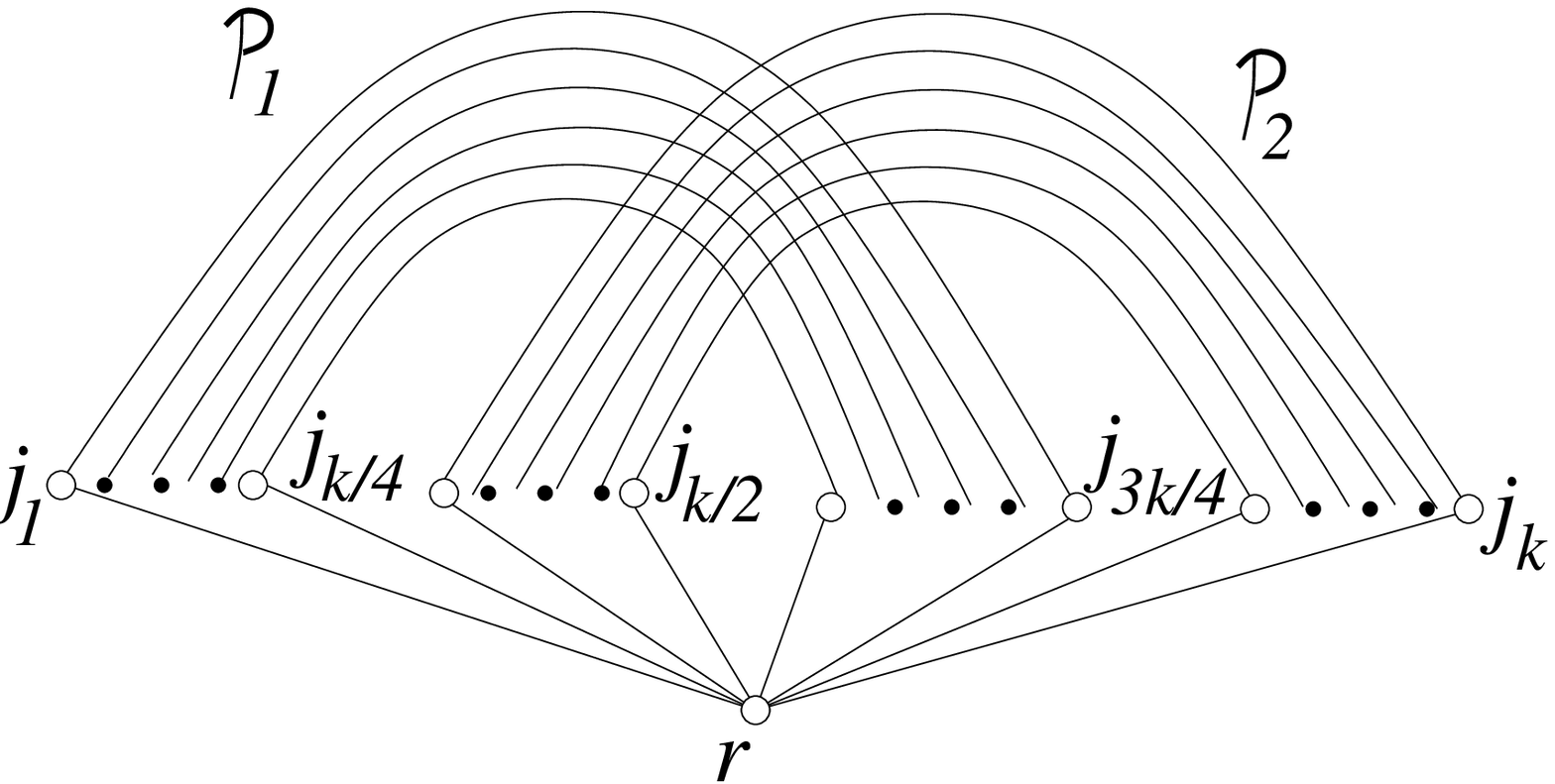}
\caption{The two sets of passages ${\mathcal P}_1$ and ${\mathcal P}_2$ described in Lemma~\ref{lemma:2_channels}.}
\label{fig:lemma5}
\end{center}
\end{figure}

\begin{proof}
Observe first that each passage of ${\mathcal P}_1$ is interconnected with each passage of ${\mathcal P}_2$ and that all the passages of ${\mathcal P}_1$ and all the passages of ${\mathcal P}_2$ are nested.

By Lemma~\ref{lemma:closed_door_in_each_passage} and Property~\ref{prop:three_bends}, for one of ${\mathcal P}_1$ and ${\mathcal P}_2$, say ${\mathcal P}_1$, either for every joint of ${\mathcal P}_1$ between the joints of ${\mathcal P}_2$ in the order around the root or for every joint of ${\mathcal P}_1$ not between the joints of ${\mathcal P}_2$, there exists a path $p_i$ in \T, starting at the root and containing these joints, which has to traverse all the doors of ${\mathcal P}_2$ by making at least $1$ and at most $3$ bends. Also, paths $p_1,\ldots , p_{\frac{k}{4}}$ can be ordered in such a way that a bendpoint of $p_i$ encloses a bendpoint of $p_h$ for each $h>i$.
It follows that there exist $x$-channels with $1\leq x \leq 3$ for each joint.
Consider now the set of joints $J'\subset J$ visited by these paths.
We assume the joints of $J'=\{j'_1,\ldots j'_r \}$ to be in this order around the root.

\begin{figure}[htb]
  \centering{
\begin{tabular}{c c}
\includegraphics[height=4.5cm]{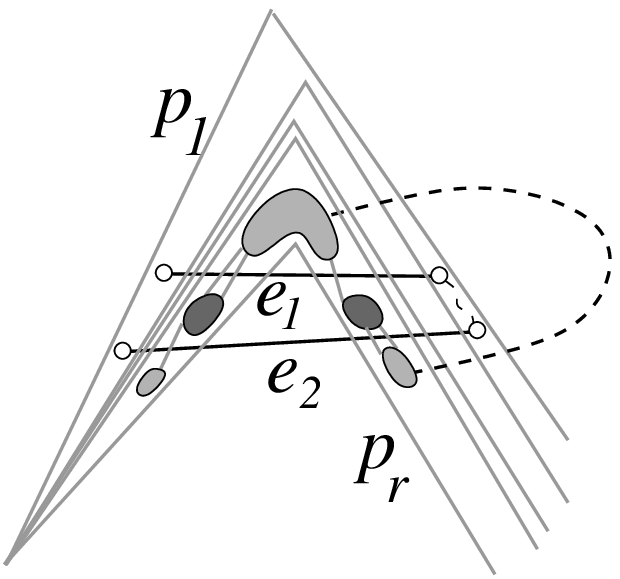} &
\includegraphics[height=4.5cm]{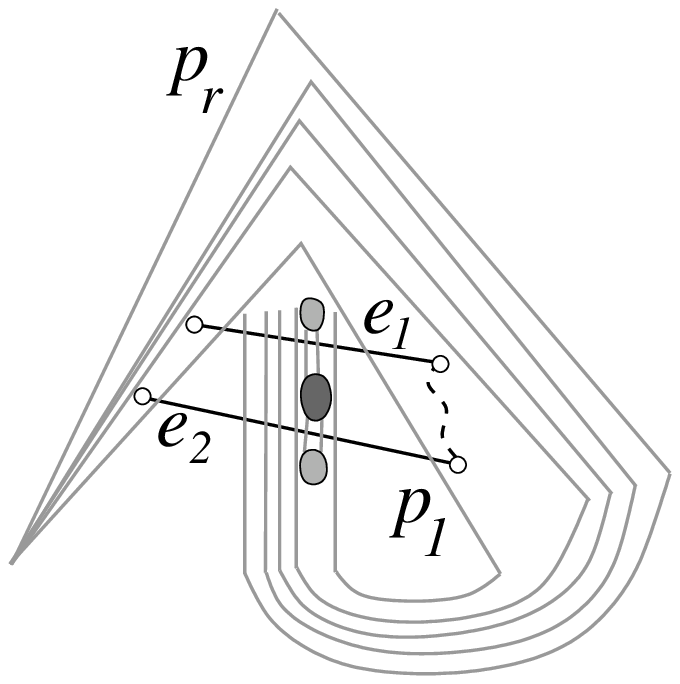} \\
(a) & (b)
\end{tabular}}
\caption{(a) The separating cell $c'$ is in the outermost channel. (b) The separating cell $c'$ is in the innermost channel.}
  \label{fig:lemma5-cases}
\end{figure}

Consider the path $p_1$ whose bendpoint encloses the bendpoint of each of all the other paths and the path $p_r$ whose bendpoint encloses the bendpoint of none of the other paths (see Figs.~\ref{fig:lemma5-cases}(a) and~\ref{fig:lemma5-cases}(b)).
Please note that either $p_1$ visits $j'_1$ and $p_r$ visits $j'_r$ or vice versa, say $p_1$ visits $j'_1$.
By construction, there exists a passage between cells from $j'_1$ and cells from $j'_r$. In this passage there exist either two path-edges $e_1, e_2$ of a cell $c'(1)$ separating two cells $c_1(r),c_2(r)$, thereby crossing the channel of $j'_r$, or two edges of a cell $c'(r)$ separating two cells $c_1(1),c_2(1)$, thereby crossing the channel of $j'_1$. We show that 1-channels are not sufficient to draw these passages.

In the first case (see Fig.~\ref{fig:lemma5-cases}(a)), both separating edges $e_1,e_2$ cross the path $p_r$ before and after the bend, thereby creating blocking cuts separating vertices of the same cell, say $c_1$. Since they are connected by the path, by Property~\ref{prop:blocking-cut}, an additional bend is needed.
In the other case (see Fig.~\ref{fig:lemma5-cases}(b)), any edge connecting vertices of $c'(j'_r)$ is not even crossing any edge of $p_1$ and therefore at least another bend is needed in the channel.
So at least one of the joints needs an additional bend. Since there are passages between each pair of joints in $J'$, all but one joint $j_q$ have a path that has to bend an additional time. We note that the additional bendpoint of each path $p_k$ aside from $p_1$, $p_r$, and $p_q$ has to enclose all the additional bendpoints either of $p_1, \ldots , p_{k-1}$ or of $p_{k+1}, \ldots , p_r$. It follows that, for at least half of the joints, there exist $x$-channels where $2\leq x \leq 3$.
\end{proof}

\rephrase{Lemma}{\ref{n-independent nestings}}{
There exist no $n\geq 2^{22}\cdot 14$ independent sets of formations $S_1,\ldots ,S_n$ inside any extended formation,
where each $S_i$ contains formations of a fixed set of channels of size $r\geq 22$.
}

\begin{proof}
Assume for a contradiction, that such independent sets $S_1,\ldots, S_n$ exist.
By Lemma~\ref{lemma:PS_passage}, each formation in each set will contain a passage and thereby an edge cutting the channel border.
By Property~\ref{prop:CS_1_2} each formation in each set will place an edge to either channel segment $cs_1$ or $cs_2$. As can be easily seen, there exists a set $S^1$ of size $\frac{n}{2}$ of sets of formations that will have at least one common connection for a fixed formation $F_i$ in each set $S_a \subset S^1$, where $1\leq n$.
By repeating the argument we can find a subset $S^2\subset S^1$ of size $\frac{n}{4}$ such that these sets will have at least two common connections for formations $F_i,F_h$ in each set $S_a \subset S^2$. By continuing this procedure we arrive at a subset $S^r$ of size $\frac{n}{2^r}$ that will have at least $r$ common connections. Since all these common connections have to connect to either $cs_1$ or $cs_2$, we have identified a set $S=\{ S'_1, \ldots , S'_{\frac{n}{2^r}}\}$ of size $\frac{n}{2^r}$ of sets of formations of size at least $\frac{r}{2}$ that has all the connections to one specific channel segment $CS$.

We now consider the cutting edges for each of the formations of $S$ in $CS$. Since any of those can intersect the channel border on two different sides, at least half of the connections for a fixed formation $F_{\frac{r}{4}}$ in all the sets will intersect with one side of the channel border, thereby crossing either all the channels $1, \ldots , \frac{r}{4}-1$ or all the channels ${\frac{r}{4}}+1, \ldots , \frac{r}{2}$, assume the first. Consider now the formations $F_{\frac{r}{8}}$ in each of the sets.
These formations of the sets $S'_2, S'_4, \ldots, S'_{\frac{n}{2^{r+1}}}$ will be separated on $CS$ by the edges of the formations $F_{\frac{r}{4}}$ of the sets $S'_3, S'_5, \ldots, S'_{\frac{n}{2^r}-1}$. To avoid a monotonic ordering of the separated formations and thereby the existence of an region-level nonplanar tree these formations $F_{\frac{r}{8}}$ have to place vertices in an adjacent channel segment $CS'$. This will create blocking cuts for either all the channels $1, \ldots , \frac{r}{8}-1$ or all the channels ${\frac{r}{8}}+1, \ldots , \frac{r}{4}$, assume the first.
Consider now the formations $F_1$ in each of the sets. These formations of the sets $S'_3, S'_5, \ldots, S'_{\frac{n}{2^r}-2}$ will be separated on $CS$ by the edges of the formations $F_{\frac{r}{8}}$ of the sets $S'_4, S'_6, \ldots, S'_{\frac{n}{2^r}-3}$.
By the same argument as above also these formations have to place vertices in an adjacent channel segment that are visible from some of the separated areas of $CS$. Since the connection from the formations $F_{\frac{r}{8}}$ are blocking for the connection to $CS'$, the formations $F_1$ have to use the remaining adjacent channel segment $CS''$, thereby blocking all the channels $1, \ldots r_2$. We finally consider the formations $F_2$  of the sets $S'_4, S'_6, \ldots, S'_{10}$. These formations are now separated in $CS$ by the blocking edges to $CS'$ of the formations $F_{\frac{r}{8}}$ and by the blocking edges to $CS''$ of the formations $F_1$. Therefore, these formations cannot use part of any channel segment (tree-)visible to the separated areas in $CS$. So, by Property~\ref{prop:four_sep_areas_PS}, we identified a region-level nonplanar tree, in contradiction to the assumption.
\end{proof}

\rephrase{Lemma}{\ref{lem:nest_sequence}}{
Consider four subsequences $Q_1,\ldots ,Q_4$, where $Q_i = (H_1,H_2,\ldots,H_x)$, of an extended formation $EF$, each consisting of a whole repetition of $EF$. Then, there exists either a pair of nested subsequences or a pair of independent subsequences.
}

\begin{proof}
Assume that no pair of nested subsequences exists. We show that a pair of independent subsequences exists.

First, we consider how $Q_1,\ldots ,Q_4$ use the first two channel segments $cs_1$ and $cs_2$. Each of these subsequences uses either only $cs_1$, only $cs_1$, or both to place its formations. Observe that, if a subsequence uses only $cs_1$ and another one uses only $cs_2$, then such subsequences are clearly independent. So we can assume that all of $Q_1,\ldots ,Q_4$ use a common channel segment, say $cs_2$.

\begin{figure}[ht]
\begin{center}
\includegraphics[height=4cm]{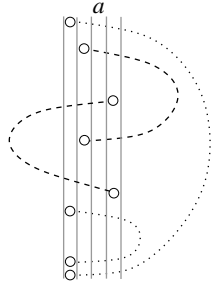}
\caption{If three subsequences use the same channel segment $cs$, then at least two of them are either nesting or separated in $cs$.}
\label{fig:nested-or-independent}
\end{center}
\end{figure}

Then we show that, if three subsequences use the same channel segment $cs$, then at least two of them are separated in $cs$. In fact, if two subsequences using $cs$ are not independent, then they contain formations on the same channel $a$ that intersect with different channel borders of $a$. However, a third subsequence containing a formation that intersects a channel border of $a$ is such that there exists either a nesting or a clear separation between this subsequence and the other subsequence intersecting the same channel border of $a$ (see Fig.~\ref{fig:nested-or-independent}). This fact implies that if three subsequences use only $cs_2$, then at least two of them are independent. From this and from the fact that there are four subsequences using $cs_2$, we derive that two subsequences, say $Q_1,Q_2$, are separated in $cs_2$ and are not separated in $cs_1$. Then, the third subsequence $Q_3$ can be placed in such a way that it is not separated from $Q_1$ and $Q_2$ in $cs_2$. However, this implies that $Q_4$ is separated in $cs_1$ from two of $Q_1,Q_2,Q_3$ and in $cs_2$ from two of $Q_1,Q_2,Q_3$, which implies that $Q_4$ is separated in both channel segments from one of $Q_1,Q_2,Q_3$.
\end{proof}

\rephrase{Lemma}{\ref{lemma:nest_independent}}{
Consider an extended formation $EF(H_1,H_2,\ldots,H_x)$.
Then, there exists a $k$-nesting, where $k \geq 6$, among the formations of $EF$.
}

\begin{proof}
Assume, for a contradiction, that there is no $k$-nesting among the sequence of formations in $EF$. We claim that, under this assumption, there exist more than $n$ sequences of independent formations in $EF$ from the same set of channels $C$, where $n\geq 2^{22}\cdot 14$ and $|C|\geq 22$. By Lemma~\ref{n-independent nestings}, such a claim clearly implies the statement.

Consider sequences that use some common channels in channel segments $cs_1$ and $cs_2$. Then, their separation in $cs_1$ has the opposite ordering with respect to their separation in $cs_2$.

Observe that, by Lemma~\ref{lem:nest_sequence}, there exist at most $(n-1) \cdot 3$ different nestings of subsequences such that there are less than $n$ independent sets of subsequences.
Also note that, if some formations belonging to two different subsequences are nesting, then all the formations of these subsequences have to be part of some nesting. However, this does not necessarily mean for all the formations to nest with each other and to build a single nesting.

Since the number of channels used inside $EF$ is greater than $(n-1) \cdot 3 \cdot 3$, where $n \geq 2^{22}\cdot 14$, we have a nesting consisting of subsequences with at least $3$ different defects.

Let the nesting consist of subsequences $Q^1_1,\ldots ,Q^{r}_1, Q^1_2, \ldots , Q^{r}_2, \ldots, Q^{1}_k \ldots , Q^{r}_k$, where $Q^h_i$ denotes the $h$-th occurrence of a subsequence of $EF$ with a defect at channel $i$. Further, let the path connect them in the order $Q^1_1, Q^1_2, \ldots , Q^1_k,Q^2_1, \ldots ,Q^2_k, \ldots, Q^{r}_k$.
We show that there exists a pair of independent subsequences within this nesting.

\begin{figure}[ht]
\begin{center}
\begin{tabular}{c c c}
\mbox{\includegraphics[height=2.7cm]{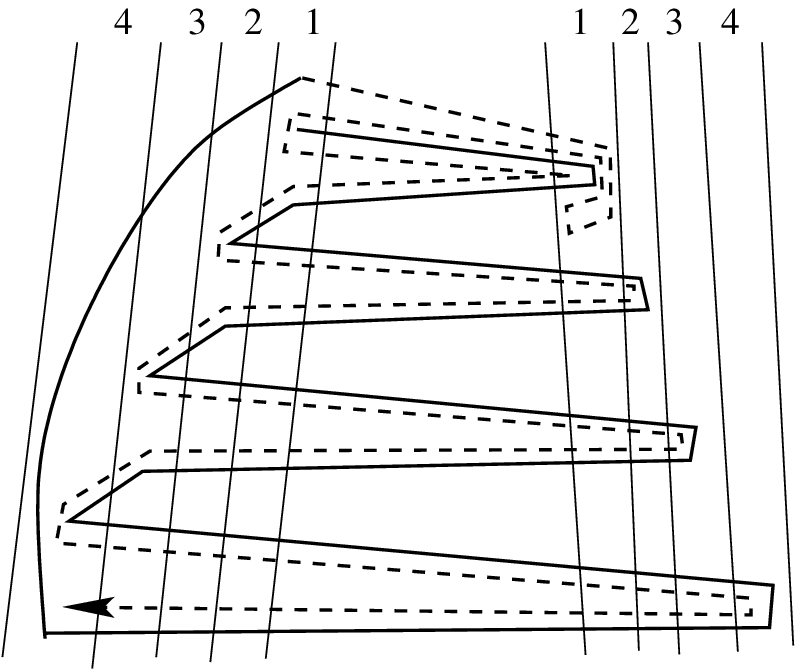}} &
\mbox{\includegraphics[height=2.7cm]{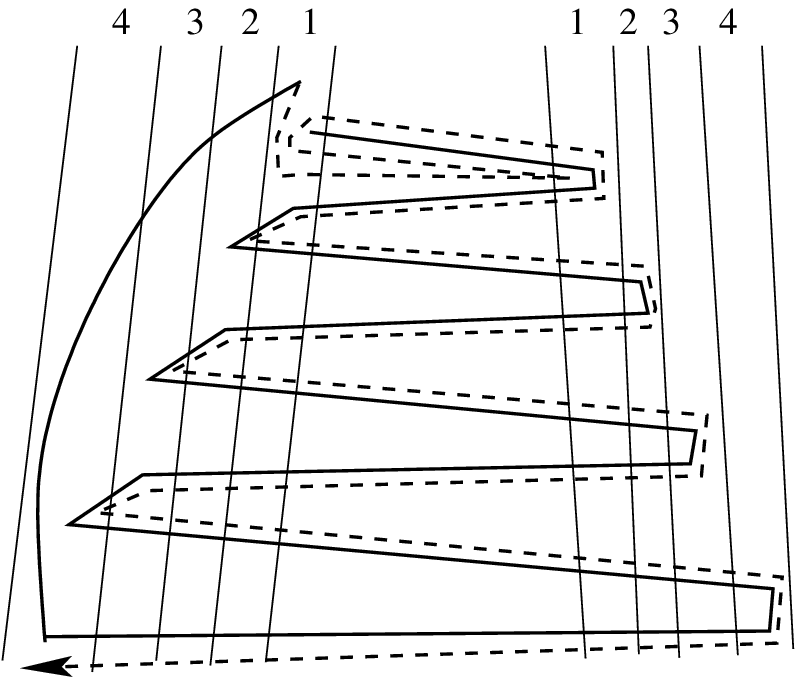}} &
\mbox{\includegraphics[height=2.7cm]{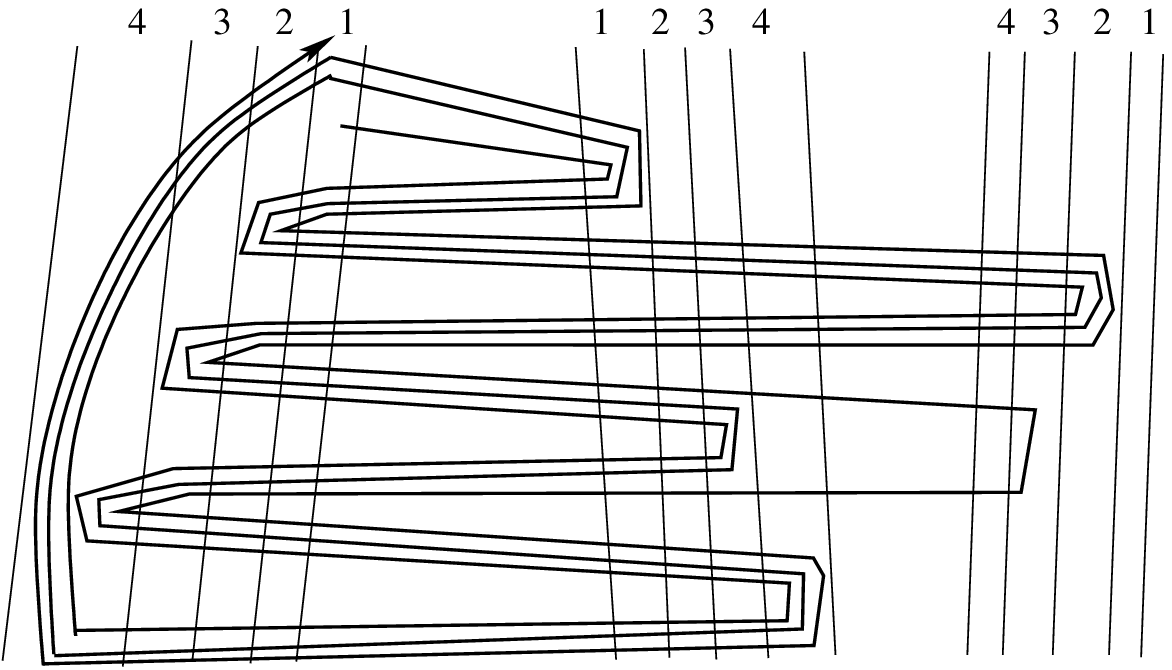}} \\
(a) & (b) & (c)\\
\end{tabular}
\caption{(a) and (b) Possible configurations for $Q^1_1$, $Q^1_2$, and $Q^1_3$. (c) The repetitions follow the outward orientation.}
\label{fig:2repetitions}
\end{center}
\end{figure}

\begin{figure}[ht]
\begin{center}
\begin{tabular}{c c}
\mbox{\includegraphics[height=3cm]{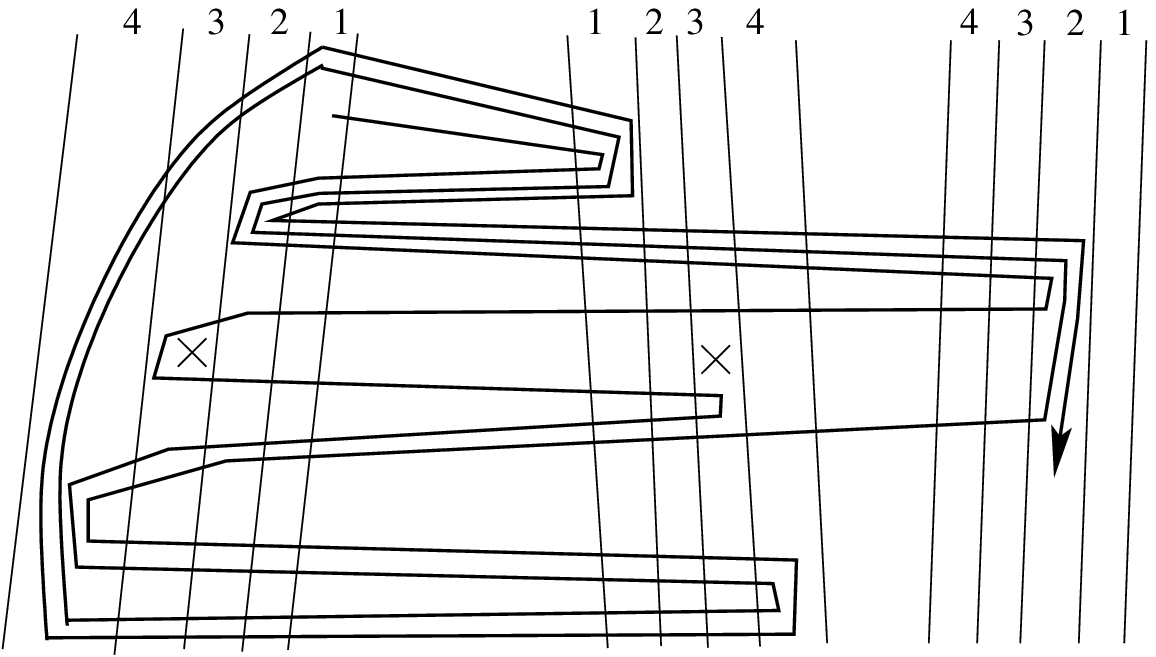}} \hspace{0.2cm} &
\mbox{\includegraphics[height=3cm]{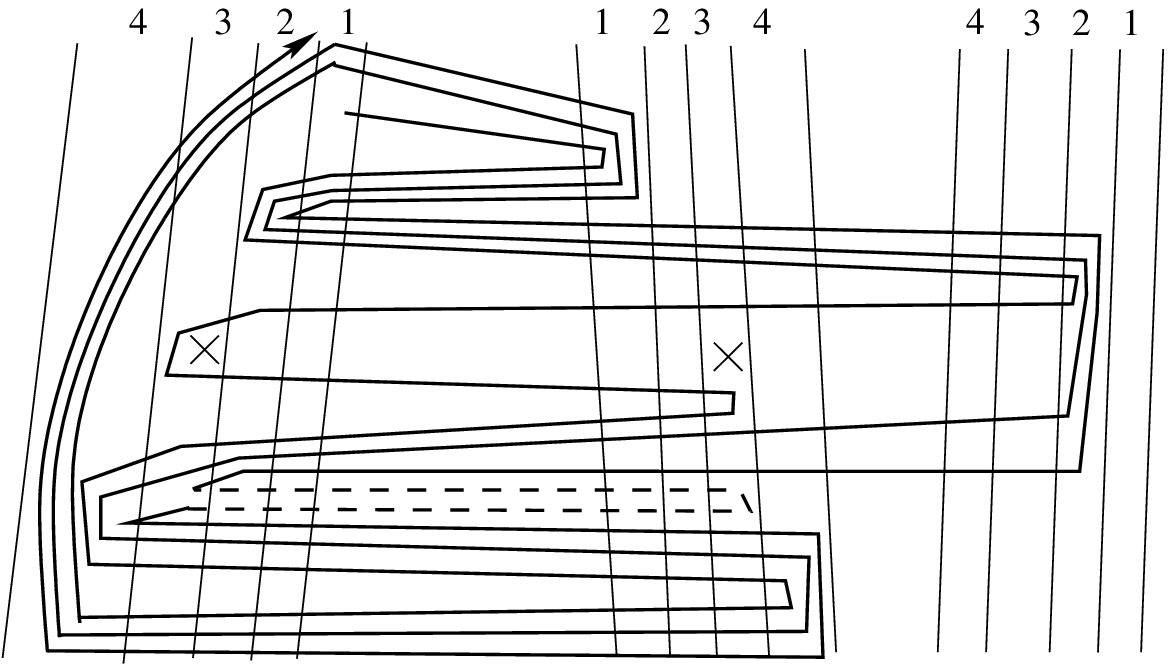}} \\
(a) & (b)\\
\end{tabular}
\caption{The connection between channels $2$ and $4$ blocks visibility for the following repetitions to the part of the channel segment where vertices of channel $3$ were placed till that repetition.}
\label{fig:defects}
\end{center}
\end{figure}

\begin{figure}[ht]
\begin{center}
\begin{tabular}{c}
\mbox{\includegraphics[height=4cm]{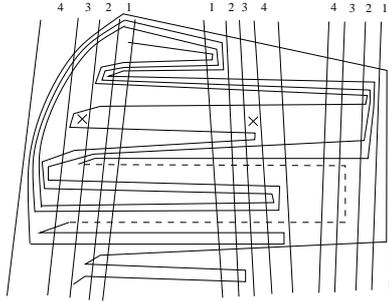}} \\
(a)\\
\end{tabular}
\caption{All the channels $c,\ldots,x$ are shifted and the next repetition starts in a completely different region.}
\label{fig:global-shift}
\end{center}
\end{figure}

Consider now the first two nesting repetitions of sequence $(H_1,H_2,\ldots,H_x)$, that is, $Q^1_1$ and $Q^1_2$. Let the nesting consist of a formation $F(k)$ from $Q^1_1$ nesting in a formation $F'(s)$ from $Q^1_2$.
Consider the edges $e_1,e_2 \in F(k)$ and $e'_1,e'_2 \in F'(s)$ that are responsible for the nesting.
Without loss of generality we assume the path $p$ that connects $e'_2$ and $e_1$ not to contain edges $e'_1,e_2$.
Consider the two parts $a,b$ of the channel border of $s$, where $a$ is between $e_1$ and $e'_1$ and $b$ is between  $e_2$ and $e'_2$. Consider now the closed region delimited by the path through $F'(s)$, the path $p$, the path through $F(k)$, and $b$. Such a region is split into two closed regions $R_{in}$ and $R_{nest}$ by $a$ (see Fig.~\ref{fig:nesting_formations}).

\begin{figure}[htb]
  \centering{
\begin{tabular}{c}
\includegraphics[height=3.5cm]{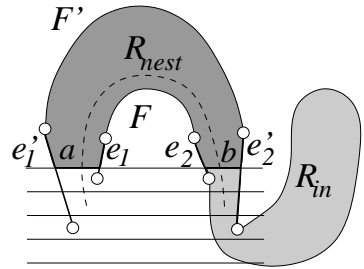}
\end{tabular}}
\caption{Regions $R_{in}$ and $R_{nest}$.}\label{fig:nesting_formations}
\end{figure}

Observe that, in order to reach from $R_{in}$ to the outer region, any path has to cross both $a$ and $b$.
We note that the part of \P starting at $e'_2$ and not containing $F(k)$ is either completely contained in the outer region or has to cross over between $R_{in}$ and the outer region by traversing $R_{nest}$. Similarly, the
the part of \P starting at $e_1$ and not containing $F'(s)$ either does not reach the outer region or has to cross over between $R_{in}$ and the outer region by traversing $R_{nest}$. Furthermore, any formation $F''$ on such a path is also either crossing over and thereby cutting both $a$ and $b$, or not. In the first case $F$ is nested in $F''$ and $F''$ is nested in $F'$.

Consider now the third nesting repetition $Q^1_3$ of sequence $(H_1,H_2,\ldots,H_x)$ (see Figs.~\ref{fig:2repetitions}(a) and~\ref{fig:2repetitions}(b)).
It is easy to see that if $Q^1_3$ is nested between $Q^1_1$ and $Q^1_2$, then there exists a nesting of depth $1$ because $Q^1_3$ contains a defect at a different channel. So we have to consider the cases when the repetitions create the nesting by strictly going either outward or inward. By this we mean that the $i$-th repetition $Q^1_i$ has to be placed such that either $Q^1_i$ is nested inside $Q^1_{i-1}$ (inward) or vice versa (outward).
Without loss of generality, we assume the latter (see Fig.~\ref{fig:2repetitions}(c)).

Consider now a defect in a channel $c$, with $1 < c < k$, at a certain repetition $Q^h_i$. Since the path is moving outward, the connection between channels $c-1$ and $c+1$ blocks visibility for the following repetitions to the part of the channel segment where vertices of channel $c$ were placed till that repetition (see Fig.~\ref{fig:defects}(a) for an example with $c=3$).

A possible placement for the vertices of $c$ in the following repetitions that does not increase the depth of the nesting could be in the same part of the channel segment where vertices of a channel $c'$, with $c' \neq c$, were placed till that repetition. We call \emph{shift} such a move. However, in order to place vertices of $c$ and of $c'$ in the same zone, all the vertices of $c$ belonging to the current cell have to be placed there (see dashed lines in Fig.~\ref{fig:defects}(b), where $c'=c+1$), which implies that a further defect in channel $c$ at one of the following repetitions encloses all the vertices of each of the previously drawn cells, hence separating them with a straight line from the following cells. Hence, also the vertices of $c'$ have to perform a shift to a channel $c''$, with $c \neq c'' \neq c'$. Again, if the vertices of $c'$ and of $c''$ lie in the same zone, we have two cells that are separated by a straight line and hence also the vertices of $c''$ have to perform a shift. By repeating such an argument we conclude that the only possibility for not having vertices of different channels lying in the same zone is to shift all the channels $c,\ldots,x$ and to go back to channel $1$ for starting the following repetition in a completely different region (see Fig.~\ref{fig:global-shift}, where the following repetition is performed completely below the previous one). However, this implies that there exist two repetitions in one configuration that have to be separated by a straight line and therefore are independent, in contradiction to our assumption.
Therefore, we can assume that, after $3\cdot x + 1 $ repetitions, we arrive at a nesting of depth 1. By repeating this argument we arrive after $3\cdot x\cdot 6$ repetitions at the nesting of depth $6$ claimed in the lemma.
\end{proof}

\rephrase{Lemma}{\ref{lemma:k-nesting}}{
If an extended formation lies in a part of the channel that contains only $1-$side connections, then \T and \P do not admit any geometric simultaneous embedding.
}

\begin{proof}
First observe that, by Lemma~\ref{lemma:nest_independent}, there exists a $k$-nesting with $k\geq 6$ in any extended formation $EF$.

\begin{figure}[ht]
\begin{center}
\includegraphics[height=4cm]{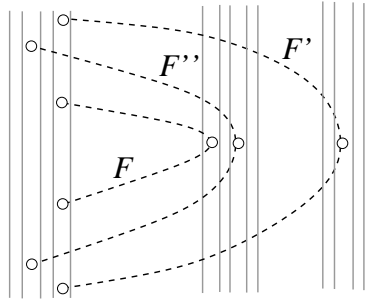}
\caption{Illustration for the case in which only $1$-side connections are possible.}
\label{fig:1-side-nested}
\end{center}
\end{figure}

Consider two nested formations $F,F' \in EF$ belonging to the $k$-nesting. Such formations, by definition, belong to the same channel. Consider now the formation $F'' \in EF$ belonging to a different channel such that $F$ is nested in $F''$ and $F''$ is nested in $F'$. Since each pair of channel segments have a $1-$side connection, we have that $F''$ blocks visibility for $F'$ on the channel segment used by $F$ for the nesting (see Fig.~\ref{fig:1-side-nested}). Hence, $F'$ has to use a different channel segment to perform its nesting, which increases by one the number of used channel segments for each level of nesting. Since the tree supports at most $4$ channel segments, the statement follows.
\end{proof}

\rephrase{Lemma}{\ref{lem:2-different-shapes}}{
Consider two channels $ch_p,ch_q$ with the same intersections. Then, none of channels $ch_i$, where $p<i<q$, have an intersection that is disjoint with the intersections of $ch_p$ and of $ch_q$.
}

\begin{proof}
The statement follows from the fact that the channel borders of $ch_p$ and $ch_q$ delimit the channel for all joints between $p$ and $q$. So, if any channel $ch_i$, with $p<i<q$, had an intersection different from the one of $ch_p$ and $ch_q$, it would either intersect with one of the channel borders of $ch_p$ or $ch_q$ or it would have to bend around one of the channel borders, hence crossing a straight line twice.
\end{proof}

\rephrase{Lemma}{\ref{lem:nesting-bending-area}}{
Consider an $x$-nesting of a sequence of extended formations on an intersection $I_{(a,b)}$, with $a\leq 2$.
Then, there exists a triangle $t$ in the nesting that separates some of the triangles nesting with $t$ from the bending area $b(a,a+1)$ (or $b(a-1,a)$).
}

\begin{proof}
Consider three extended formations $EF_1(H_1), EF_2(H_1), EF_3(H_1)$ lying in a channel $ch_1$ and two extended formations $EF_1(H_2),EF_2(H_2)$ lying in a channel $ch_2$ such that all the channels of the sequence of extended formations are between $ch_1$ and $ch_2$ and there is no formation $F \not\in EF(H_1),EF(H_2)$ nesting between $EF_1(H_1), EF_2(H_1), EF_3(H_1)$ and $EF_1(H_2),EF_2(H_2)$. Suppose, without loss of generality, that the bending point of $ch_1$ is enclosed into the bending point of $ch_2$.

Consider a formation $F_1 \in EF_1(H_1)$ nesting with a formation $F_1' \in EF_1(H_2)$. We have that the connections from $F_1'$ to channel segment $a$ and back has to go around the vertex placed by $F_1$ on channel segment $a$. Therefore, at least one of the connections of $F'_1$ cuts all the channels between $ch_1$ and $ch_2$, that is, all the channels where the sequence of extended formations is placed. Such a connection separates the vertices of $F_1$ from the vertices of a formation $F_2 \in EF_2(H_1)$ on channel segment $a$. Therefore, at least one of the connections of $F_2$ to channel segment $a$ cuts either all the channels in channel segment $a$ or all the channels in channel segment $a+1$ (or $a-1$), hence becoming a blocking cut for such channels. It follows that all the formations nesting inside  $F_2$ on such channels can not place vertices in the bending area $b(a,a+1)$ (or $b(a-1,a)$) outside $F_2$.
\end{proof}

\rephrase{Lemma}{\ref{lem:one_channel_segment}}{
In a situation as described in Proposition~\ref{prop:triangle}, not all the extended formations in a sequence of extended formations can place turning vertices in the same channel segment.
}

\begin{proof}
Assume, for a contradiction, that all the turning vertices are in the same channel segment. Consider a sequence of extended formations $SEF$ and the extended formations in $SEF$ using one of the sets of channels $\{H_1,\ldots ,H_4\}$.

We first show that in $SEF$ there exist some extended formations using connections in $\alpha$ configuration and some using connections in $\beta$ configuration on the channels $\{H_1,\ldots ,H_4\}$. Consider the continuous subsequence of extended formations $EF(H_1),$ $\ldots,EF(H_3)$ in $SEF$. Assume that all the turning vertices of these extended formations are in $\alpha$ configuration. Consider a further subsequence of $SEF$ on the same set of channels with a defect at $H_2$. Then, the connection between $H_1$ and $H_3$ crosses $H_2$, thereby blocking any further $EF(H_2)$ from being in $\alpha$ configuration. Therefore, when considering another subsequence of $SEF$ on the same set of channels which does not contain defects at $H_1,\ldots,H_3$, either the extended formation $EF(H_2)$ is in $\beta$ configuration or it uses another channel segment to place the turning vertex, as stated in the lemma.

So, consider two channels $H_1,H_2$ such that there exists an extended formation $EF(H_1)$ in $\alpha$ configuration and an extended formation $EF(H_2)$ in $\beta$ configuration. Since all the extended formations contain a triangle open on one side that is nested with triangle $t$, we consider five of such triangles, one for each set of channels $H_2,H_3,H_4$ and two for set $H_1$, such that four of the considered extended formations $EF(H_1),\ldots,$ $EF(H_4)$ are continuous in $SEF$ and the other one $EF'(H_1)$ is the first extended formation on the set of channels $H_1$ following $EF(H_4)$ in $SEF$.

\begin{figure}[ht]
\begin{center}
\begin{tabular}{c c}
\mbox{\includegraphics[height=5cm]{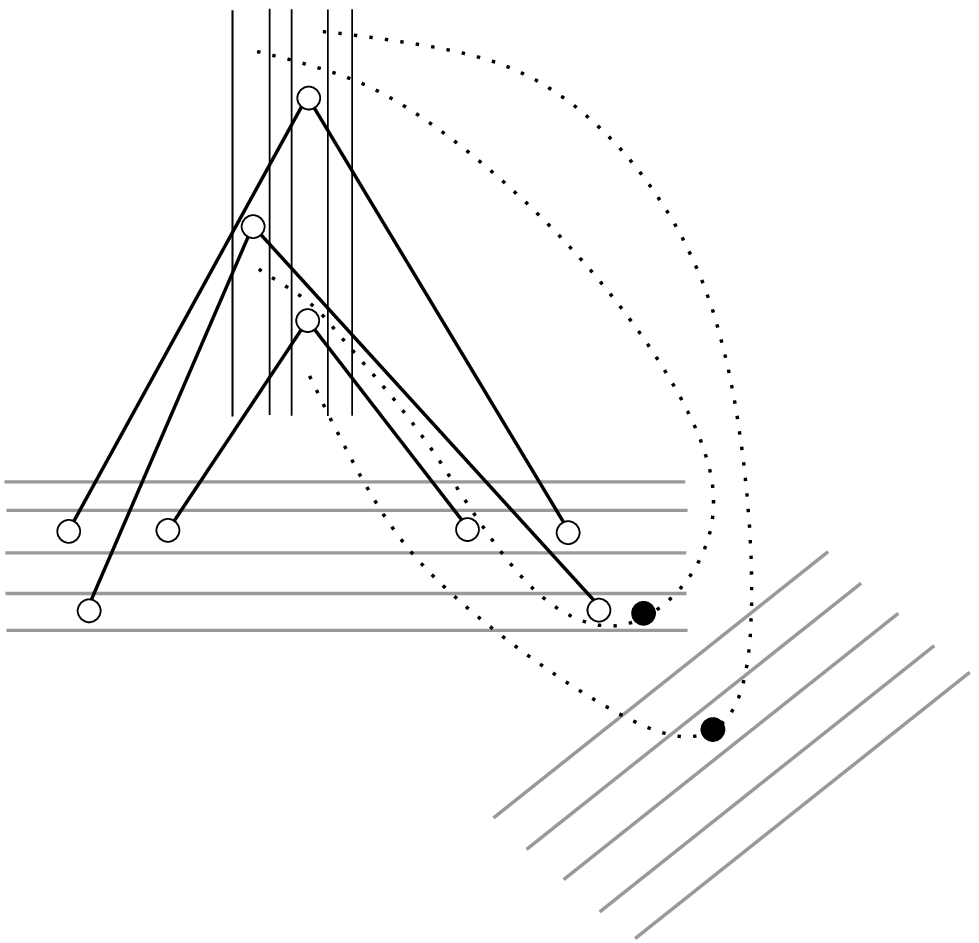}} \hspace{0.1cm} &
\mbox{\includegraphics[height=4.5cm]{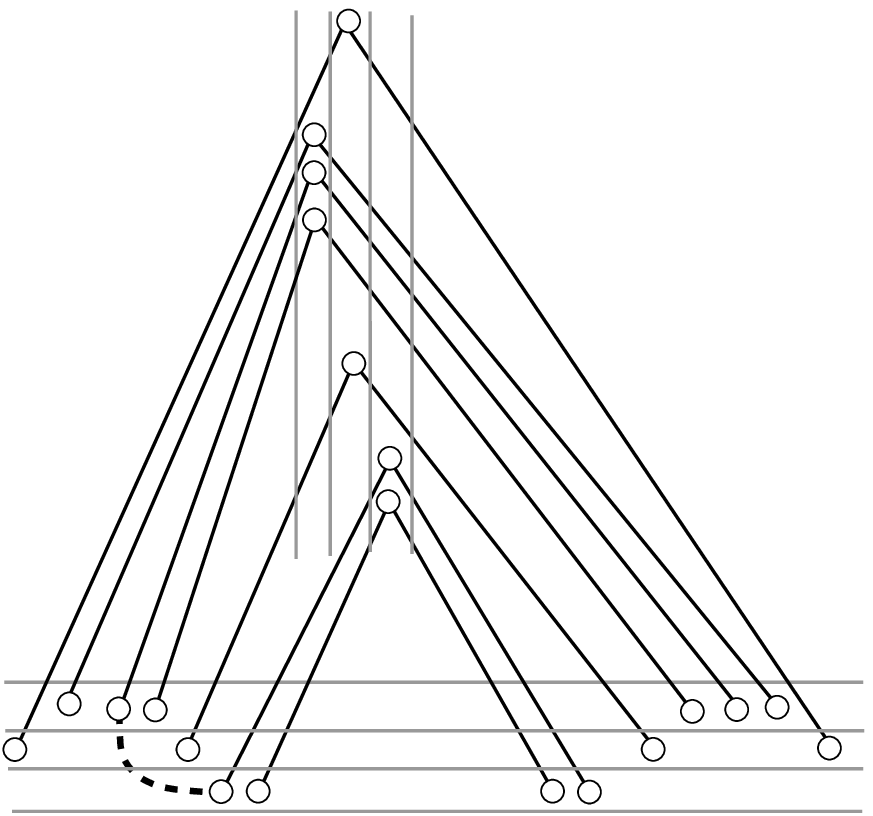}} \\
(a) & (b)\\
\end{tabular}
\caption{(a) Two triangles from the same channel have to use different channel segments if a triangle of another channel is between them. Turning vertices are represented by black circles. (b) When a defect at $H_2$ in encountered, the connection between $EF(H_1)$ and $EF(H_3)$ does not permit the following $EF(H_2)$ to respect the ordering of triangles.}
\label{fig:nested-triangles-prop-3}
\end{center}
\end{figure}

Notice that, if a triangle of an extended formation $EF(H_k)$ is nested in a triangle of an extended formation $EF(H_s)$ and the triangle of $EF(H_s)$ is nested inside a triangle of an extended formation $EF'(H_k)$, with $k<s$, then $EF(H_k)$ has to use a different channel segment to place its turning vertex (see Fig.~\ref{fig:nested-triangles-prop-3}(a)). Hence, the triangles have to be ordered according to the order of the used channels. Also, if the continuous path connecting two triangles $t_1=(u,v,w), t_2=(u',v',w')$ of consecutive extended formations $EF(H_s),EF(H_{s+1})$ connects vertex $u$ to vertex $w'$ (or $u'$ to $w$) via the outer area, then a triangle of $EF(H_1)$ that occurs prior to $EF(H_s)$ and a triangle of $EF'(H_1)$ that occurs after $EF(H_{s+1})$ are nested with the triangle given by the connection of $t_1$ and $t_2$ in an ordering that is different from the order of the channels.

Consider now the following subsequence of $SEF$ having a defect at $H_2$. The connection of $EF(H_1)$ to $EF(H_3)$ in this subsequence blocks access for the following $EF(H_2)$ to the area where it would have to place vertices in order to respect the ordering of triangles (see Fig.~\ref{fig:nested-triangles-prop-3}(b)).
Therefore, after 3 full repetitions of the sequence in $SEF$, at least one extended formation has to use a different channel segment to place its turning vertex.
\end{proof}

\rephrase{Lemma}{\ref{lem:prop3-nodrawing}}{
In a situation as described in Proposition~\ref{prop:triangle}, \T and \P do not admit any geometric simultaneous embedding.
}

\begin{proof}
Consider two extended formations $EF(H_x), EF(H_1)$ that are consecutive in $SEF$.
First note that the connection between $EF(H_x)$ and $EF(H_1)$ cuts all channels $2, \ldots , x-1$ in either channel segment $cs_1$ or $cs_2$.
Since both of these extended formations are also connected to the bending area between channel segments $cs_3$ and $cs_4$, it is not possible for an extended formation $EF(s)$, with $s \in \{2,\ldots ,x-1\}$, to connect from vertices above the connection between $EF(H_x)$ and $EF(H_1)$ to vertices below it by following a path to the bending area.
Note, further, that if all the extended formations $EF(s)$, with $s \in \{2,\ldots ,x-1\}$, are in the channel
segment that is not cut by the connection between $EF(1)$ and $EF(x)$, then a connection is needed from $cs_1$ to $cs_2$ in channel $x$. However, by Lemma~\ref{lem:one_channel_segment}, after three defects in the subsequence of $\{2,\ldots ,x-1\}$ it is no longer possible for some extended formation $EF(s)$, with $s \in \{2,\ldots ,x-1\}$, to place its turning vertex in the same channel segment. Therefore, different channel segments have to be used by the extended formation $EF(s)$, with $s \in \{2,\ldots ,x-1\}$. However, since the path is continuous and since the connection between $EF(H_x)$ and $EF(H_1)$ is repeated after a certain number of steps, we can follow that the path creates a spiral. Also, we note that, in order to respect the order of the sequence, it will be impossible for the path to reverse the direction of the spiral. Hence, once a direction of the spiral has been chosen, either inward or outward, all the connections in the remaining part of the sequence have to follow the same.
This implies that, if a connection between $EF(s)$ and $EF(s+1)$ changes channel segment, that is, it is performed in a different channel segment than the one between $EF(s-1)$ and $EF(s)$, then all the connections of this type have to change. However, when a defect at channel $s+1$ is encountered, also the connection between $EF(s)$ and $EF(s+2)$ has to change channel segment, thereby making impossible for any future connection between $EF(s)$ to $EF(s+1)$ to change channel segment. Therefore, after a whole repetition of the sequence of $SEF$ containing defects at each channel, all the extended formations have to place their turning vertices in the same channel segment, which is not possible, by Lemma~\ref{lem:one_channel_segment}. This concludes the proof that no valid drawing can be achieved in this configuration.
\end{proof}

\rephrase{Lemma}{\ref{lem:intersects_one_three}}{
If a shape contains an intersection $I_{(1,3)}$ and does not contain any other intersection that is disjoint with $I_{(1,3)}$, then \T and \P do not admit any geometric simultaneous embedding.
}

\begin{proof}
First observe that only the intersections $I_{(2,4)}$ and $I_{(1,4)}$ are not disjoint with $I_{(1,3)}$ and could occur at the same time as $I_{(1,3)}$. By Lemma~\ref{lemma:nest_independent}, there exists at least a nesting greater than, or equal to, 6. Each of such nestings has to take place either at intersections $I_{(1,3)}$, $I_{(2,4)}$ or at $I_{(1,4)}$. Remind that, by Property~\ref{prop:CS_1_2}, $1$-vertices can only be placed in $cs_1$ or $cs_2$. Also, the sorting of head vertices to avoid a region-level nonplanar trees can only be done by placing vertices into $cs_3$ or $cs_4$. This implies that the stabilizers have to be placed in $cs_1$ or $cs_2$. Note that the stabilizers also work as $1$-vertices in the tails of other cells. This means that if there exist seven sets of tails that can be separated by straight lines, then there exist a region-level nonplanar tree, by Lemma~\ref{n-independent nestings}.
Observe that, by nesting them according to the sequence, the previous condition would be fulfilled.
This means that we have either a sorting or other nestings. We first show that there exist at most two $x$-nestings with $x \geq 6$. Every $x$-nesting has to take place at either $I_{(1,3)}$, $I_{(2,4)}$ or $I_{(1,4)}$. We assume, w.l.o.g., to have to deal with the greatest possible number of intersections.

Consider the case $I_{(2,4)}^h$ (see Fig~\ref{fig:1324}(a)). Observe that intersections $I_{(1,4)}$ and $I_{(1,3)}$ are either both high or both low and use channel segment $cs_1$. Also, every connection from $cs_1$ to $cs_4$ cuts either $cs_2$ or $cs_3$ and, if one of these connections cuts $cs_2$, then every nesting cutting $cs_1$ closer to $b(1,2)$ has to cut $cs_2$. Hence, we can consider all the connections to $cs_4$ as connections to $cs_2$ or $cs_3$. Also, since any connection cutting a channel segment is more restrictive than a connection inside the same channel segment, such two nestings can be considered as one. Finally, since such a nesting connects to $b_{(2,3)}$, it is not possible to have at the same time a nesting taking place at $I_{(2,4)}^h$. Hence, we conclude that only one nesting is possible in this case.

\begin{figure}[ht]
\begin{center}
\begin{tabular}{c c}
\mbox{\includegraphics[height=4.5cm]{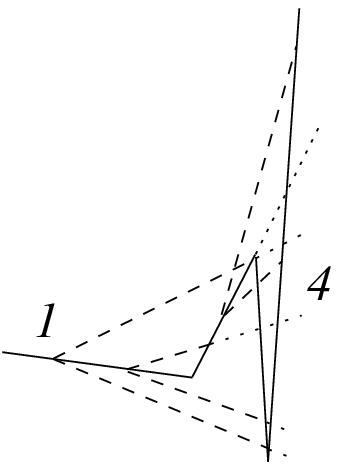}} \hspace{0.1cm} &
\mbox{\includegraphics[height=3cm]{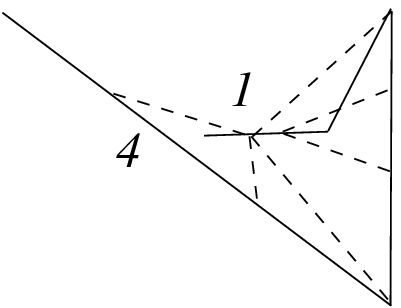}} \\
(a) & (b)\\
\end{tabular}
\caption{(a) Case $I_{(1,3)}$ $I_{(2,4)}^h$. (b) Case $I_{(1,3)}$ $I_{(2,4)}^l$.}
\label{fig:1324}
\end{center}
\end{figure}

Consider the case $I_{(2,4)}^l$ (see Fig~\ref{fig:1324}(b)). Observe that $1$-vertices can be placed at most in $cs_2$ and $2$-vertices can be placed at most in $cs_3$. This means that the extended formations in every nesting have to visit these vertices. Therefore, if there exists both a nesting at $I_{(1,3)}$ and at $I_{(1,4)}$, then the connections to the 1- and 2-vertices in the bending areas $b(2,3)$ and $b(3,4)$ are such that every EF nesting at $I_{(1,4)}$ makes a nesting with the extended formations nesting at $I_{(1,3)}$. Hence, also in this case only one nesting is possible.

So we consider the unique nesting of depth $x\leq 6$ and we show that any way of sorting the nesting formations in the channels will cause separated cells, hence proving the existence of a nonplanar region-level tree.
Consider four consecutive repetitions of the sequence of formations. It is clear that these formations are visiting areas of $cs_1$ and are separated by previously placed formations from other formations on the same channels. This will result in some cells to become separated in $cs_1$. Since, by Property~\ref{prop:four_sep_areas_PS}, the number of monotonically separated cells in $cs_1$ cannot be larger than $3$, for any set of four such separated formations there exists a pair of formations $F_1,F_2$ that change their order in $cs_1$. These connections have to be made on either side of the nesting. If between this pair of formations there is a formation of a different channel, then this formation has to choose the other side to reorder with a formation outside $F_1,F_2$. We further note that, if there are two such connections $F_1,F_4$ and $F_2,F_3$ on the same side that are connecting formations of one channel, nested in the order $F_1,F_2,F_3,F_4$, and another connection on the same side between $F'_1,F'_2$ such that $F'_1$ is nested between $F_1, F_2$ and $F'_2$ between $F_3,F_4$, then this creates a 1-nesting. In the following we show that a nesting of depth at least 6 is reached.

Assume the repetitions of formations in the extended formation to be placed in the order $a,b,c,d,e$.
If this order is not coherent with the order in which the channels appear in the sequence of formations inside the $EF$, then we have already some connections that are closing either side of the nesting for some formations. So we assume them to be in the order given by the sequence. Then, consider a repetition of formations with a defect at some channel $C_i$. We have that there exists a connection closing off at one side all the previously placed formations of $C_i$. However, there are sequences with defects also at channels $C_{i+1}$ and $C_{i-1}$, which can not be realized on the same side as the defects at $C_i$. We generalize this to the fact that all the defects at odd channels are to one side, while the defects at even channels are to the other side.
Since the path is continuous and has to reach from the last formation in a sequence again to the first one, the continuation of the path can only use either the odd or the even defects. This implies that, when considering three further repetitions of formations, the first and the third having a defect at a channel $C_i$ and the second having no defect at $C_i$, there will be a nesting of depth one between these three formations. Since, by Lemma~\ref{lemma:k-nesting}, there cannot be a nesting of depth greater than 5 at this place, we conclude that after 6 repetitions of such a triple of formations there will be at least two formations that are separated from each other. By repeating this argument we arrive after $7\cdot6\cdot2$ repetitions at either the existence of 7 formations that are separated on $cs_1$ and $cs_2$ or at the existence of a nesting of depth 6, both of which will not be drawable without the aid of another intersection that is able to support the second nesting of depth greater than 5.
\end{proof}

\rephrase{Lemma}{\ref{lem:intersects_three_one}}{
If there exists a sequence of extended formation in any shape containing an intersection $I_{(3,1)}$, then \T and \P do not admit any geometric simultaneous embedding.
}

\begin{proof}
Consider a sequence of extended formation in a shape containing an intersection $I_{(3,1)}$. We show that \T and \P do not admit any geometric simultaneous embedding. Observe that there exist several possibilities for channel segment $cs_4$ to be placed. Either there exists no intersection of an elongation of one channel segment with another channel segment or there exists at least one of the intersections $I_{(1,4)}$, $I_{(4,2)}$, $I_{(4,1)}$ or $I_{(2,4)}$.
If there are more than one of such intersections, then it is possible to have several nestings of depth $x$.
We note that, if there exists the intersection $I_{(3,1)}$, then at least one of $cs_1$, $cs_2$, and $cs_4$ are part of the convex hull (see Fig.~\ref{fig:non-convex}).

\begin{figure}[ht]
\begin{center}
\begin{tabular}{c c}
\mbox{\includegraphics[height=3cm]{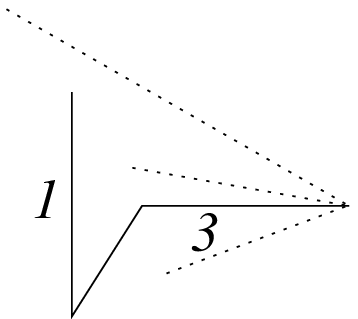}} \hspace{0.1cm} &
\mbox{\includegraphics[height=3cm]{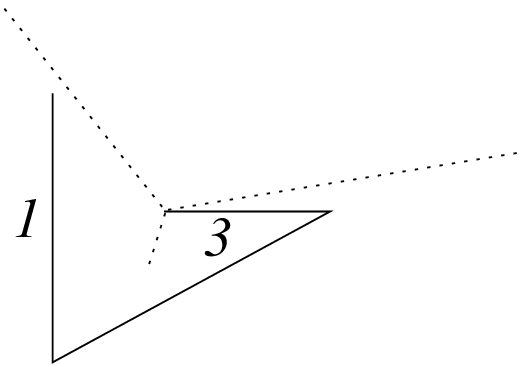}} \\
(a) & (b)\\
\end{tabular}
\caption{If channel segment four is not part of the convex hull then either $cs_1$ or $cs_2$ is part of the convex hull. (a) Case $I_{(1,3)}^l$. (b) Case $I_{(1,3)}^h$.}
\label{fig:non-convex}
\end{center}
\end{figure}

First, we show that there exists a nesting at $I_{(3,1)}$.

Consider case $I^h_{(3,1)}$. We have that $cs_2$ is on the convex hull restricted to the first three channel segments and $cs_4$ can force at most one of $cs_2$ or $cs_1$ out of the convex hull. Hence, one of them is part of the convex hull. We distinguish the two cases.

Suppose $cs_2$ to be part of the convex hull. Assume there exists a nesting at $I_{(1,4)}$. From $cs_4$ the only possible connection without a $1$-side connection is the one to $cs_2$, which, however, is on the convex hull. Hence, an argument analogous to the one used in Lemma~\ref{lem:intersects_one_three} proves that the nesting at $I_{(2,4)}$ has size smaller than $7*12$, which implies that the rest of the nesting has to take place at $I_{(3,1)}$.

Suppose $cs_1$ to be part of the convex hull. Assume that there exists a nesting at $I_{(2,4)}$. Every connection from $cs_4$ has to be either to $cs_1$ or to $cs_2$, by Property~\ref{prop:CS_1_2}. Since $cs_2$ is already part of the nesting, we have connections to $cs_1$. However, $cs_1$ is on the convex hull, hence allowing only $1$-side connections. Therefore, an argument analogous to the one used in Lemma~\ref{lem:intersects_one_three} proves that the nesting at $I_{(2,4)}$ has size smaller than $7*12$, which implies the rest of the nesting has to take place at $I_{(3,1)}$.

Consider case $I^l_{(3,1)}$. Since $cs_2$ is not part of the convex hull, either $cs_1$ or $cs_4$ are. If $cs_1$ is on the convex hull, then the same argument as before holds, while if $cs_4$ is on the convex hull, then no reordering is possible.

Clearly, if there is no intersection other than $I_{(3,1)}$, a nesting in the intersection $I_{(3,1)}$ has to be performed.

Hence, we conclude that a nesting with a depth of $7*12$ in every extended formation has to take place at $I_{(3,1)}$ (or at $I_{(4,1)}$, which can be considered as the same case).

By Lemma~\ref{lem:nesting-bending-area}, the nesting in the bending area is limited. Every extended formation $EF$ which has at least one vertex either in $cs_3$ or in $cs_4$ has a vertex in the bending area. Consider a sequence of extended formations $SEF$ which uses only channels in this particular shape. It's obvious that all of these $EF$ in $SEF$ have to do a nesting at $I_{({3,4},1)}$. Observe that there exist two consecutive edges which are forming a triangle with $cs_1$, $cs_2$, and $cs_3$ by simply placing vertices inside the channel segments. Since every EF creates such triangles, there exists a triangle which is not in the bending area and such that there exists no other triangle
between the bending area and this triangle. This triangle is separating the nesting area from the bending area in all but $s$ extended formations. However, since every EF has to use both of such areas, the inner area of $cs_3$ (or $cs_4$) has to connect to the outer area of $cs_3$ (or $cs_4$). If $cs_1$ is on the convex hull, then there exist only $1$-sided connections, which implies the statement, by Lemma~\ref{lem:prop3-nodrawing}.
On the other hand, if $cs_1$ is not on the convex hull, then there exists $I_{(1,4)}$ and $cs_4$ can be also used to perform connections from the inner to the outer area. However, since $cs_4$ is on the convex hull, such connections are only $1$-side. Hence, by Lemma~\ref{lem:prop3-nodrawing}, the statement follows.
\end{proof}

\rephrase{Lemma}{\ref{lemma:no-ordered-double-cuts}}{
Let $cs_i$ and $cs_{i+1}$ be two consecutive channel segments. If there exists an ordered set $S:=(1,2,\ldots ,5)^3$ of extremal double cuts cutting $cs_i$ and $cs_{i+1}$ such that the order of the intersections of the double cuts with $cs_i$ (with $cs_{i+1}$) is coherent with the order of $S$, then \T and \P do not admit any geometric simultaneous embedding.
}

\begin{proof}
Suppose, for a contradiction, that such a set $S$ exists. Assume first that $cs_i$ and $cs_{i+1}$ are such that the bendpoint of channel $5$ encloses the bendpoint of all the other channels. Hence, any edge creating a double cut at a channel $c$ has to cut all the channels $c'$ with $c'>c$, either in $cs_i$ or in $cs_{i+1}$. Refer to Fig.~\ref{fig:no-ordered-double-cuts}.

Consider the first repetition $(1,2,\ldots ,5)$. Let $e_1$ be an edge creating a double cut at channel $1$. Assume, without loss of generality, that $e_1$ cuts channel segment $cs_i$. Observe that, for channel $1$, the visibility constraints determined in channels $2,\dots,5$ in $cs_i$ and in $cs_{i+1}$ by the double cut created by $e_1$ do not depend on whether it is simple or non-simple. Indeed, by Property~\ref{prop:double-cut}, edge $e_1$ blocks visibility to $b(i,i+1)$ for the part of $cs_i$ where edges creating double cuts at channels $2,\dots,5$ following $e_1$ in $S$ have to place their end-vertices.

\begin{figure}[ht]
\begin{center}
\begin{tabular}{cc}
\mbox{\includegraphics[height=4cm]{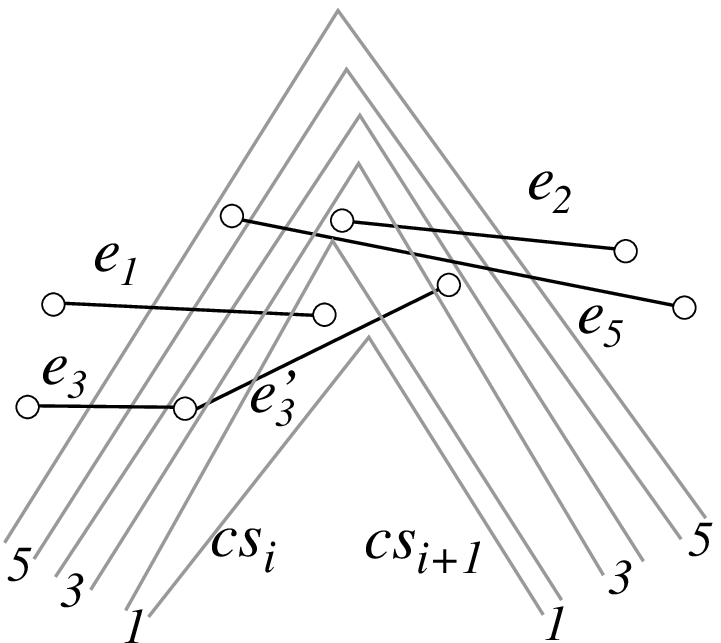}} \hspace{0.1cm} &
\mbox{\includegraphics[height=4cm]{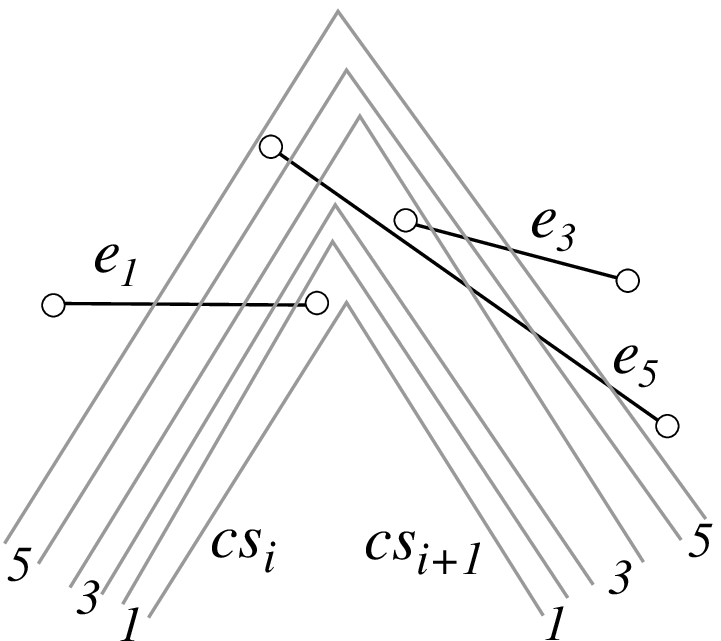}} \\
(a) & (b)\\
\end{tabular}
\caption{Proof of Lemma~\ref{lemma:no-ordered-double-cuts}. (a) $e_3$ cuts $cs_i$. (b) $e_3$ cuts $cs_{i+1}$.}\label{fig:no-ordered-double-cuts}
\end{center}
\end{figure}

Then, consider an edge $e_3$ creating a double cut at channel $3$ in the first repetition of $(1,2,\ldots ,5)$.

If $e_3$ cuts $cs_i$ (see Fig.~\ref{fig:no-ordered-double-cuts}(a)), then it has to create either a non-simple double cut or a simple one. However, in the latter case, an edge $e_3'$ between $cs_i$ and $cs_{i+1}$ in channel $3$, which creates a blocking cut in channel $2$, is needed. Hence, in both cases, channel $2$ is cut both in $cs_i$ and in $cs_{i+1}$, either by $e_3$ or by $e_3'$. It follows that an edge $e_2$ creating a double cut at channel $2$ in the second repetition of $(1,2,\ldots ,5)$ has to cut $cs_{i+1}$, hence blocking visibility to $b(i,i+1)$ for the part of $cs_{i+1}$ where edges creating double cuts at channels $3,\dots,5$ following it in $S$ have to place their end-vertices, by Property~\ref{prop:double-cut}. Further, consider an edge $e_5$ creating a double cut at channel $5$ in the second repetition of $(1,2,\ldots ,5)$. Since visibility to $b(i,i+1)$ is blocked by $e_1$ and $e_3$ in $cs_i$ and by $e_2$ in $cs_{i+1}$, $e_2$ has to create a non-simple double cut (or a simple one plus a blocking cut), hence cutting channel $4$ both in $cs_i$ and in $cs_{i+1}$. It follows that, by Property~\ref{prop:blocking-cut}, an edge $e_4$ creating a double cut at channel $4$ in the third repetition of $(1,2,\ldots ,5)$ can place its end-vertex neither in $cs_i$ nor in $cs_{i+1}$.

If $e_3$ cuts $cs_{i+1}$ (see Fig.~\ref{fig:no-ordered-double-cuts}(b)), then it has to create a simple double cut. Again, by Property~\ref{prop:double-cut}, edge $e_3$ blocks visibility to $b(i,i+1)$ for the part of $cs_{i+1}$ where edges creating double cuts following $e_3$ in $S$ have to place their end-vertices. Hence, an edge $e_5$ creating a double cut at channel $5$ in the first repetition of $(1,2,\ldots ,5)$ cannot create a simple double cut, since its visibility to $b(i,i+1)$ is blocked by $e_1$ in $cs_i$ and by $e_3$ in $cs_{i+1}$. This implies that $e_5$ creates a non-simple double cut (or a simple one plus a blocking cut) at channel $5$, cutting either $cs_i$ or $cs_{i+1}$, hence cutting channel $4$ both in $cs_i$ and in $cs_{i+1}$. It follows that, by Property~\ref{prop:blocking-cut}, an edge $e_4$ creating a double cut at channel $4$ in the second repetition of $(1,2,\ldots ,5)$ can place its end-vertex neither in $cs_i$ nor in $cs_{i+1}$.

The case in which $cs_i$ and $cs_{i+1}$ are such that the bendpoint of $1$ encloses the bendpoint of all the other channels can be proved analogously. Namely, the same argumentation holds with channel $5$ playing the role of channel $1$, channel $1$ playing the role of channel $5$, channel $3$ having the same role as before, channel $4$ playing the role of channel $2$, and channel $2$ playing the role of channel $4$. Observe that, in order to obtain the needed ordering in this setting, $3$ repetitions of $(1,2,\ldots ,5)$ are needed. In fact, we consider channel $5$ in the first repetition, channels $3$ and $4$ in the second one, and channels $1$ and $2$ in the third one.
\end{proof}

\rephrase{Lemma}{\ref{lem:double_cuts_13}}{
Each extended formation in shape $I_{(1,3)}^h$ $I_{(4,2)}^h$ creates double cuts in at least one bending area.
}

\begin{proof}
Refer to Fig.~\ref{fig:13high-41high}(a). Assume, without loss of generality, that the first bendpoint of channel $c_1$ encloses the first bendpoint of all the other channels. This implies that the second and the third bendpoints of channel $c_1$ are enclosed by the second and the third bendpoints of all the other channels, respectively.

\begin{figure}[ht]
\begin{center}
\begin{tabular}{cc}
\mbox{\includegraphics[height=4cm]{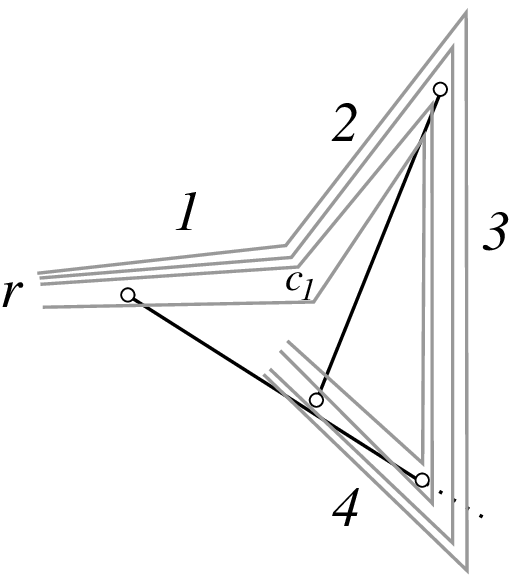}} \hspace{0.1cm} &
\mbox{\includegraphics[height=4cm]{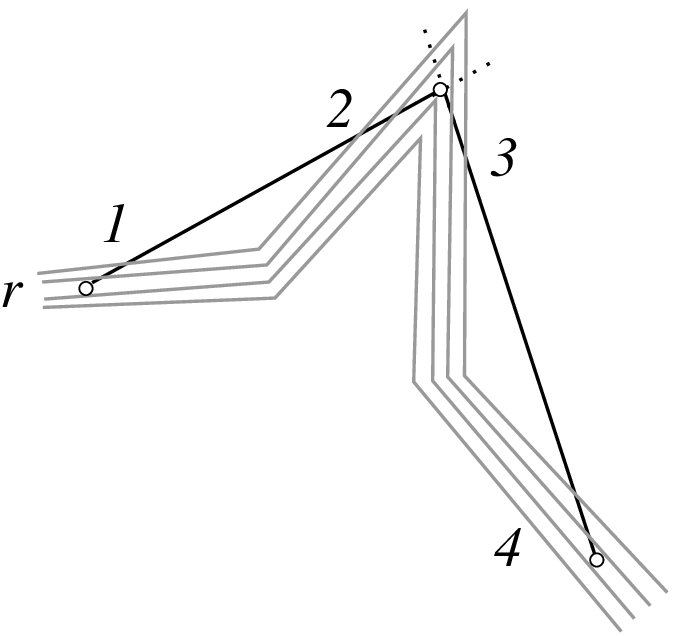}} \\
(a) & (b)\\
\end{tabular}
\caption{(a) Shape $I_{(1,3)}^h$ $I_{(4,\{1,2\})}^h$ has to connect at least one bend with double cuts. (b) Shape $I_{(1,3)}^h$ $I_{(4,2)}^l$ has to connect bend $b(2,3)$ with double cuts.}\label{fig:13high-41high}
\end{center}
\end{figure}

Suppose, for a contradiction, that there exists no double cut in $b(2,3)$ and in $b(3,4)$. Hence, any edge $e$ connecting to $b(2,3)$ (to $b(3,4)$) is such that $e$ and its elongation cut each channel once. Consider an edge connecting to $b(2,3)$ in a channel $c_i$. Such an edge creates a triangle together with channel segments $3$ and $4$ of channel $c_i$ which encloses the bending areas $b(3,4)$ of all the the channels $c_h$ with $h<i$ by cutting such channels twice. Hence, a connection to such a bending area in one of these channels has to be performed from outside the triangle. However, since in shape $I_{(1,3)}^h$ $I_{(4,2)}^h$ both the bending areas $b(2,3)$ and $b(3,4)$ are on the convex hull, this is only possible with a double cut, which contradicts the hypothesis.
\end{proof}

\rephrase{Lemma}{\ref{lem:ordered-set-of-double-cuts-exists}}{
Every sequence of extending formations in shape $I_{(1,3)}^h$ $I_{(4,2)}^{h,l}$ contains an ordered set $(1,2,\ldots ,5)^3$ of extremal double cuts with respect to bending area either $b(2,3)$ or $b(3,4)$.
}

\begin{proof}
Shape $I_{(1,3)}^h$ $I_{(4,2)}^{h}$ is similar to shape $I_{(1,3)}^h$ $I_{(4,1)}^{h}$, depicted in Fig.~\ref{fig:13high-41high}(a), with the only difference on the slope of channel segment $4$, which is such that its elongation crosses channel segment $2$ and not channel segment $1$. Shape $I_{(1,3)}^h$ $I_{(4,2)}^{l}$ is depicted in Fig.~\ref{fig:13high-41high}(b).

Assume, without loss of generality, that the first bendpoint of channel $c_1$ is enclosed by the first bendpoint of all the other channels. This implies that the second bendpoint of channel $c_1$ encloses the second bendpoint of all the other channels.

First observe that bending area $b(2,3)$ is on the convex hull, both in shape $I_{(1,3)}^h$ $I_{(4,2)}^h$ and in shape $I_{(1,3)}^h$ $I_{(4,2)}^l$.

Also, observe that all the extended formations have some vertices in $b(2,3)$ and in $b(3,4)$, and hence all the extended formations have to reach such vertices with path-edges.

In shape $I_{(1,3)}^h$ $I_{(4,2)}^h$, by Lemma~\ref{lem:double_cuts_13}, there exist double cuts either in $b(2,3)$ or in $b(3,4)$, while in shape $I_{(1,3)}^h$ $I_{(4,2)}^l$ there exist double cuts in $b(2,3)$, since the only possible connections to $b(2,3)$ are from channel segments $1$ and $4$, which are both creating double cuts (see Fig.~\ref{fig:13high-41high}(b)). Hence, we consider the extremal double cuts of each extended formation with respect to one of $b(2,3)$ or $b(3,4)$, say $b(2,3)$.

\begin{figure}[htb]
\begin{center}
\begin{tabular}{cc}
\mbox{\includegraphics[height=5.5cm]{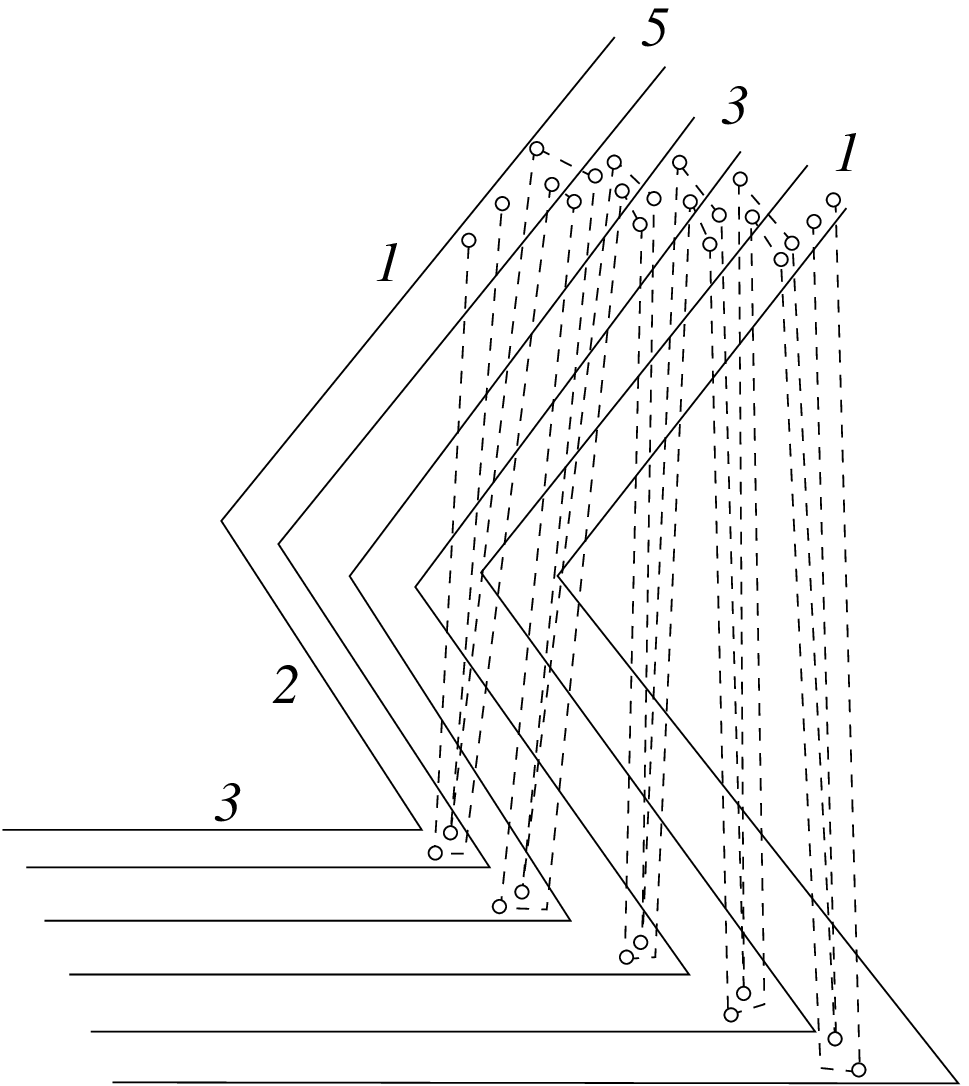}} \hspace{0.1cm} &
\mbox{\includegraphics[height=5.5cm]{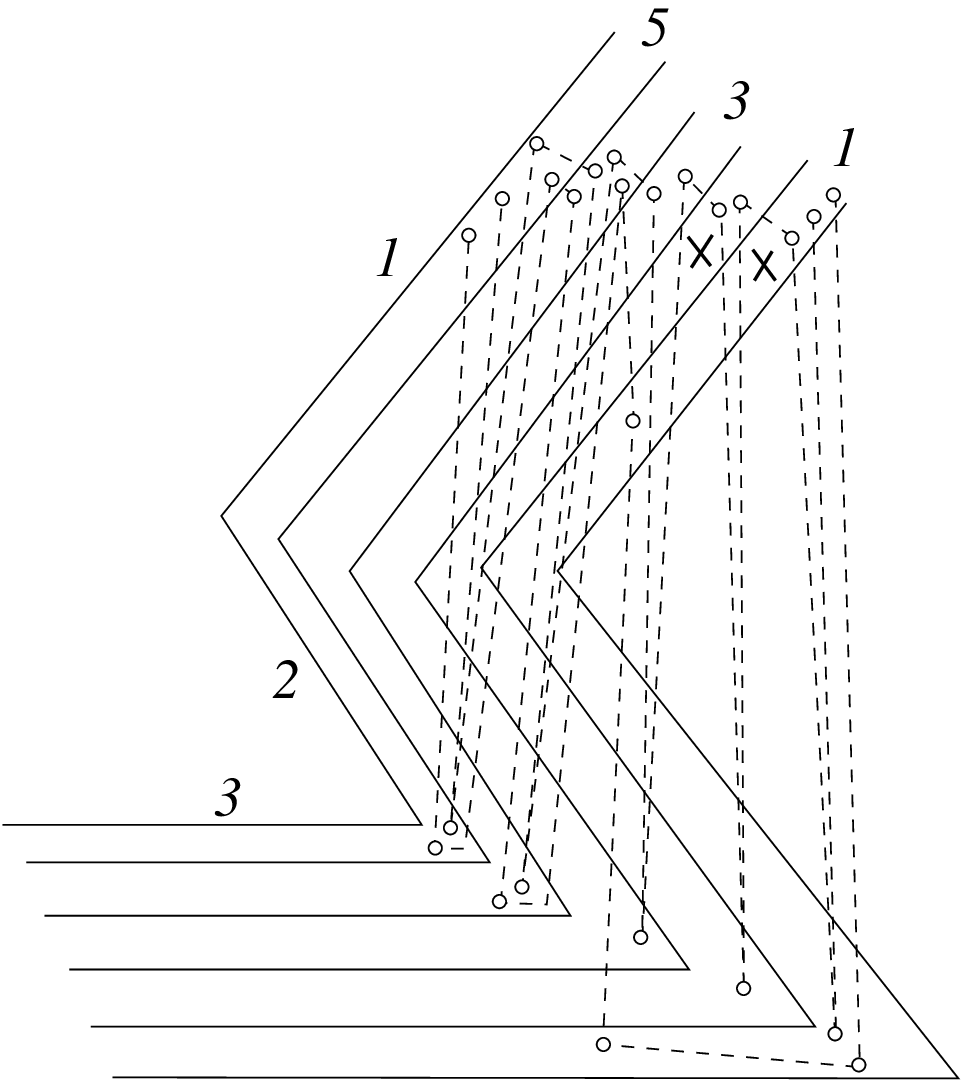}} \\
(a) & (b)\\
\mbox{\includegraphics[height=5.5cm]{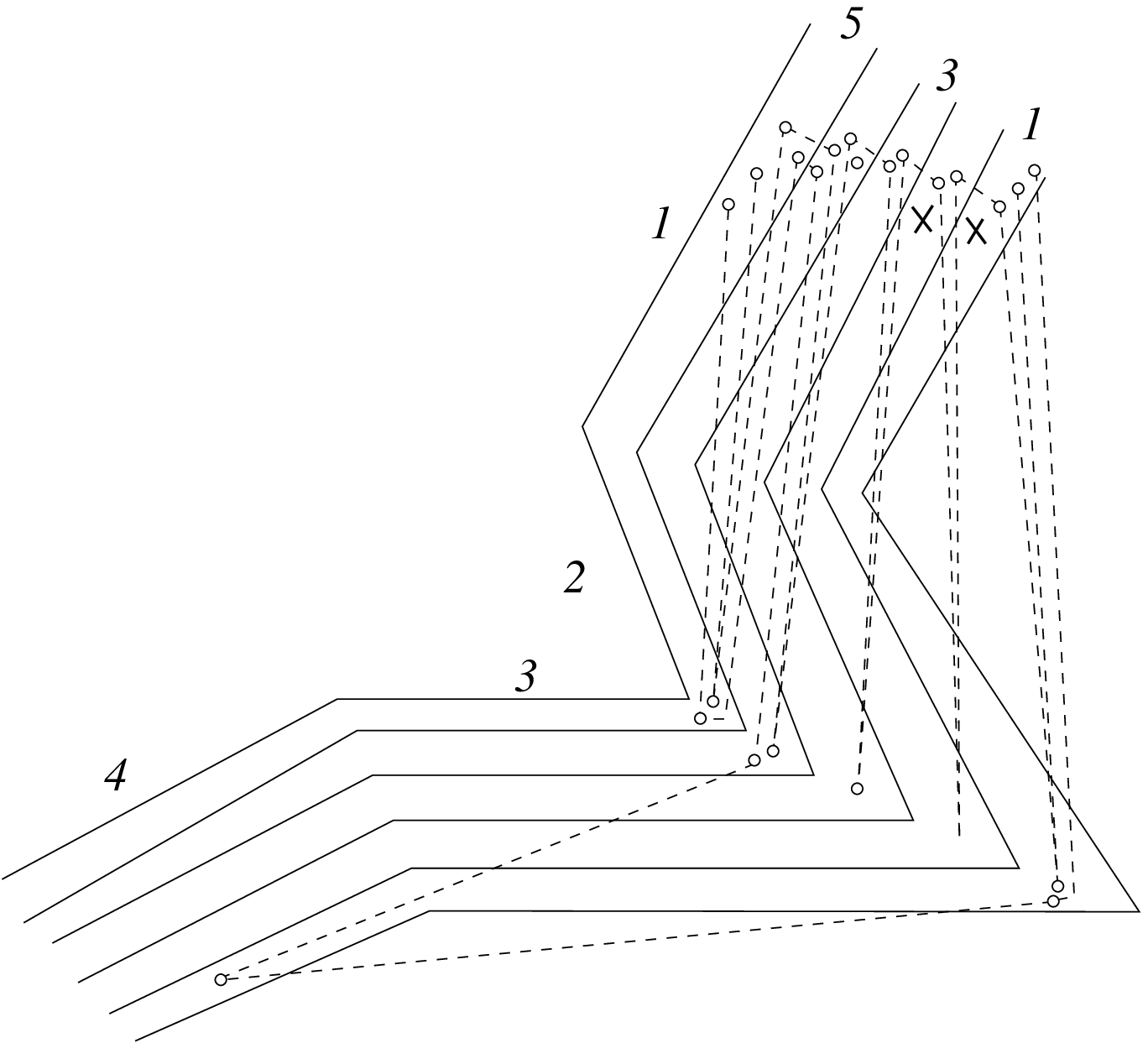}} \hspace{0.1cm} &
\mbox{\includegraphics[height=5.5cm]{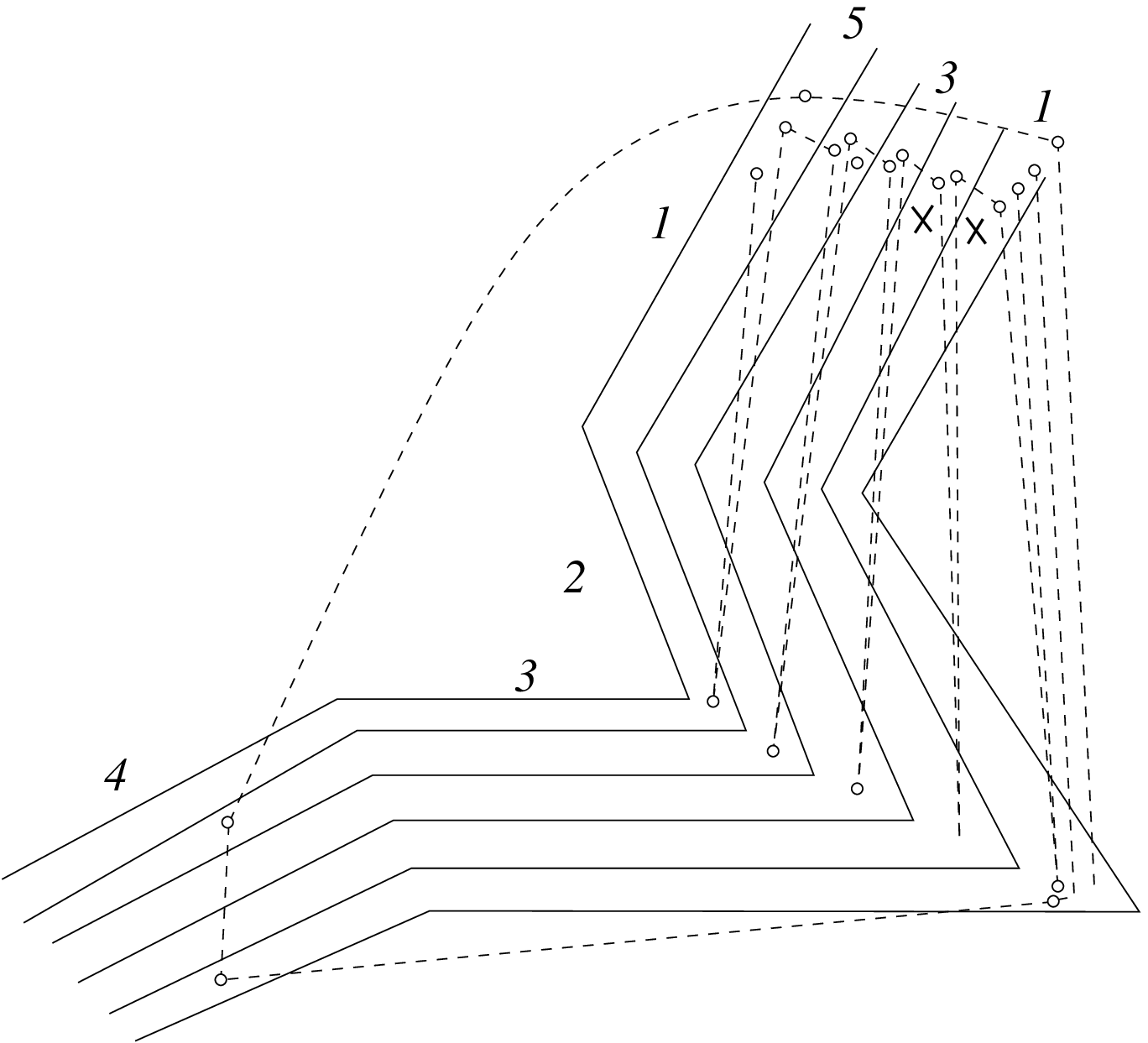}} \\
(c) & (d)\\
\end{tabular}
\caption{(a) The ordering of the extremal double cuts is $(1,1,2,2,\dots,5,5)$. (b) and (c) When a double defect is encountered, the connection between channels $i-1$ and $i+2$ cannot be performed in the same area as the connection between channels $i-1$ and $i$ and between channels $i$ and $i+1$ was performed in the previous repetition: (b) The connection is performed in the same area as the connection between channels $i+1$ and $i+2$ was performed. (c) The connection is performed in channel segment $4$. (d) If channel segment four is used to spiral, the considered double cut was not extremal.}\label{fig:bend23-no-ordered}
\end{center}
\end{figure}

\begin{figure}[htb]
\begin{center}
\begin{tabular}{cc}
\mbox{\includegraphics[height=5.5cm]{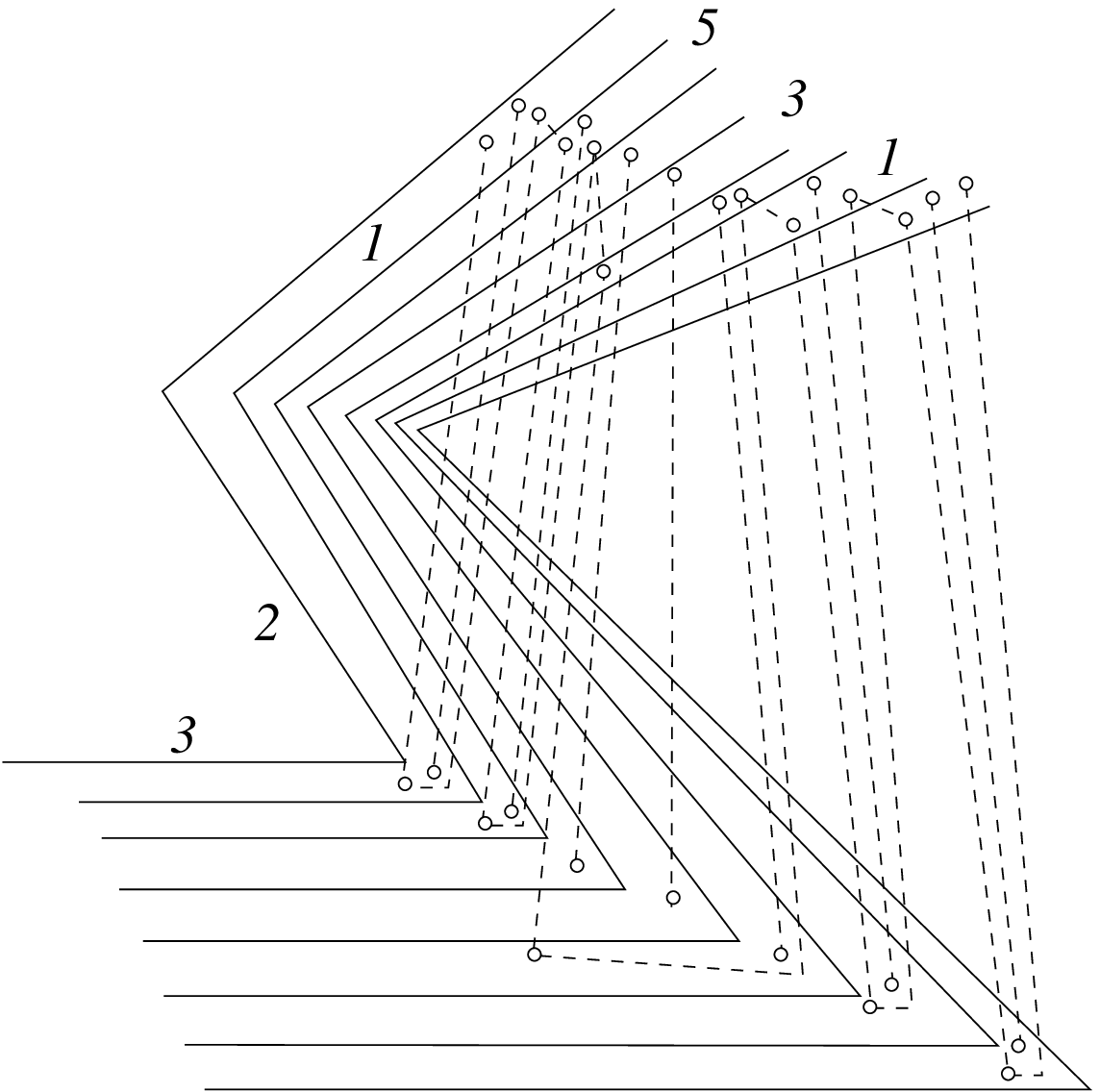}} \hspace{0.1cm} &
\mbox{\includegraphics[height=5.5cm]{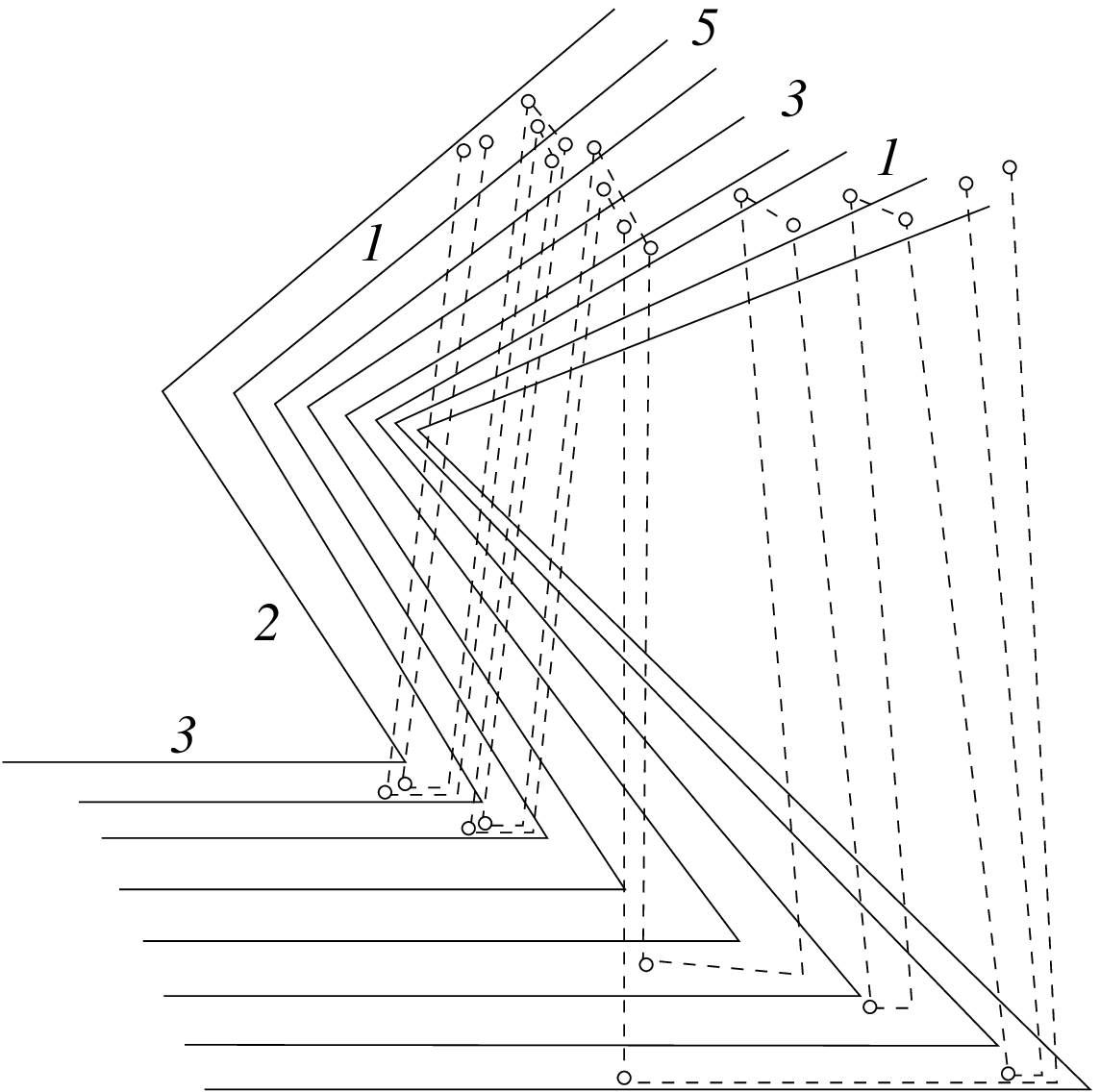}} \\
(a) & (b)\\
\mbox{\includegraphics[height=5.5cm]{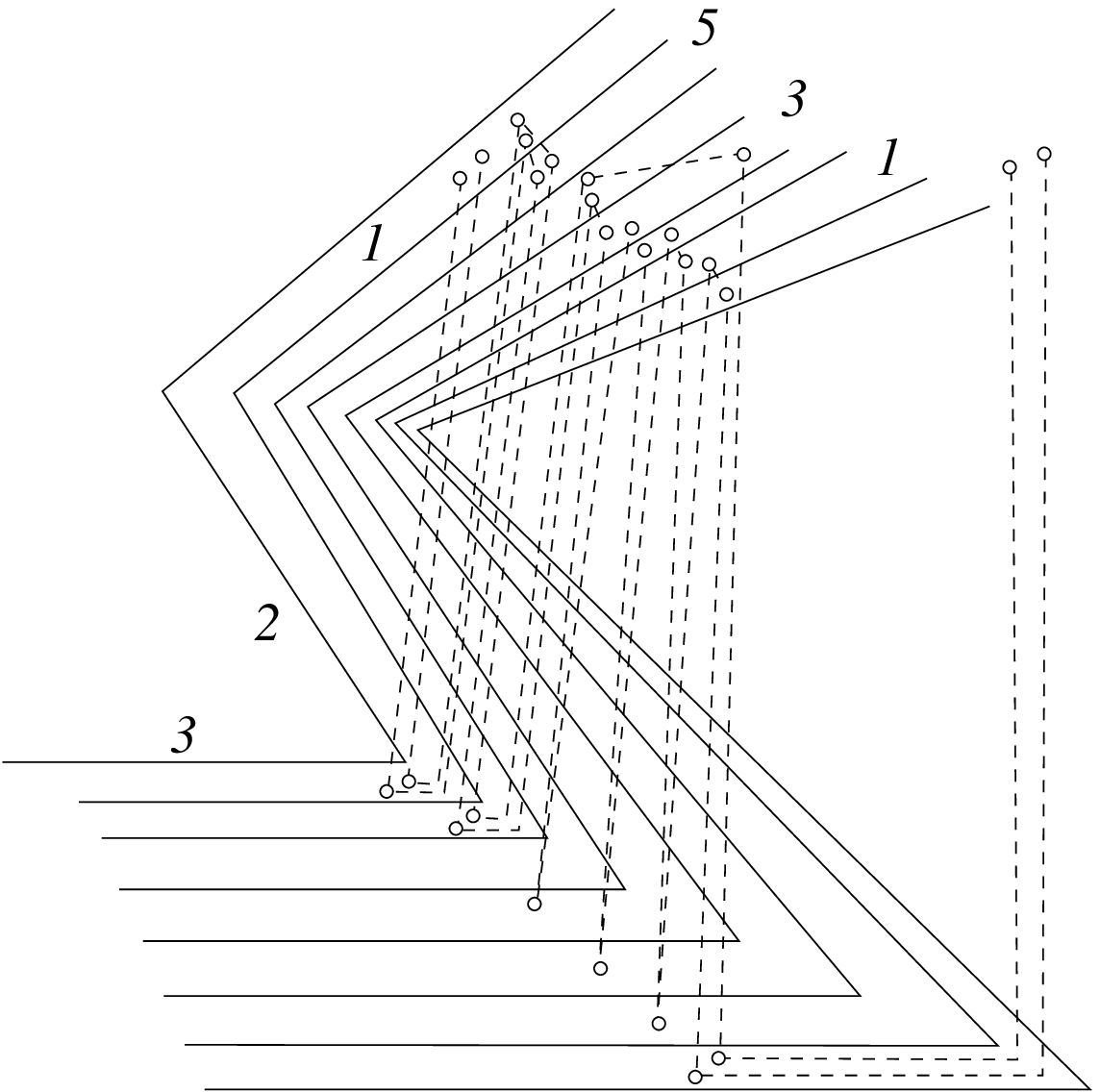}} \hspace{0.1cm} &
\mbox{\includegraphics[height=5.5cm]{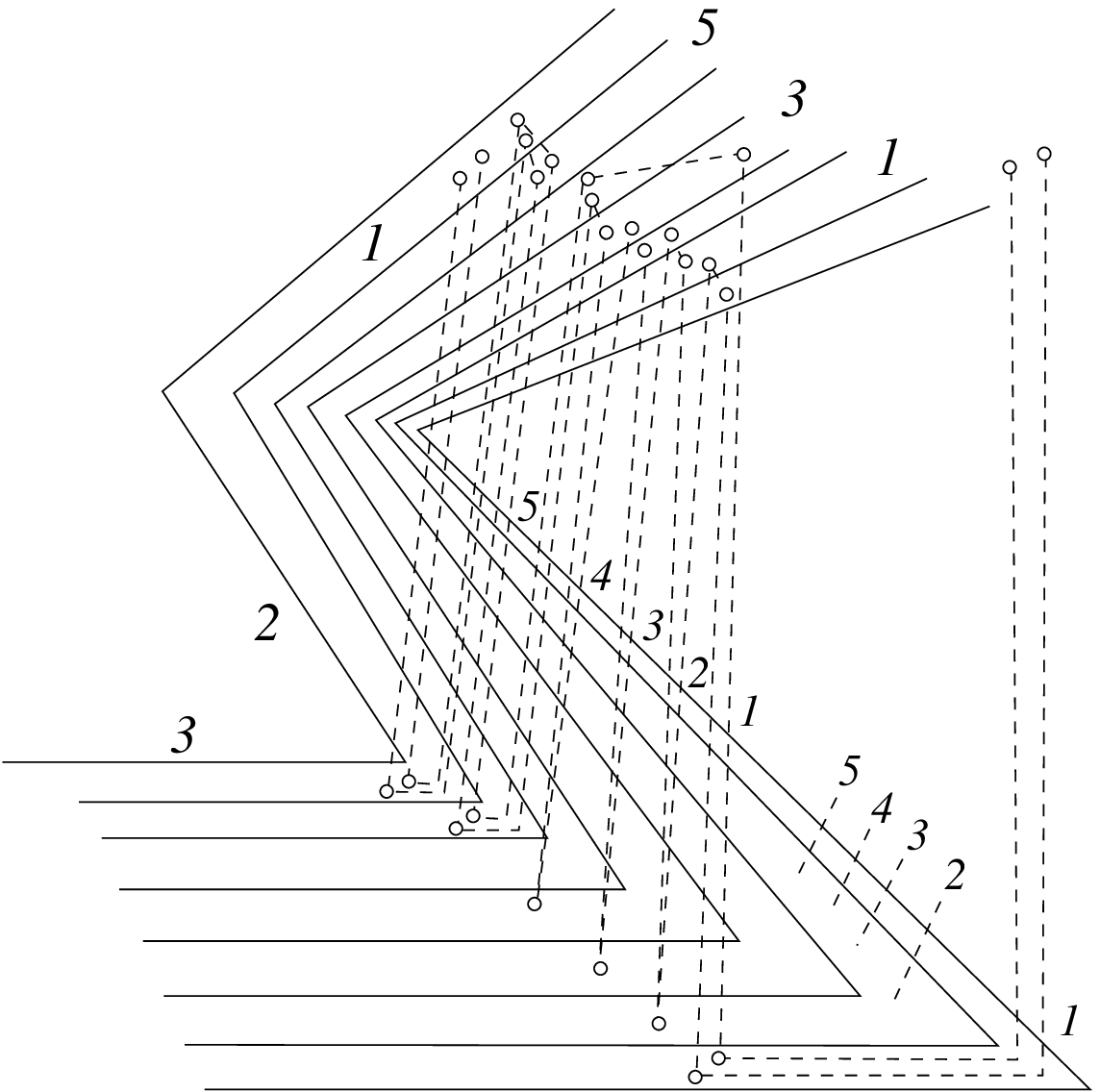}} \\
(c) & (d)\\
\end{tabular}
\caption{(a) A repetition with a double defect in channel $2$ is considered. (b) A repetition with a double defect in channel $0$ is considered. (c) A repetition without any double defect in channels $1,\dots,5$ is considered. (d) An ordered set $(1,\dots,5)$ is obtained.}\label{fig:1342-repetitions}
\end{center}
\end{figure}

Consider two sets of extended formations creating double cuts in $b(2,3)$ at channels $1, \dots, 5$, respectively. Observe that the extended formations in these two sets could be placed in such a way that the ordering of their extremal double cuts is $(1,1,2,2,\dots,5,5)$. The same holds for the following occurrences of extended formations creating double cuts in $b(2,3)$ at channels $1, \dots, 5$, respectively. Clearly, in this way an ordering $(1^n,2^n,\ldots ,5^n)$ could be achieved and hence an ordered set $(1,2,\ldots ,5)^3$ of double cuts would be never obtained (see Fig.~\ref{fig:bend23-no-ordered}(a)).

However, every repetition of extended formations inside a sequence of extended formations contains a double defect at some channel. We show, with an argument similar to the one used in Lemma~\ref{lemma:nest_independent}, that the presence of such double defects determines an ordering $(1,2,\ldots ,5)^3$ of extremal double cuts after a certain number of repetitions of extended formations inside a sequence of extended formations. Namely, consider a double defect at channel $i$ in a certain repetition. The connection between channels $i-1$ and $i+2$ cannot be performed in the same area as the connection between channels $i-1$ and $i$ and between channels $i$ and $i+1$ was performed in the previous repetition. Hence, such a connection has to be performed either in the same area as the connection between channels $i+1$ and $i+2$ was performed (see Fig.~\ref{fig:bend23-no-ordered}(b)), or in channel segment $4$ (this is only possible in shape $I_{(1,3)}^h$ $I_{(4,2)}^l$, see Fig.~\ref{fig:bend23-no-ordered}(c)).
Observe that, going to channel segment $4$ to make the connection, then to channel segment $1$, and finally back to $b(2,3)$, hence creating a spiral, implies that the considered double cut is not extremal (see Fig.~\ref{fig:bend23-no-ordered}(d)). Therefore, the only possibility to consider when channel segment $4$ is used is to make the connection between channels $i-1$ and $i+2$ there and then to come back to $b(2,3)$ with a double cut. Hence, independently on whether channel segment $4$ is used or not, the connection between channels $i-1$ and $i+2$ blocks visibility for the following repetitions to the areas where the connections between some channels were performed in the previous repetition. This implies that the ordering $(1^n,2^n,\ldots ,5^n)$ of extremal double cuts cannot be respected in the following repetitions. In fact, a partial order $(i,i+1,i+2)^2$ is obtained in a repetition of formations creating extremal double cuts at channels $1, \dots, 5$. Also, when two different double defects having a channel in common are considered, the effect of such defects is combined. Namely, consider a double defect at channel $3$ in a certain repetition. The connection between channels $2$ and $5$ blocks visibility to the areas where the connection between $2$ and $3$ and between $3$ and $4$ were performed at the previous repetitions (see Fig.~\ref{fig:1342-repetitions}(a)). Then, consider a double defect at channel $1$ in a following repetition. We have that the connection between channels $0$ and $3$ can not be performed where the connection between $2$ and $3$ was performed in the previous repetitions, since such an area is blocked by the presence of the connection between channels $2$ and $5$. Hence, a double cut at channel $3$ has to be placed after the double cut at channel $5$ created in the previous repetition (see Fig.~\ref{fig:1342-repetitions}(b)). Consider now a further repetition with a defect not involving any of channels $1,\dots,5$. We have that the area where the connection from $1$ to $2$ was performed in the previous repetitions is blocked by the connection between $0$ and $3$ and hence a double cut at channel $1$ has to be placed after the double cut at channel $3$ created in the previous repetition, which, in its turn, was created after the double cut at channel $5$ (see Fig.~\ref{fig:1342-repetitions}(c)).

\clearpage

Also, all the double cuts at channels $2,\dots,5$ have to be placed after the double cut at $1$, and hence a shift of the whole sequence $1,\dots,5$ after the double cut at $5$ is performed and an ordered set $(1,2,\ldots ,5)^2$ is obtained (see Fig.~\ref{fig:1342-repetitions}(d)). Observe that at most two sets of repetitions of extended formation inside a sequence of extended formations such that each set contains a double defect at each channel are needed to obtain such a shift. By repeating such an argument we obtain another shifting of the whole sequence $(1,\dots,5)$, which results in the desired ordered set $(1,2,\ldots ,5)^3$. We have that a set of repetitions of extended formation containing a double defect at each channel is needed to obtain the first sequence $(1,2,\ldots ,5)^2$, then two of such sets are needed to get to $(1,2,\ldots ,5)^2$, and two more are needed to get to $(1,2,\ldots ,5)^3$, which proves the statement.

Observe that, if it were possible to partition the defects into two sets such that there exists no pair of defects involving a common channel inside the same set, then such sets could be independently drawn inside two different areas and the effects of the defects could not be combined to obtain $(1,2,\ldots ,5)^3$. However, since each double defect involves two consecutive channels, at least three sets are needed to obtain a partition with such a property. In that case, however, an ordered set $(1,2,\ldots ,5)^3$ could be obtained by simply considering a repetition of $(1,2,\ldots ,5)$ in each of the sets.
\end{proof}

\rephrase{Lemma}{\ref{lem:cs_two_convex_hull}}{
If channel segment $cs_2$ is part of the convex hull, then \T and \P do not admit any geometric simultaneous embedding.
}

\begin{proof}
First observe that, with an argument analogous to the one used in Lemma~\ref{lem:intersects_one_three}, it is possible to show that there exists a nesting at intersection $I_(4,{1,2})$. Then, by Property~\ref{prop:CS_1_2}, every vertex that is placed in $cs_4$ is connected to two vertices that are placed either in $cs_1$ or in $cs_2$. Hence, the continuous path connecting to a vertex placed in $cs_4$ creates a triangle, having one corner in $cs_4$ and two corners either in $cs_1$ or in its elongation, which cuts $cs_4$ into two parts, the inner and the outer area.

By Lemma~\ref{lem:nesting-bending-area}, not all of these triangles can be placed in the bending area $b(3,4)$. Hence, every extended formation, starting from the second of the sequence, have to place their vertices in both the inner and the outer area of the triangle created by the first one.

Observe that, in order to connect the inner to the outer area, the extended formations can only use $1$-side connections. Namely, $cs_1$ creates a $1$-side connection. Channel segment $cs_2$ is on the convex hull. Since, by Property~\ref{prop:CS_1_2}, every vertex that is placed in $cs_3$ is connected to two vertices that are placed either in $cs_1$ or in $cs_2$, also $cs_3$ creates a $1$-side connection.

From this we conclude that in this configuration the preconditions of Proposition~\ref{prop:triangle} are satisfied, and hence the statement follows.
\end{proof}

\section{An Algorithm for the Geometric Simultaneous Embedding of a Tree of Depth $2$ and a Path}\label{se:depth2}

In this section we describe an algorithm for constructing a geometric simultaneous embedding of any tree \T of depth $2$ and any path \P. Refer to Fig.~\ref{fig:pseudo}.

Start by drawing the root $r$ of $\mathcal T$ on the origin in a coordinate system. Choose a ray $R_1$ emanating from the origin and entering the first quadrant, and a ray $R_2$ emanating from the origin and entering the fourth quadrant.
Consider the wedge $W$ delimited by $R_1$ and $R_2$ and containing the positive $x$-axis. Split $W$ into $t$ wedges $W_1, \ldots ,W_t$, in this clockwise order around the origin, where $t$ is the number of vertices adjacent to $r$ in $\mathcal T$, by emanating $t-2$ equispaced rays from the origin.

Then, consider the two subpaths $\mathcal{P}_1$ and $\mathcal{P}_2$ of $\mathcal{P}$ starting at $r$. Assign an orientation to $\mathcal{P}_1$ and $\mathcal{P}_2$ such that the two edges $(r,u) \in \mathcal{P}_1$ and $(r,v) \in \mathcal{P}_2$ incident to $r$ in $\mathcal{P}$ are exiting $r$.

Finally, consider the $t$ subtrees $\mathcal{T}_1, \ldots, \mathcal{T}_t$ of $\mathcal{T}$ rooted at a node adjacent to $r$, such that $u \in \mathcal{T}_1$ and $v \in \mathcal{T}_t$.

\begin{figure}[h]
\begin{center}
\includegraphics[height=7.5cm]{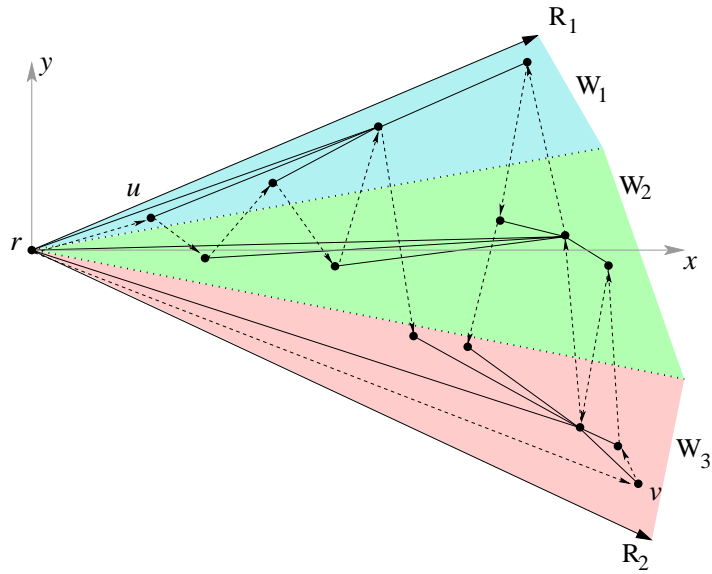}
\caption{A tree with depth two and a path always admit a geometric simultaneous embedding.}
\label{fig:pseudo}
\end{center}
\end{figure}

The vertices of each subtree $\mathcal{T}_i$ are drawn inside wedge $W_i$, in such a way that:
\begin{enumerate}
\item vertex $u$ is the vertex with the lowest $x$-coordinate in the drawing, except for $r$;
\item vertices belonging to $\mathcal{P}_1$ are placed in increasing order of $x$-coordinate according to the orientation of $p_1$;
\item vertex $v$ is the vertex with the highest $x$-coordinate in the drawing;
\item vertices belonging to $\mathcal{P}_2 \setminus r$ are placed in decreasing order of $x$-coordinate according to the orientation of $p_2$, in such a way that the leftmost vertex of $\mathcal{P}_2 \setminus r$ is to the right of the rightmost vertex of $\mathcal{P}_1$; and
\item no vertex is placed below segment $\overline{rv}$.
\end{enumerate}

Since $\mathcal{T}$ has depth $2$, each subtree $\mathcal{T}_i$, with $i=1\ldots,t$, is a star. Hence, it can be drawn inside its own wedge $W_i$ without creating any intersection among tree-edges. Observe that the same holds even for subtree $\mathcal{T}_t$, where the wedge to consider is the part of $W_t$ above segment $\overline{rv}$.

Since $\mathcal{P}_1$ and $\mathcal{P}_2 \setminus \{r\}$ are drawn in monotonic order of $x$-coordinate and are separated from each other, and edge $(r,v)$ connecting such two paths is on the convex hull of the point-set, no intersection among path-edges is created.

From the discussion above, we have the following theorem.

\begin{theorem}\label{th:depth2}
A tree of depth $2$ and a path always admit a geometric simultaneous embedding.
\end{theorem}

\section{Conclusions}\label{se:conclusions}

In this paper we have shown that there exist a tree \T and a path \P on the same set of vertices that do not admit any geometric simultaneous embedding, which means that there exists no set of points in the plane allowing a planar embedding of both \T and \P. We obtained this result by extending the concept of level nonplanar trees~\cite{fk-mlnpt-07} to the one of region-level nonplanar trees. Namely, we showed that there exist trees that do not admit any planar embedding if the vertices are forced to lie inside particularly defined regions. Then, we constructed \T and \P so that the path creates these particular regions and at least one of the many region-level nonplanar trees composing $\T$ has its vertices forced to lie inside them in the desired order. Observe that our result also implies that there exist two edge-disjoint trees that do not admit any geometric simultaneous embedding, which answers an open question posed in~\cite{gkv-ttsids-09}, where the case of two non-edge-disjoint trees was solved.

It is important to note that, even if our counterexample consists of a huge number of vertices, it can also be considered as ``simple'', in the sense that the depth of the tree is just $4$. In this direction, we proved that, if the tree has depth $2$, then it admits a geometric simultaneous embedding with any path. This gives raise to an intriguing open question about whether a tree of depth $3$ and a path always admit a geometric simultaneous embedding or not.

\bibliography{SimTreePathArXiv}
\bibliographystyle{plain}

\end{document}